\PassOptionsToPackage{dvipsnames,table}{xcolor}
\documentclass[acmsmall,screen,nonacm]{acmart}
\makeatletter
\newif\ifchg@lead
\def\chg@after{\futurelet\chg@next\chg@fix}
\def\chg@fix{\ifx\chg@next\@sptoken\else\space\fi}
\DeclareRobustCommand{\chgmark}[2]{%
  \chg@leadfalse
  \ifhmode\ifdim\lastskip>\z@\chg@leadtrue\fi\fi
  {\color{#1}{\ifchg@lead\expandafter\ignorespaces\fi#2}%
   \ifhmode\ifdim\lastskip>\z@\unskip\aftergroup\chg@after\fi\fi}}
\makeatother
\setcopyright{none}

\usepackage{aliascnt}
\usepackage{stmaryrd}
\usepackage{bm}
\usepackage{booktabs}
\usepackage{tikz}
\usetikzlibrary{positioning,arrows.meta,calc}
\usepackage{cleveref}
\crefname{section}{\S}{\S}
\Crefname{section}{\S}{\S}
\crefformat{section}{\S#2#1#3}
\Crefformat{section}{\S#2#1#3}
\crefrangeformat{section}{\S#3#1#4--#5#2#6}
\Crefrangeformat{section}{\S#3#1#4--#5#2#6}
\crefmultiformat{section}{\S\S#2#1#3}{ and~#2#1#3}{, #2#1#3}{, and~#2#1#3}
\Crefmultiformat{section}{\S\S#2#1#3}{ and~#2#1#3}{, #2#1#3}{, and~#2#1#3}
\usepackage{multirow}
\usepackage{array}
\usepackage{amsmath}
\usepackage{amsfonts}
\usepackage{algorithm}
\usepackage{algpseudocode}
\usepackage[framemethod=default]{mdframed}

\theoremstyle{plain}
\newtheorem{theorem}{Theorem}[section]
\newaliascnt{lemma}{theorem}
\newtheorem{lemma}[lemma]{Lemma}
\aliascntresetthe{lemma}
\newaliascnt{proposition}{theorem}
\newtheorem{proposition}[proposition]{Proposition}
\aliascntresetthe{proposition}
\theoremstyle{definition}
\newaliascnt{definition}{theorem}
\newtheorem{definition}[definition]{Definition}
\aliascntresetthe{definition}
\newaliascnt{example}{theorem}
\newtheorem{example}[example]{Example}
\aliascntresetthe{example}

\crefname{theorem}{theorem}{theorems}
\Crefname{theorem}{Theorem}{Theorems}
\crefname{lemma}{lemma}{lemmas}
\Crefname{lemma}{Lemma}{Lemmas}
\crefname{proposition}{proposition}{propositions}
\Crefname{proposition}{Proposition}{Propositions}
\crefname{definition}{definition}{definitions}
\Crefname{definition}{Definition}{Definitions}
\crefname{example}{example}{examples}
\Crefname{example}{Example}{Examples}
\crefname{algorithm}{algorithm}{algorithms}
\Crefname{algorithm}{Algorithm}{Algorithms}

\usepackage{colortbl}
\usepackage{amsmath}
\usepackage{graphicx}
\usepackage{booktabs}
\usepackage{array}		
\usepackage{multirow}
\usepackage{subcaption}
\usepackage{algorithm}
\usepackage{algpseudocode}
\usepackage{paralist}   
\usepackage{xspace}
\usepackage{adjustbox}
\usepackage{balance}
\usepackage{rotating}   
\usepackage{listings}
\usepackage{pifont}
\usepackage{framed}
\usepackage[shortlabels]{enumitem}
\setlist{nosep,leftmargin=\parindent}
\usepackage{microtype}
\usepackage{amsmath, amsthm}

\makeatletter
\newcommand{\ostar}{\mathbin{\mathpalette\make@circled\star}}
\newcommand{\make@circled}[2]{%
  \ooalign{$\m@th#1\smallbigcirc{#1}$\cr\hidewidth$\m@th#1#2$\hidewidth\cr}%
}
\newcommand{\smallbigcirc}[1]{%
  \vcenter{\hbox{\scalebox{0.77778}{$\m@th#1\bigcirc$}}}%
}
\makeatother

\definecolor{mygreen}{rgb}{0,0.6,0}
\definecolor{mygray}{rgb}{0.5,0.5,0.5}
\definecolor{mymauve}{rgb}{0.58,0,0.82}
\definecolor{darkblue}{rgb}{0.0,0.0,0.6}
\definecolor{maroon}{RGB}{102, 0, 0}
\definecolor{Maroon}{cmyk}{0,0.87,0.68,0.32}
\definecolor{darkred}{RGB}{139, 0, 0}
\definecolor{forestgreen}{RGB}{34, 139, 34}

\lstdefinelanguage{json}{
  basicstyle=\ttfamily\small,
  numbers=left,
  numberstyle=\tiny,
  stepnumber=1,
  showstringspaces=false,
  breaklines=true,
  frame=single,
  xleftmargin=1.8em,        
  numbersep=6pt, 
  literate=
   *{0}{{{\color{blue}0}}}{1}
    {1}{{{\color{blue}1}}}{1}
    {2}{{{\color{blue}2}}}{1}
    {3}{{{\color{blue}3}}}{1}
    {4}{{{\color{blue}4}}}{1}
    {5}{{{\color{blue}5}}}{1}
    {6}{{{\color{blue}6}}}{1}
    {7}{{{\color{blue}7}}}{1}
    {8}{{{\color{blue}8}}}{1}
    {9}{{{\color{blue}9}}}{1}
    {:}{{{\color{red}:}}}{1}
    {,}{{{\color{red},}}}{1}
    {\{}{{{\color{red}{\{}}}}{1}
    {\}}{{{\color{red}{\}}}}}{1}
    {[}{{{\color{red}[}}}{1}
    {]}{{{\color{red}]}}}{1},
}

\newcommand{\tool}{\textsc{Dickens}\xspace}
\newcommand{\sem}[1]{\llbracket #1\rrbracket}
\newcommand{\E}{\mathbb{E}}
\newcommand{\N}{\mathbb{N}}
\newcommand{\dW}{d_{W}}
\newcommand{\Geom}{\mathrm{Geom}}
\newcommand{\Atom}{\mathsf{Atom}}
\newcommand{\Retry}{\mathsf{Retry}}
\newcommand{\RepeatOp}{\mathsf{Repeat}}
\newcommand{\Mix}{\mathsf{Mix}}
\newcommand{\pre}{\mathrm{pre}}
\newcommand{\Down}{\Downarrow}

\newcommand{\mypara}[1]{\textit{#1}}

\AtBeginDocument{\setlength{\textfloatsep}{12pt plus 2pt minus 4pt}\setlength{\intextsep}{10pt plus 2pt minus 2pt}}

\usepackage{multibib}
\makeatletter
\gdef\dk@appcites{}
\def\dk@appbibcite#1#2{\@ifundefined{b@#1}{\global\@namedef{b@#1}{#2}}{}}
\let\dk@input@\@input@
\def\@input@#1{\begingroup\def\bibcite##1##2{\g@addto@macro\dk@appcites{\dk@appbibcite{##1}{##2}}}\dk@input@{#1}\endgroup}
\newcites{App}{References for the Appendices}
\let\@input@\dk@input@
\AtBeginDocument{\dk@appcites}
\makeatother
\newcommand{\restatedname}{}
\newtheorem*{restatedinner}{\restatedname}
\newenvironment{restated}[2]{\renewcommand{\restatedname}{#1}\begin{restatedinner}[#2]}{\end{restatedinner}}
\makeatletter
\let\dk@newlabel\newlabel
\def\newlabel#1#2{\dk@newlabel{#1}{#2}\dk@newlabel{main:#1}{#2}}
\newcommand{\dkTocSupplement}{\setcounter{tocdepth}{2}}
\newcommand{\dkSupplementStart}{%
  \clearpage
  \addtocontents{toc}{\protect\dkTocSupplement}%
  \let\dk@cite\cite
  \def\cite{\let\@citex\mb@@citex\let\@newciteauxhandle\@auxoutApp\dk@cite}%
  \counterwithin{equation}{section}\counterwithin{table}{section}%
  \counterwithin{figure}{section}\counterwithin{algorithm}{section}%
  {\noindent\raggedright\LARGE\sffamily\bfseries Appendices for ``The Best of Times, the Worst of Times: Moment-Based Analysis of Probabilistic Cost Structures''\par}%
  \bigskip}
\newcommand{\supplementcontents}{{\setcounter{tocdepth}{-10}\renewcommand{\addcontentsline}[3]{}\tableofcontents}\setcounter{tocdepth}{2}}
\makeatother

\begin{document}

\title[Moment-Based Analysis of Probabilistic Cost Structures]{The Best of Times, the Worst of Times:}
\subtitle{Moment-Based Analysis of Probabilistic Cost Structures}

\author{Chenyu Zhou}
\orcid{0009-0006-8493-6886}
\affiliation{%
  \institution{University of Southern California}
  \city{Los Angeles}
  \country{USA}
}
\email{czhou691@usc.edu}

\author{Di Wang}
\orcid{0000-0002-2418-7987}
\affiliation{%
  \institution{Peking University}
  \city{Beijing}
  \country{China}
}
\email{wangdi95@pku.edu.cn}

\author{Thomas Reps}
\orcid{0000-0002-5676-9949}
\affiliation{%
  \institution{University of Wisconsin-Madison}
  \city{Madison}
  \country{USA}
}
\email{reps@cs.wisc.edu}

\begin{abstract}
Many costs that arise in practice---such as the completion time of a quantum repeater, the number of rounds of an RFID collision-resolution protocol, or the response time of a fork--join computation---are described by a \emph{hierarchical cost expression}: a tree of independent subcomputations with nonnegative integer costs, combined not only by addition but also by $\max$ and $\min$.
We study how to compute the moments of such costs with sound error bounds.
Moment-based analyses developed for additive costs do not apply, because the moments of $\max(X,Y)$ are not determined by those of $X$ and $Y$;
propagating full distributions becomes expensive as the trees grow;
and the approximations that do scale are hand-derived for particular families of structures and come without error bounds.

Our analysis proceeds bottom-up, summarizing each subtree by a
\emph{surrogate distribution}: an exact short-horizon prefix plus a compact parametric tail fitted to preserve the queried moment.
Each prefix is computed exactly from the children's prefixes---for a retry, whose cost sums an unbounded number of attempts, by solving a local recurrence equation---so all approximation is confined to the tails.
Abstraction therefore introduces shape error but no moment error;
nonlinear operators, such as $\max$ and $\min$, convert shape error into moment error.
The analysis tracks two bounds, one on shape error and one on moment error.
Tracking both lets it report a sound interval for the mean, the second moment, and the variance.
Lengthening the prefix, or solving subtrees exactly, narrows the interval at the cost of additional computation.

For quantum-repeater chains of up to 4096 links with high swap-success probability, our \tool tool bounds the mean completion time within 1\%, whereas the exact Markov-chain solution exhausts memory with 16 links, and the best hand-derived approximation we evaluated has 47\% error with 64 links.
\tool also bounds second moments, variances, and deadline-miss probabilities, and handles RFID and fork--join costs.
\end{abstract}



\maketitle

\section{Introduction}
\label{sec:introduction}

Many stochastic systems are built by divide and conquer.
A quantum repeater establishes long-distance entanglement by generating short links in parallel and joining adjacent segments by entanglement swaps.
An RFID reader identifies a population of tags by splitting it into subpopulations that are resolved separately.
A fork--join computation spawns parallel branches and waits for all of them.
In each of these cases, the cost of a subproblem is probabilistically \emph{independent} of the costs of its siblings.
Independence  makes the cost distribution of a composite subproblem a function of the cost distributions of its parts, and our analysis relies on that property.
The cost of such a system (completion time, number of protocol rounds, or response time) is a random variable assembled from the costs of its parts by a few recurring patterns:
subcomputations that must all finish (a maximum), subcomputations that race (a minimum), sequencing (a sum), probabilistic choice (a mixture), and a step that, on failure, forces its whole subcomponents to be redone (a retry).
We call a description of such a cost a \emph{hierarchical cost expression}: a tree
whose leaves are base costs and whose internal nodes are these composition operators.
For example, a node that waits for both of its children has cost $\max(X_1, X_2)$, where
$X_1$ and $X_2$ are the children's costs.

The designer of such a system wants to know how long a run will take, or how likely a run is to miss a deadline.
The mean answers those questions only when runs cluster tightly around it, but often they do not.
For quantum repeaters, Shchukin et al.~\cite{repeater-paper} found that the standard deviation of the completion time is comparable to its mean across practically relevant regimes, so a run that takes twice the mean is unremarkable.
The mean alone bounds the probability of such a run only by Markov’s inequality, $\Pr(X \ge 2\mu) \le 1/2$;
tighter bounds require higher moments or the shape of the distribution.
The problem we address is therefore

\vspace{.8ex}
\begin{mdframed}[innerleftmargin = 3pt, innerrightmargin = 3pt, skipbelow=-0.0em]
Approximate the moments of the cost denoted by a hierarchical cost expression (its mean, variance, and beyond), along with sound bounds on the error of each computed moment.
\end{mdframed}

\smallskip
Such costs fall outside the scope of moment-based cost analyses of probabilistic programs~\cite{kaminski2018weakest,ngo2018bounded,kura2019tail,wang2021central,moosbrugger2022polar}, which target \emph{additive costs}.
Those analyses can precisely propagate moments compositionally, which works for sums, mixtures, and random repetitions.
The difficulty is that $\max$ and $\min$ are \emph{nonlinear}:
the moments of $\max(X, Y)$ cannot, in general, be determined precisely by those of $X$ and $Y$.
For example, let $X \equiv 1$, let $Y_1 \equiv 1$, and let $Y_2$ be $0$ or $2$ with equal probability.
All three of $\E[X]$, $\E[Y_1]$, and $\E[Y_2]$ equal $1$, yet $\E[\max(X,Y_1)] = 1$, whereas $\E[\max(X,Y_2)] = 3/2$.
Propagating full distributions would suffice, but it becomes expensive as the cost structures grow;
for quantum repeaters, the exact Markov-chain model of Shchukin et al.~\cite{repeater-paper} has $2^n$ states for $n$ links.
Analyses that do scale for such structures rely on approximations hand-derived for particular families of cost structures.
\citeauthor{repeater-paper}, for example, nest exact small-system results inside effective-probability approximations, which target only the mean and come without error bounds.
For nested repeaters and related protocols, Coopmans et al.~\cite{coopmans2022improved} derive analytical bounds on the mean and the quantiles of the completion time.

\begin{figure}[!tb]
\centering
\includegraphics[width=\linewidth]{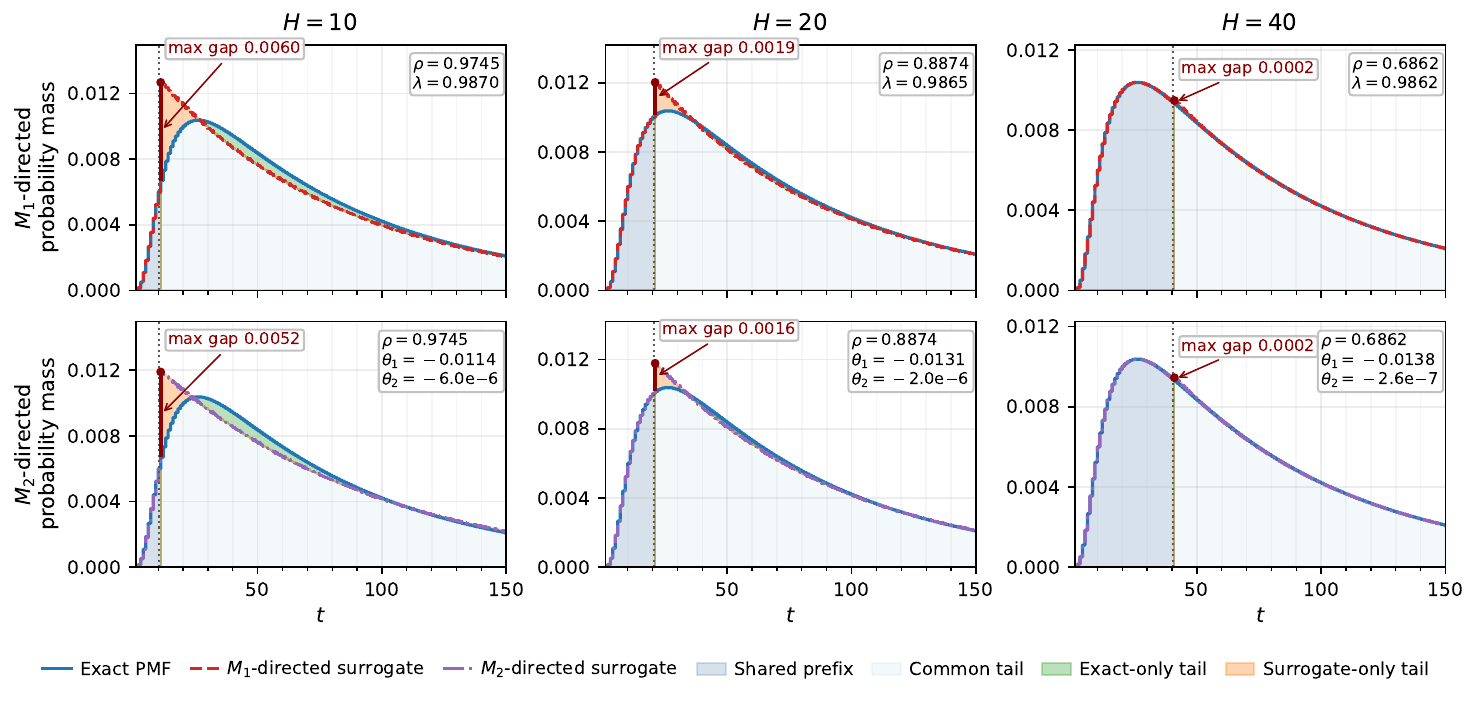}
\caption{
Increasing the prefix horizon $H$ improves the surrogate ($H\in\{10,20,40\}$, from left to right).
Each panel overlays the exact cost distribution for the completion time of a balanced repeater tree with $n=4$ links, leaf success probability $p=0.1$, and $\Retry$-node success probability $a=0.5$,
and the prefix--tail surrogate;
to the left of the dotted vertical line, the surrogate is exact;
to the right of it is a fitted tail.
The top row shows expectation-directed
($M_1$)
abstractions; the bottom row shows second-moment-directed
($M_2$)
abstractions.
In each panel, $\rho$ is the tail mass after time $H$, $\lambda$ (top) and $\theta_1,\theta_2$ (bottom) are the tail parameters (\Cref{sec:moment-tail}), and the red marker shows the largest gap in probability mass.}
\label{fig:teaser}
\vspace{-2ex}
\end{figure}

This paper adopts a middle ground between moments and full distributions, based on the observation that the operands' moments confine those of a $\max$ or $\min$ only to an interval. 
In the example above, $Y_1$ and $Y_2$ give $1$ and $3/2$, and over all costs $Y$ with $\mathbb{E}[Y] = 1$, $\mathbb{E}[\max(X, Y)]$ can lie anywhere in $[1, 2)$.
An accurate compositional analysis must therefore
retain some information about the shape of each operand's distribution.
We introduce \emph{prefix--tail surrogate distributions}, each consisting of (i) an exact short-horizon prefix and (ii) a compact parametric tail that captures the remaining mass (\Cref{fig:teaser}).
Our analysis proceeds bottom-up through the tree, computing a surrogate for each subproblem from the surrogates of its children.
The analysis is automated, applies to any hierarchical cost expression over the operators of \Cref{sec:background}, and reports, for
the mean, the second moment, and the variance,
an interval guaranteed to contain the true value.


The first challenge is to find surrogate distributions that remain small, yet retain enough information about shape that the moments of $\max(X, Y)$ and $\min(X, Y)$ computed from the surrogates of $X$ and $Y$ are close to the true moments.
Our first key observation is that composition treats the two ends of a distribution differently:

\vspace{.8ex}
\begin{mdframed}[innerleftmargin = 3pt, innerrightmargin = 3pt, skipbelow=-0.0em]
The prefix of a composite cost is determined by the prefixes of its parts.
Prefixes can therefore be kept exact, and all approximation confined to the tail.
\end{mdframed}

\smallskip
\noindent
For every built-in operator,
the result’s prefix is computed locally from the operands’ prefixes:
for $\max$ and $\min$, pointwise from the operands' survival functions;
for a sum, by convolution;
and for a retry, whose cost sums an unbounded number of attempts, by solving a finite recurrence  that conditions on whether the first attempt succeeds (\Cref{sec:overview}).
A prefix--tail surrogate therefore stores the masses of the costs $0, 1, \ldots, H$ exactly, for a prefix horizon $H$, and replaces the remaining mass by a compact parametric tail (geometric for the mean; maximum-entropy for higher moments), fitted to preserve the queried moment and all lower moments.
Unlike truncation, which discards that mass, the tail retains it, so abstraction changes a distribution’s shape but not its queried moment.

The second challenge is soundness.
Bounding only the distance between each surrogate and the true distribution would be sound---because that distance also bounds the moment error---but pessimistic:
in the example in \Cref{sec:overview}, the distance bound limits the moment error at the root to 0.40, versus 0.009 when moment error is bounded separately.
Bounding only the moment error would not suffice, because a $\max$ or $\min$ higher in the tree needs a bound on shape error.
Our second key observation is that the two kinds of errors arise in different places:

\vspace{.8ex}
\begin{mdframed}[innerleftmargin = 3pt, innerrightmargin = 3pt, skipbelow=-0.0em]
Abstraction introduces shape error, but no moment error
(modulo fitting error for second and higher moments, \Cref{sec:moment-tail}).
Nonlinear composition converts shape error into moment error.
\end{mdframed}

\smallskip
\noindent
\noindent
The analysis therefore maintains two bounds at each node: $\epsilon$ on shape error, measured by a moment-weighted survival distance $d$ (for the mean, the Wasserstein-1 distance), and $\delta$ on moment error.
For a cost expression $e$ with distribution
$\llbracket e \rrbracket$ and surrogate $s$, the bounds satisfy, for the mean,\footnote{
  In this paper, we use $\mathbb{E}[\mu]$ to denote $\mathbb{E}_{X \sim \mu}[X]$, where $\mu$ is a cost distribution.
}
\[
  d(\llbracket e \rrbracket, s) \le \epsilon
  \quad\text{and}\quad
  |\mathbb{E}[\llbracket e \rrbracket] - \mathbb{E}[s]| \le \delta .
\]
The analysis propagates $\epsilon$ and $\delta$ through a hierarchical cost expression via per-operator error-bound transformers.
Every such transformer has the same form, a linear combination of the bounds on the operator's inputs (plus, for $\epsilon$, the local loss of re-abstracting the result).
The transformers of different operators differ only in their coefficients, which must be Lipschitz constants of the operator's semantics:
perturbing an input by a given distance may change the output distribution, and the output moment, by at most a fixed multiple of that distance (\Cref{sec:operator-contracts}).
Thus, a new operator can be incorporated soundly by identifying coefficients for it and proving that they satisfy the Lipschitz condition;
the soundness theorem (\Cref{thm:generic-soundness}) then applies to the new operator automatically.


Prefix--tail surrogate distributions support two kinds of refinements.
Increasing the prefix horizon improves precision locally, which compounds as the analysis moves up the tree (\Cref{fig:teaser}).
Alternatively, solving a subproblem exactly and reusing it as a leaf eliminates its internal approximation error,
leaving only the error of abstracting its result at its root.

Our work makes the following contributions:
\begin{itemize}
  \item
    We formulate the problem of bounding the moments of the cost denoted by a hierarchical cost expression as that of computing (i) a compact surrogate distribution, and (ii) two kinds of error bounds---$\epsilon$, which bounds shape error, and $\delta$, which bounds moment error (\Cref{sec:background}).
  \item We present an abstract interpretation \cite{DBLP:conf/popl/CousotC77} based on
    \emph{prefix--tail surrogate distributions}, in which prefixes are computed exactly by local transformers (for retries, by solving a local recurrence) and tails preserve the queried moment and all lower moments (\Cref{sec:method} and \Cref{sec:moment-method}).
    A prefix--tail surrogate distribution, together with $\epsilon$ and $\delta$, represents a set of concrete distributions.
  \item
    We give an abstract semantics that propagates $\epsilon$ and $\delta$ via per-operator error-bound transformers whose coefficients are Lipschitz bounds for the operators' semantics.
    A single soundness proof covers every such operator (\Cref{sec:operator-contracts}, \Cref{sec:error-bounded-semantics}, and \Cref{sec:moment-semantics}).
    The computed bounds guide both refinement (\Cref{sec:greedy-refinement}) and a search over alternative tree shapes (\Cref{sec:design-search}).
  \item
    We implemented the framework in a tool, called \tool.
    On repeater chains modeled as in Shchukin et al.~\cite{repeater-paper}, \tool bounds the mean completion time within 1\% for up to 4096 links in our main configuration, whereas the more accurate of their approximations has 47\% error with 64 links, and their exact solution exhausts memory with 16 links (\Cref{sec:exp-chains}).
    On a ten-link chain, a beam search guided by \tool's intervals finds the tree with the smallest mean among all 4862 possible trees at its ninth candidate (\Cref{subsec:exp-tree-design}).
    On RFID collision resolution over a noisy channel, \tool's estimate of the expected number of rounds errs by at most $5 \times 10^{-6}\%$, versus 3\%--25\% for closed-form and scalar-summary baselines.
    On fork--join deadline-miss queries, its upper bounds are 23--90 times tighter than Markov bounds (\Cref{sec:exp-other}).
\end{itemize}
\tool supports a cost-expression language (\Cref{sec:background-recursive-cost-programs} and \Cref{sec:background-cost-composition}) and tree generators (\Cref{sec:background-recursive-cost-programs} and \Cref{sec:general-generator-constructions}). With them, the analysis can be invoked programmatically over families of instances and over the alternative tree shapes explored by a search (\Cref{sec:design-search} and \Cref{subsec:exp-tree-design}).

\Cref{sec:discussion} and \Cref{sec:related} cover extensions and related work.
Additional details, proofs, and experiments are given in the appendices.


\section{Overview}
\label{sec:overview}

This section walks through the analysis on the small hierarchical cost expression of \Cref{fig:overview-numeric-repeater}, computing its expected completion time with error bounds.
The input is such an expression and a query, here the expected cost.
The output is an interval $\widehat K\pm\delta$, where $\widehat K$ is the analyzer’s estimate of the expected cost and $\delta$ bounds its error:
$|K-\widehat K|\le\delta$ for the exact expected cost $K$.

\subsection{The Concrete Input}
\label{sec:overview-example}

\begin{figure}[!tb]
\centering
\begin{minipage}[c]{0.40\linewidth}
\centering
\begin{tikzpicture}[x=0.62cm,y=0.44cm,
  v/.style={circle,draw,minimum size=6.5mm,inner sep=0pt,font=\small},
  i/.style={v,label={[font=\scriptsize]#1:$a{=}1/2$}},
  l/.style={v,label={[font=\scriptsize]below:$p{=}1/2$}}]
  \node[i=right] (r) at (0,3.2) {$r$};
  \node[i=left] (vl) at (-2.1,1.6) {$v_L$};
  \node[i=right] (vr) at (2.1,1.6) {$v_R$};
  \node[l] (u1) at (-3.3,0) {$u_1$};
  \node[l] (u2) at (-0.9,0) {$u_2$};
  \node[l] (u3) at (0.9,0) {$u_3$};
  \node[l] (u4) at (3.3,0) {$u_4$};
  \draw (r)--(vl)--(u1) (vl)--(u2) (r)--(vr)--(u3) (vr)--(u4);
\end{tikzpicture}
\end{minipage}%
\begin{minipage}[c]{0.58\linewidth}
\small
\[
\begin{aligned}
  u_i &= \Atom(\Geom_{\ge 1}(1/2)) \qquad (1\le i\le 4)\\
  v_L &= \Retry_{1/2}(\max(u_1,u_2))\\
  v_R &= \Retry_{1/2}(\max(u_3,u_4))\\
  r   &= \Retry_{1/2}(\max(v_L,v_R))
\end{aligned}
\]
\end{minipage}
\vspace{-1.0ex}
\caption{The tree structure and per-node parameters used throughout this section, and its cost expression.}
\label{fig:overview-numeric-repeater}
\vspace{-1ex}
\end{figure}

\Cref{fig:overview-numeric-repeater} (left) shows the tree and each node's parameter.
Each leaf is a task that succeeds in each time step independently with probability $p$,
so its cost is $\Geom_{\ge1}(p)$.
Each internal node waits for both of its children, then performs a step that succeeds
with probability $a$. If that step fails, the node's whole subtree must be redone.
To keep the arithmetic readable we take $p=a=1/2$ and horizon $H=4$.\footnote{
  The example is a four-segment quantum repeater~\cite{repeater-paper}: leaves are elementary entanglement links and internal nodes are entanglement swaps. A failed swap destroys the quantum state built below it, so failure repeats the entire subtree.
}
Numbers are rounded to four decimal places or three significant digits.

Assigning a composition operator to every internal node yields the hierarchical cost expression $e$, written in \Cref{fig:overview-numeric-repeater} (right) as a system of equations.\footnote{
  We use one symbol for a node of the expression tree, its cost (a random variable), and that cost’s distribution.
  Which quantity is intended should always be clear from context.
}
$\Retry_a(X)$ describes the cost of rejection sampling with acceptance probability $a$: perform the work described by $X$, accept it with probability $a$, and otherwise throw it away and start over.
Its cost is therefore the sum of $N$ i.i.d.\ copies of $X$ with $N\sim\Geom_{\ge1}(a)$ (\Cref{def:cost-operators}).
For larger instances, $e$ need not be written out explicitly: a tree generator can produce it programmatically (\Cref{sec:cost-expressions}).

\subsection{Summaries}
\label{Overview:Summaries}

A summary that stores only the expected completion time would not suffice: at $v_L$, the expectation of $\max(u_1,u_2)$ is not determined by the child expectations alone, so the parent needs information about the shape of the child distributions.
A full distribution would contain enough information, but such distributions can be too costly to propagate through large trees.
The analyzer therefore stores a \emph{prefix--tail summary}\footnote{
  Henceforth, we use ``summary'' and ``surrogate distribution'' interchangeably.
}
$s=(f_1,\ldots,f_4,\lambda)$ at each node, where $f_t$ is the probability that the cost equals $t$---costs here are positive, so there is no $f_0$---and $\lambda\in[0,1)$ is the parameter of a geometric tail.
Summary $s$ denotes the \emph{surrogate distribution}
\begin{equation}
\label{eq:overview-summary}
  \Pr(s=t)=f_t \quad (t\le H),
  \qquad
  \Pr(s=t)=\rho\,(1-\lambda)\,\lambda^{t-H-1} \quad (t>H),
\end{equation}
where the tail mass $\rho=1-\sum_{t\le H}f_t$ is derived from the prefix.
The prefix is exact.
The tail is approximate, but $\lambda$ is chosen to preserve the expectation of the distribution being abstracted.

In our example, each leaf has the geometric distribution $\Pr(u_i = t) = p(1-p)^{t-1}$, with support $t \in \{1, 2, 3, \ldots\}$, which has expectation $2$.
For $H = 4$, the summary
has the prefix $(0.5,0.25,0.125,0.0625)$, which leaves tail mass $\rho = 0.0625$.
To preserve the expectation $2$, we use $\lambda = 0.5$:
the tail distribution is $\Pr(s = t) = 0.0625 \cdot 0.5^{t-4}$, for $t > 4$.
Because the surrogate is exact, $\epsilon = 0$ and $\delta = 0$.

\subsection{Computing Summaries}
\label{sec:overview-compute-summary}


\begin{table}[!tb]
\centering
\caption{
Bottom-up computation for the running example ($H=4$).
The rows for $\max$ show the abstracted distributions before $\Retry$.
The abstraction for $v_R$ is the same as the one for $v_L$ by symmetry.
$\widehat{K}$ denotes the estimated mean;
$\epsilon$ bounds the shape error, and $\delta$ bounds the error in $\widehat{K}$ (\Cref{sec:Overview:Bounds}).
}
\label{tab:overview-first-step}
\vspace{-1.0ex}
\footnotesize
\setlength{\tabcolsep}{3.6pt}
\begin{tabular}{@{}l|rrrrrr r|rr@{}}
\toprule
Node & $f_1$ & $f_2$ & $f_3$ & $f_4$ & $\rho$ & $\lambda$ & $\widehat{K}$ & $\epsilon$ & $\delta$ \\
\midrule
Leaf $u_i$
& 0.5000 & 0.2500 & 0.1250 & 0.0625 & 0.0625 & 0.5000 & 2.0000 & 0 & 0 \\
$B_L = \max(u_1,u_2)$
& 0.2500 & 0.3125 & 0.2031 & 0.1133 & 0.1211 & 0.5053 & 2.6667 & 0.000835 & 0 \\
$v_L=\Retry_{1/2}(\max(u_1,u_2))$
& 0.1250 & 0.1719 & 0.1426 & 0.1140 & 0.4465 & 0.7965 & 5.3333 & 0.00227 & 0 \\
$B_r = \max(v_L,v_R)$
& 0.0156 & 0.0725 & 0.1050 & 0.1132 & 0.6937 & 0.8195 & 7.5470 & 0.1750 & 0.00453 \\
$r=\Retry_{1/2}(\max(v_L,v_R))$
& 0.0078 & 0.0363 & 0.0531 & 0.0587 & 0.8441 & 0.9249 & 15.0939 & 0.4034 & 0.00906 \\
\bottomrule
\end{tabular}

\vspace{-2ex}
\end{table}

The summary at $v_L$ is computed in two stages:
first for $B_L = \max(u_1,u_2)$; then for $v_L = \Retry_{1/2}(B_L)$.
In each stage, the prefix (the first four masses) is computed from the children’s prefixes.
Because $u_1$ and $u_2$ are independent, and $B_L\le t$ exactly when both finish by time $t$, $\Pr(B_L\le t) = {\Pr(u_1\le t)}\Pr(u_2\le t)$.
Thus, $\Pr(B_L=2) = \Pr(B_L\le 2) - \Pr(B_L\le 1) = 0.75^2-0.5^2=0.3125$.
Continuing to time $4$ gives the prefix $(0.2500,0.3125,0.2031,0.1133)$, leaving tail mass $\rho_{B_L} = 0.1211$.
The analyzer selects
the geometric tail with parameter $\lambda = 0.5053$, which preserves $\E[B_L] = 8/3$ (so $\delta_{B_L} = 0$).
Its shape-error bound is $\epsilon_{B_L} = 0.000835$.

Next, $\Retry_{1/2}$ is applied to the surrogate of $B_L$.
The cost $v_L$ is the total completion time of the subtree rooted at $v_L$, including the cost of any failed attempts.
With $b_t=\Pr(B_L=t)$ and $q_t=\Pr(v_L=t)$, conditioning on whether the first attempt succeeds gives the finite recurrence
\begin{equation}
\label{Eq:RetryRecurrence}
  q_t \;=\; a\,b_t+(1-a)\sum_{u=1}^{t-1}b_u\,q_{t-u}
  \qquad(t\ge1),
\end{equation}
so the prefix of $v_L$ depends only on that of $B_L$.
Evaluating \labelcref{Eq:RetryRecurrence} for $t = 1, \ldots, 4$ gives the $f_1, \ldots, f_4$ values in the $v_L$ row of \Cref{tab:overview-first-step}, leaving tail mass $\rho_{v_L} = 0.4465$.
By Wald's identity, $\E[v_L] = \E[B_L]/a\approx5.3333$.
The analyzer selects the geometric tail with parameter $\lambda = 0.7965$, which preserves $\E[v_L]$ (so $\delta_{v_L} = 0$).
Its shape-error bound is $\epsilon_{v_L} = 0.00227$.

The root repeats this pattern.
For $B_r = \max(v_L,v_R)$, most of the mass is beyond $H$, so the shape error is larger ($\epsilon_{B_r} = 0.1750$).
Because the summaries
of $v_L$ and $v_R$ have nonzero shape error,
there is now uncertainty in the mean for $B_r$, $\widehat{K}_{B_r} = 7.5470$, but
the analyzer
computes the error bound $\delta_{B_r} = 0.00453$.
At the root $r$, the analyzer reports $\widehat{K}_r \pm \delta_r = 15.0939 \pm 0.00906$.

Because our example is a small tree, an exact calculation is feasible, which gives $K_r \approx 15.0939068$.
Thus, the actual error in the surrogate is $<$$10^{-5}$ (well within the computed bound of $0.00906$).

\subsection{Bounds on Shape Error and Moment Error}
\label{sec:Overview:Bounds}

At every node, the shape bound $\epsilon$ bounds shape error:
the Wasserstein-1 distance $\dW(X, X')=\sum_{t \ge 0}|\Pr(X > t) - \Pr(X' > t)|$ on $\N$ between the exact distribution $X$ and an approximation $X'$, such as its surrogate.
The moment bound $\delta$ bounds the error in the queried moment, here the mean.
Three properties of $\dW$ yield the bounds in \Cref{tab:overview-first-step}:
(i)~$|\E[X]-\E[X']|\le\dW(X,X')$;
(ii)~if $X$ and $Y$ are independent, and so are $X'$ and $Y'$, then $\dW(\max(X,Y),\max(X',Y'))\le\dW(X,X')+\dW(Y,Y')$;
and (iii)~$\dW(\Retry_a(X),\Retry_a(X'))\le\dW(X,X')/a$, with $\E[\Retry_a(X)]=\E[X]/a$.

Let us elaborate on how to use these three properties.
Each abstraction step also records its \emph{local loss} $\eta$, an upper bound on the $\dW$ distance between the distribution it abstracts and the resulting surrogate.
In the example, abstracting $B_L$ incurs $\eta_1=0.000835$, and abstracting the $\Retry$ at $v_L$ incurs $\eta_2=0.000596$.
At $v_L$, the two children are exact, so the bound of $B_L$ equals its local loss, $\epsilon_{B_L}=\eta_1=0.000835$, and $\epsilon_{v_L}=\eta_1/a+\eta_2\approx0.00227$:
by (iii), the first loss, incurred before the $\Retry$, is repeated $1/a$ times in expectation, and the second is added directly.
Abstraction changes the shape of a tail but creates no expectation bias, so $\delta_{v_L}=0$ although $\epsilon_{v_L}\ne0$.
At the root, the maximum of two approximate children turns their shape errors into an error in the expected maximum, which is at most $\epsilon_{v_L}+\epsilon_{v_R}=0.00453$ by (ii) and (i);
the root $\Retry$ multiplies it by $1/a$, giving $\delta_r\approx0.00906$.
These two steps instantiate the second key observation of \Cref{sec:introduction}:
at $v_L$, abstraction has produced shape error but no moment error;
and at the root, $\max$ has converted the children's shape errors into moment error.
Without $\delta$, the moment error at the root could be bounded only through (i), by $\epsilon_r=0.4034$, which is about $45$ times $\delta_r$.
On the 4096-link chains with $p=0.37$ of \Cref{sec:exp-chains}, $\epsilon_r$ is about $10$ times $\delta_r$ at $H=1024$ and about $100$ times at $H=2048$.
Without $\epsilon$, $\delta$ could not be propagated through the $\max$ at the root, which needs the children's shape errors.
The analysis therefore maintains both.

\subsection{Precision Control and Refinement}
\label{sec:overview-refinement}

The result of one run is an interval, which the analyzer can tighten at the cost of additional computation.
One option is to increase the prefix horizon $H$, which raises three questions.
\emph{What makes a good $H$?}
The local loss at a node is at most twice the expected excess of the abstracted cost over $H$ (proof of \Cref{prop:convergence}, \Cref*{app:refinement:prop:convergence}), so a good $H$ leaves little mass beyond it while keeping the cost per node affordable.
\emph{Do the error bounds improve monotonically?}
They are monotone in the local losses (\Cref{prop:monotonicity}), which a larger $H$ typically reduces, and they converge to zero as $H$ grows (\Cref{prop:convergence}).
\emph{What does it cost?} Each node takes
at most
$O(H^2)$ operations with direct evaluation, or $O(H\log H)$ with the fast Fourier transform (FFT) (\Cref{sec:complexity}).
Another option is to \emph{solve a small subtree exactly} and use its abstracted distribution as an atom, which removes the losses inside the subtree but keeps one at its
root.
A change is kept only if reanalysis shows a tighter bound (\Cref{sec:refinement}).

The error bounds also let one compare alternative tree shapes, such as different orders in which a repeater joins its links.
If two candidates
yield separated intervals, say $\widehat K_1+\delta_1<\widehat K_2-\delta_2$, the first has the smaller expected completion time; if they overlap, more analysis effort is needed.

\mypara{From the example to the framework.}
In general, the analyzer works bottom-up over a hierarchical cost expression.
The running example uses $\max$ and $\Retry$; other hierarchical and divide-and-conquer protocol families, such as fork--join response times~\cite{davis2019survey} and RFID tree-splitting~\cite{yan2015memoryless}, also use $+$, random repetition, and mixtures.
Each operator supplies a finite-prefix transformer for the exact prefix and the coefficients of its error-bound transformer;
\Crefrange{sec:background}{sec:refinement} formalize these pieces.

\section{Formal Setting and Problem Statement}
\label{sec:background}

This section introduces the formal setting
and defines the problem that the rest of the paper addresses.
The central \emph{semantic objects} are nonnegative integer-valued cost distributions induced by cost expressions, and the analysis computes error-bounded summaries of those distributions.
\Cref{tab:notation} gives a concise guide to the notation used in the paper.

\begin{table}[!tb]
\caption{Notation, with the sections that define it.}
\label{tab:notation}
\vspace{-1.0ex}
\scriptsize
\setlength{\tabcolsep}{2.4pt}
\begin{tabular}{@{}lll@{\hspace{5pt}}|@{\hspace{5pt}}lll@{}}
\toprule
Symbol & Meaning & \S & Symbol & Meaning & \S \\
\midrule
$e$, $\sem{e}$; $\Atom(\mu)$ & cost expression; its distribution; leaf & \ref{sec:cost-expressions} &
$\pre_H$, $\omega^{\#}_H$ & prefix of a cost; finite-prefix transformer & \ref{sec:cost-operator-contracts} \\
$\omega$, $\sem{\omega}$ & composition operator; its semantics & \ref{sec:background-cost-composition} &
$L_i$, $A_i$, $C_i$ & coefficients of error-bound transformers & \ref{sec:cost-operator-contracts} \\
$Q$; $K$, $\widehat K$ & query; exact and estimated mean cost & \ref{sec:overview} &
$\gamma_H(s,\epsilon,\delta)$ & concretization of an abstract value & \ref{sec:prefix-tail-domain} \\
$\epsilon$, $\delta$ & shape and moment bounds & \ref{sec:problem} &
$k$; $M_j$, $d_j$ & order; raw moment $\E[X^j]$; its distance & \ref{sec:moment-domain} \\
$S_X$; $\dW$ & survival function $\Pr(X>t)$; Wasserstein-1 distance & \ref{sec:distance} &
$(H,\mathcal O)$ & precision profile; subtrees solved exactly & \ref{sec:problem} \\
$H$; $f_0,\ldots,f_H$ & prefix horizon; stored prefix masses & \ref{sec:prefix-tail-domain} &
$W_\ell$ & amplification of the loss at site $\ell$ & \ref{sec:refinement-profiles} \\
$s$; $\rho$; $\lambda$, $\theta$ & summary; tail mass; tail parameters & \ref{sec:prefix-tail-domain} &
$p$, $a$ & link and swap success probabilities & \ref{sec:overview-example} \\
$\beta_H$; $\eta$; $J$ & abstraction; local loss; evaluation horizon & \ref{sec:prefix-tail-domain} & $\zeta_j$ & bias bound & \ref{sec:moment-tail} \\
$w_j(t)$ & moment weight $(t+1)^j - t^j$ & \ref{sec:moment-domain} & $H_s$, $H_f$ & search and final prefix horizons & \ref{sec:design-search} \\
\bottomrule
\end{tabular}
\vspace{-2ex}
\end{table}

\subsection{Hierarchical Cost Expressions and Their Distributions}
\label{sec:cost-expressions}
\label{sec:background-recursive-cost-programs}

As described in \Cref{sec:introduction}, the object the analysis takes as input is a single \emph{hierarchical cost expression}: a finite tree whose leaves are base costs and whose internal nodes are composition operators, so that a fixed tree induces a \emph{distribution} over costs.

\begin{definition}[Hierarchical cost expressions]
\label{def:cost-trees}
Hierarchical cost expressions (\emph{cost expressions} for short) are defined by the following regular-tree grammar
of constructor functions:
\[
e ::= \Atom(\mu) \mid \omega_\theta(e_1,\ldots,e_m).
\]
$\theta$ denotes a parameter
(or set of parameters) of $\omega$,
such as the retry-success probability $a$ in $\Retry_a$.
\end{definition}

An \emph{atom} $\Atom(\mu)$ is a leaf with base distribution $\mu$ on $\N=\{0,1,\ldots\}$; a constant cost $c$, as in the RFID model of \Cref{sec:front-ends}, stands for the atom whose base distribution is the point mass at $c$.
Expressions contain no variables and are finite; \Cref{sec:extensions} discusses unbounded generators.

\begin{definition}[Cost and induced cost distribution]
\label{def:induced-cost}
A \emph{cost} is a nonnegative integer-valued random variable.
Costs count discrete units, such as the time steps of a repeater or the query rounds of an RFID reader.
A cost expression $e$ denotes a distribution $\sem{e}$ on $\N$, its \emph{induced cost distribution}: $\sem{\Atom(\mu)}=\mu$, and $\sem{\omega_\theta(e_1,\ldots,e_m)}=\sem{\omega_\theta}(\sem{e_1},\ldots,\sem{e_m})$, where $\sem{\omega_\theta}(\mu_1,\ldots,\mu_m)$ is the distribution of $\omega_\theta(X_1,\ldots,X_m)$ for independent $X_i\sim\mu_i$ (\Cref{def:cost-operators}).
\end{definition}

The user may write this tree down directly, or obtain it from an optional \emph{generator}: a recursive program that constructs a cost expression from parameters, such as the interval $i,\ldots,j$ of a repeater chain (\Cref{sec:front-ends}).
Recursive calls receive strictly smaller parameters, so the
generated
expression is finite.

If a subexpression is generated several times, each occurrence denotes an independent copy of its cost.
Generators are deterministic: randomness in a system's structure, such as the sizes of the parts of a split, is expressed by $\Mix$ and $\RepeatOp_N$ nodes, as in the RFID generator of \Cref{sec:general-generator-constructions}.

\subsection{Cost Composition}
\label{sec:background-cost-composition}

Cost composition is not limited to addition.
The following operators formalize the composition patterns that appear across divide-and-conquer systems.

\begin{definition}[Composition operators]
\label{def:cost-operators}
Let $X, Y, X_1, \dots, X_m$ be independent costs. The composition operators are
\emph{sequential composition} $X + Y$, which accumulates cost;
\emph{barrier composition} $\max(X_1, \dots, X_m)$, which waits for all subcomputations to finish;
\emph{race composition} $\min(X_1, \dots, X_m)$, which waits for the first subcomputation to finish;
\emph{probabilistic choice} $\Mix_w(X_1, \dots, X_m) := X_I$, where $w$ is a probability vector and $I$ is independent of the $X_i$ with $\Pr(I = i) = w_i$;
and \emph{random repetition} $\RepeatOp_N(X) := \sum_{j=1}^{N} X^{(j)}$, where the count $N$ is an independent $\N$-valued random variable with finite mean, the $X^{(j)}$ are i.i.d.\ copies of $X$, and the empty sum is $0$.
\emph{Geometric retry} is the special case $\Retry_a(X) := \RepeatOp_N(X)$ with $N \sim \Geom_{\ge1}(a)$ for $a\in(0,1]$, so that $\E[N] = 1/a$; it models a subcomputation that is restarted until an outcome succeeds independently with probability $a$.
Here $\Geom_{\ge 1}(a)$ denotes the geometric distribution with success probability $a$ on $\{1,2,\dots\}$.
\end{definition}

\Cref{def:cost-operators} defines the five operators built into \tool. 
New operators can also be added:
\Cref{sec:cost-operator-contracts} states the interface for analyzing an operator soundly with exact prefixes.
\Cref{ex:new-operator} uses that interface to add a quorum operator that waits for $k$ of $m$ subcomputations, and \Cref{sec:extensions} relaxes it to add an operator for the idle time $|X-Y|$ of the faster of two branches that start together.

\mypara{Independence.}
\label{sec:independence}
As noted in \Cref{sec:introduction},
independence is a modeling assumption, built into the semantics (\Cref{def:induced-cost}): the analysis is sound for a system whose sibling costs are independent; in the three applications we study, independence is a property of the problem being modeled.
The quantum-repeater model treats elementary-link generation and successive swap attempts as independent events, as in the Markov-chain formulation of Shchukin et al.~\cite{repeater-paper}.\footnote{
  They also present a non-independent setting as a generalization.
}
After a split, the RFID protocol resolves two disjoint groups of tags, which respond independently.
The fork--join model assumes parallel branches with independent execution times.

\subsection{Error-Bounded Analysis}
\label{sec:problem}
\label{sec:background-error-bounded-summarization}

\begin{definition}[Error-bounded analysis]
\label{def:err-bounded-summarization}
Fix a distance $d$ on cost distributions, a target query $Q$ (in our main development $Q(X)=\E[X]$), and a
\emph{precision profile} $(H,\mathcal O)$.
The profile consists of a prefix horizon $H$ (which refinement may increase within a subtree; see \Cref{sec:refinement-profiles}) and a set $\mathcal O$ of subtrees---none containing another---whose exact distributions are used as atoms.
An \emph{error-bounded analysis} maps a cost expression $e$ to a summary $s$---a surrogate distribution---and to error bounds $\epsilon,\delta\ge0$ such that
\[
  d(\sem{e},s)\le\epsilon
  \qquad\text{and}\qquad
  |Q(\sem{e})-Q(s)|\le\delta .
\]
\end{definition}

The first inequality bounds the shape error;
the second bounds the moment error.
Varying the precision profile---for example, by increasing a local prefix horizon or by replacing a subproblem with an exact summary---changes both the surrogate and the error-bound interval.
As the prefix horizon grows, both bounds converge to zero (\Cref{prop:convergence}).
These guarantees are static upper bounds derived compositionally from the operators' error-bound transformers and local loss bounds; they hold for every precision profile, including the one that the greedy refinement of \Cref{sec:greedy-refinement} selects under a cost budget.
\Cref{sec:method} instantiates the prefix--tail abstract domain, the distance $d$ (a Wasserstein-1 survival distance), and the abstraction $\beta_H$.
The framework is thus built around a small \emph{kernel}---the semantics of cost expressions, the prefix--tail abstraction, the error-bounded abstract semantics, and refinement---and application-specific \emph{plugins}: the operators, generators, tail family, and query.

\section{Error-Bounded Prefix--Tail Analysis for Expectations}
\label{sec:method}

This section presents the analysis for expectations: the input is a hierarchical cost expression over the operators of \Cref{sec:background}, and the query is $Q(X)=\E[X]$.
The analyzer computes a compact summary at each node bottom-up and propagates a shape bound $\epsilon$ and a moment bound $\delta$ through the operators' error-bound transformers.
We assume that
the distributions of
atoms have finite means and that every $\RepeatOp_N$ has $\E[N]<\infty$, so that all costs have finite means.

\subsection{Prefix--Tail Abstract Domain}
\label{sec:prefix-tail-domain}
\label{sec:distance}

\begin{definition}[Prefix--tail summary]
\label{def:summary}
For a horizon $H\in\N$, a \emph{summary} is a tuple $s=(f_0,\ldots,f_H,\lambda)$ with $f_t\ge0$, $\sum_{t\le H}f_t\le1$, and $\lambda\in[0,1)$.
Its \emph{tail mass} $\rho=1-\sum_{t\le H}f_t$ is derived from the prefix.
The summary denotes the distribution
\begin{equation}
\label{Eqn:Concretization}
  \Pr(s=t)=
  \begin{cases}
    f_t, & 0\le t\le H,\\
    \rho\,(1-\lambda)\,\lambda^{t-H-1}, & t\ge H+1,
  \end{cases}
\end{equation}
and we use $s$ also for this distribution.
\end{definition}

A summary has the closed-form mean $\E[s]=\sum_{t\le H}t\,f_t+\rho\,\bigl(H+1/(1-\lambda)\bigr)$.

\mypara{Distance.}
We measure shape error by the survival distance
\[
  \dW(X, Y) = \sum_{t \ge 0} |S_X(t) - S_Y(t)|,
\]
where $S_X(t) = \Pr(X > t)$ is the survival function of $X$; on $\mathbb{N}$, $\dW$ is
the Wasserstein-1 distance.
By the tail-sum formula $\E[X]=\sum_{t\ge0}S_X(t)$, the difference of two means is the sum of the differences of their survival functions, so
\begin{equation}
\label{eq:mean-dist}
  |\E[X] - \E[Y]| \le \dW(X,Y).
\end{equation}
The analysis applies \labelcref{eq:mean-dist} with $Y$ a surrogate of $X$:
a bound on $\dW(X, Y)$ then also bounds the error in the mean, but not conversely, because distributions with equal means can differ in shape.
The distance $\dW$ keeps enough shape information for $\max$ and $\min$, whose means are not determined by the means of their operands.
For these operators, bounds on the operands' mean errors give no bound on the result's mean error, but bounds on the operands' distances do:
the exact result is within the sum of those distances of the result computed from the surrogates (\Cref{sec:operator-contracts}), and \labelcref{eq:mean-dist} turns that bound into a bound on the error in the mean.

\mypara{Concretization.}
An abstract value is a triple $(s,\epsilon,\delta)$ of a summary and two error bounds:
$s$ is a surrogate distribution;
the shape bound $\epsilon$ bounds the $\dW$ distance of the exact distribution from $s$, and the moment bound $\delta$ bounds the error in the mean.
The concretization of $(s,\epsilon,\delta)$ is
\begin{equation}
\label{eq:set-concretization}
  \gamma_H(s,\epsilon,\delta)=\bigl\{X \bigm| \Pr(X=t)=f_t\ \text{for}\ t\le H,\ \dW(X,s)\le\epsilon,\ \text{and}\ |\E[X]-\E[s]|\le\delta\bigr\},
\end{equation}
an ordinary set-valued concretization in the sense of abstract interpretation.

\mypara{Abstraction.}
The abstraction $\beta_H$, a representation function~\cite{nielson1999principles}, maps a cost distribution $X$ to the summary $\beta_H(X)$ that keeps the prefix of $X$, $f_t=\Pr(X=t)$ for $t\le H$; hence its tail mass is $\rho=\Pr(X>H)$.
If $\rho>0$, it chooses $\lambda$ with $\lambda/(1-\lambda)=\E[X-H-1\mid X>H]$, the mean of the residual tail of $X$; if $\rho=0$, it sets $\lambda=0$.
Because a geometric distribution on $\N$ with parameter $\lambda$ has mean $\lambda/(1-\lambda)$, the surrogate tail has the same conditional mean, and
\begin{equation}
\label{eq:mean-preserving}
  \E[\beta_H(X)]=\E[X].
\end{equation}
Abstraction
via $\beta_H$ thus changes only the shape of the tail.
Hence every cost distribution $X$ lies in the concretization of its own summary $\beta_H(X)$, paired with any shape bound $\eta \ge \dW(X,\beta_H(X))$ and moment bound $0$---i.e., $X\in\gamma_H(\beta_H(X),\eta,0)$.

\mypara{Local loss.}
The loss of abstracting $X$ is
$\dW(X,s)$, where $s=\beta_H(X)$.
It depends on 
the whole tail of $X$, so the analysis records a computable upper bound $\eta$.
For an evaluation horizon $J\ge H$,
\begin{equation}
\label{eq:loss-bound}
  \eta=\underbrace{\sum_{t=0}^{J}\bigl|S_X(t)-S_s(t)\bigr|}_{\text{difference up to }J}
  +\underbrace{\Bigl(\E[X]-\sum_{t=0}^{J}S_X(t)\Bigr)}_{=\,\sum_{t>J}S_X(t)}
  +\underbrace{\Bigl(\E[s]-\sum_{t=0}^{J}S_s(t)\Bigr)}_{=\,\sum_{t>J}S_s(t)}.
\end{equation}
The first term of \labelcref{eq:loss-bound} is the sum defining $\dW(X,s)$ up to $J$.
Beyond $J$, each summand satisfies $|S_X(t) - S_s(t)| \le S_X(t) + S_s(t)$, so the rest is at most the remaining tail sums of $X$ and $s$; by the tail-sum formula, each equals a mean minus the survival values up to $J$.
Both means are known:
$\E[X]$ is given for an atom and computed from the operands' summaries otherwise (\Cref{sec:operator-contracts}),
and $\E[s]=\E[X]$ by \labelcref{eq:mean-preserving}.
The bound is sound for every $J$ and tightens with $J$.

\mypara{Why geometric tails.}
Geometric tails have closed-form masses, survival values, and means, so abstraction is cheap and preserves the mean exactly; beyond $H$, the survival function of a $\max$ or $\min$ of surrogates is a combination of geometric sequences, and costs built from retries decay geometrically. For higher moments, \Cref{sec:moment-tail} uses a tail that also preserves them.

\subsection{Operator Contracts}
\label{sec:operator-contracts}
\label{sec:cost-operator-contracts}

For a cost $X$, let $\pre_H(X)=(\Pr(X=0),\ldots,\Pr(X=H))$ be its prefix, and write $\vec X=(X_1,\ldots,X_m)$ and $\vec Y=(Y_1,\ldots,Y_m)$ for tuples of costs.

\begin{definition}[Admissible operator]
\label{def:admissible-operator}
An operator $\omega$ of arity $m$ is \emph{admissible} if it provides:
\begin{description}[leftmargin=1.6em,labelindent=0pt]
  \item[(CSem)] \emph{Concrete semantics.} The distribution $\sem{\omega}(\vec X)$ of the result for independent operands $\vec X$.
  \item[(PComplete)] \emph{Prefix completeness.} A \emph{finite-prefix transformer} $\omega^{\#}_{H}$, so that $\pre_H$ is backward complete for $\sem{\omega}$~\cite{giacobazzi2000complete}: for all $\vec X$,
  $\pre_H(\sem{\omega}(\vec X))=\omega^{\#}_{H}(\pre_H(X_1),\ldots,\pre_H(X_m))$.
  \item[(DistErr, QueryErr)] \emph{Error contracts.} Constants $L_i,A_i,C_i\ge0$ such that, for all $\vec X,\vec Y$,
  \begin{align*}
    \text{(DistErr)}\quad&\dW(\sem{\omega}(\vec X),\sem{\omega}(\vec Y))\le\textstyle\sum_{i}L_i\,\dW(X_i,Y_i),\\
    \text{(QueryErr)}\quad&\bigl|\E[\sem{\omega}(\vec X)]-\E[\sem{\omega}(\vec Y)]\bigr|\le\textstyle\sum_{i}A_i\,\dW(X_i,Y_i)+\sum_{i}C_i\,|\E[X_i]-\E[Y_i]|.
  \end{align*}
  \item[(FinComp)] Finite computability. For summaries $s_1,\ldots,s_m$,
  the masses of $\sem{\omega}(s_1,\ldots,s_m)$ up to any given horizon, and its mean, are computable from the summaries.
\end{description}
\end{definition}

The first component, (CSem), defines what the operator means on distributions; the other three determine 
how the analysis can use the operator.
Prefix completeness (PComplete) keeps the stored prefixes exact (\Cref{lem:exact-prefix}),\footnote{
  \Cref{sec:extensions}
  describes the changes needed for operators that do not satisfy property (PComplete).
}
and the error contracts (DistErr) and (QueryErr) bound the shape and moment errors; together they make the analysis sound with respect to $\gamma_H$ (\Cref{thm:generic-soundness}).
(DistErr) and (QueryErr) are the Lipschitz conditions of \Cref{sec:introduction}, and $L_i$, $A_i$, $C_i$ are the coefficients of the operator's error-bound transformer (\Cref{sec:error-bounded-semantics}).
Although surrogates have infinite support, the abstraction $\beta_H$ and the loss bound \labelcref{eq:loss-bound} use only the result's masses up to $J$ and its mean, which (FinComp) makes computable.

The coefficient $L_i$ propagates shape error from operand $i$.
The coefficient $A_i$ converts shape error into mean error, and $C_i$ propagates mean error directly.
By (DistErr) and \labelcref{eq:mean-dist}, one can always take $A_i = L_i$ and $C_i = 0$; the coefficients in \Cref{tab:operator-contracts} are tighter where they differ.

\begin{table}[!tb]
\centering
\caption{The built-in operators: distribution and mean of the result $Z$, and error-bound coefficients of operand $X_i$ (for $\RepeatOp_N$, 
the single operand
$X$).
$G_X(z)=\sum_t\Pr(X=t)\,z^t$ 
is the generating function of $X$,
and $\Retry_a$ is the instance $N\sim\Geom_{\ge1}(a)$.}
\label{tab:operator-contracts}
\small
\setlength{\tabcolsep}{4pt}
\begin{tabular}{@{}lllccc@{}}
\toprule
Operator & Distribution of the result $Z$ & Mean $\E[Z]$ & $L_i$ & $A_i$ & $C_i$ \\
\midrule
$\Mix_w(X_1,\ldots,X_m)$ & $\Pr(Z=t)=\sum_i w_i\Pr(X_i=t)$ & $\sum_i w_i\,\E[X_i]$ & $w_i$ & $0$ & $w_i$ \\
$X_1+X_2$ & $\Pr(Z=t)=\sum_{u=0}^{t}\Pr(X_1=u)\Pr(X_2=t-u)$ & $\E[X_1]+\E[X_2]$ & $1$ & $0$ & $1$ \\
$\max(X_1,\ldots,X_m)$ & 
$\Pr(Z \leq t) = \prod_i \Pr(X_i \leq t)$  & $\sum_{t\ge0}S_Z(t)$ & $1$ & $1$ & $0$ \\
$\min(X_1,\ldots,X_m)$ & $S_Z(t)=\prod_i S_{X_i}(t)$ & $\sum_{t\ge0}S_Z(t)$ & $1$ & $1$ & $0$ \\
$\RepeatOp_N(X)$ & $G_Z(z)=G_N\bigl(G_X(z)\bigr)$ & $\E[N]\,\E[X]$ & $\E[N]$ & $0$ & $\E[N]$ \\
\bottomrule
\end{tabular}
\vspace{-2ex}
\end{table}

To apply an operator to summaries, the analysis evaluates the formulas of \Cref{tab:operator-contracts} for $t \le J$, using each summary's masses:
its prefix
up to $H$ and its geometric tail \labelcref{Eqn:Concretization} beyond.
For $\RepeatOp_N$, the masses up to $J$ are the first $J+1$ coefficients of $G_N(G_X(z))$;
for $\Retry_a$ with positive attempt costs, the recurrence \labelcref{Eq:RetryRecurrence} computes them.
For $\max$ and $\min$, the mean $\sum_{t \ge 0} S_Z(t)$ needs $S_Z(t)$ for all $t$.
Beyond $H$, each operand's survival function is geometric:
$S_{s_i}(t) = \rho_i \lambda_i^{t-H}$ for $t \ge H$.
For $\min$, $S_Z(t) = \prod_i \rho_i \lambda_i^{t-H}$ is geometric as well, and its sum
over $t \ge H$ is $\prod_i \rho_i / (1 - \prod_i \lambda_i)$.
For $\max$, expanding $S_Z(t) = 1 - \prod_i (1 - \rho_i \lambda_i^{t-H})$ gives a signed
sum, over nonempty sets $I$ of operands, of the geometric terms
$(-1)^{|I|+1} \prod_{i \in I} \rho_i \lambda_i^{t-H}$, each of which sums to
$(-1)^{|I|+1} \prod_{i \in I} \rho_i / (1 - \prod_{i \in I} \lambda_i)$ over $t \ge H$.
All five operators satisfy (PComplete), because their masses up to $H$ use only the operands' masses up to $H$; their $L_i$ follow from couplings of the operands.

\mypara{Complexity.}
\label{sec:complexity}
At a node of arity $m$, a $\Mix$, $\max$, or $\min$ takes $O(mJ)$ operations for the masses, and the tail mean of a $\max$ sums $2^m-1$ geometric series; a sum needs a convolution of length $J$, and a $\Retry$ needs the recurrence \eqref{Eq:RetryRecurrence}, both $O(J^2)$ when evaluated directly.
The fast Fourier transform computes the convolution in $O(J\log J)$ operations, and because the generating function of a retried cost is $q(z)=a\,b(z)/\bigl(1-(1-a)\,b(z)\bigr)$, where $b$ is that of one attempt, power-series inversion gives its first $J+1$ coefficients in $O(J\log J)$ as well.

\begin{example}[Geometric retry]
For $\Retry_a$, the last row of \Cref{tab:operator-contracts} gives $\E[\Retry_a(X)]=\E[X]/a$ and $\dW(\Retry_a(X),\Retry_a(Y))\le\dW(X,Y)/a$, so a low success probability multiplies both the expected cost and the propagated shape error.
For attempts with positive costs, the recurrence \labelcref{Eq:RetryRecurrence} is the finite-prefix transformer of $\Retry_a$.
\end{example}

\begin{example}[A new operator]
\label{ex:new-operator}
The quorum $\mathrm{Quorum}_k(X_1,\ldots,X_m)$ is the time at which $k$ of $m$ parallel subcomputations have finished, as in an erasure-coded read that needs any $k$ of $m$ fragments; 
$\mathrm{Quorum}_1$ is $\min$ and $\mathrm{Quorum}_m$ is $\max$.
At time $t$, the number of finished subcomputations is a sum of independent Bernoulli variables with success probabilities $1-S_{X_i}(t)$, which has a Poisson binomial distribution.
The survival function of the result at $t$ is the probability that this number is below $k$;
this probability defines the semantics (CSem).
The result's survival value at $t$ depends on the operands only through their survival values at $t$,
so a dynamic program that computes this Poisson binomial distribution for each $t\le H$ is a finite-prefix transformer (PComplete).
The $k^{\textit{th}}$ smallest of $m$ numbers changes by at most $d$ when one of them changes by $d$, so coupling each $X_i$ with $Y_i$ optimally and independently 
establishes
(DistErr) and (QueryErr) with $L_i=A_i=1$ and $C_i=0$, as for $\max$ and $\min$.
Beyond $H$, the result's survival function is a polynomial in geometric sequences, so (FinComp) holds.
\end{example}

\subsection{Error-Bounded Abstract Semantics}
\label{sec:error-bounded-semantics}
\label{sec:exact-prefix}

The analysis derives judgments $e\Down(s,\epsilon,\delta)$ bottom-up by two rules.
\begin{description}[leftmargin=1.6em,labelindent=0pt]
  \item[Atom.] For an atom, or a subtree solved exactly (\Cref{sec:refinement}), with distribution $X$ whose mean and masses up to $J$ are known, let $s=\beta_H(X)$, $\delta=0$, and $\epsilon=\eta$ for a loss bound $\eta\ge\dW(X,s)$.
  \item[Operator.] For $e=\omega(e_1,\ldots,e_m)$ with $e_i\Down(s_i,\epsilon_i,\delta_i)$, compute $\widetilde X=\sem{\omega}(s_1,\ldots,s_m)$ by (FinComp), and let $s=\beta_H(\widetilde X)$ with loss bound $\eta\ge\dW(\widetilde X,s)$. The \emph{error-bound transformer} of $\omega$ then gives
  \[
    \epsilon=\sum_{i=1}^{m}L_i\,\epsilon_i+\eta,
    \qquad
    \delta=\sum_{i=1}^{m}A_i\,\epsilon_i+\sum_{i=1}^{m}C_i\,\delta_i .
  \]
\end{description}
A derivation of $e\Down(s,\epsilon,\delta)$ asserts $\sem{e}\in\gamma_H(s,\epsilon,\delta)$ (\Cref{thm:generic-soundness}).

In the example from \Cref{sec:overview}, each internal node has the form $\Retry_a(\max(e_L, e_R))$, with local losses $\eta_1$ after the $\max$ and $\eta_2$ after the $\Retry_a$. Applying the Operator rule to both, with the coefficients of \Cref{tab:operator-contracts}, gives $\epsilon=(\epsilon_L+\epsilon_R+\eta_1)/a+\eta_2$ and $\delta=(\epsilon_L+\epsilon_R)/a$.
The values of $\epsilon$ and $\delta$ in \Cref{tab:overview-first-step} follow from these formulas;
e.g., at the root $r$, $\delta_r = (\epsilon_{v_L} + \epsilon_{v_R})/a \approx 0.00906$.
A low success probability $a$ thus amplifies upstream losses, which motivates the refinement methods discussed in \Cref{sec:refinement}.

\begin{lemma}[Exact prefixes]
\label{lem:exact-prefix}
Let every operator in $e$ satisfy (PComplete).
If $e\Down(s,\epsilon,\delta)$, then $\pre_H(s)=\pre_H(\sem{e})$.
\end{lemma}
The proof is by induction on the structure of the derivation; see \Cref*{app:method:lem:exact-prefix}.

\begin{theorem}[Soundness]
\label{thm:generic-soundness}
Let every operator in $e$ be admissible.
If $e\Down(s,\epsilon,\delta)$, then $\sem{e}\in\gamma_H(s,\epsilon,\delta)$, that is, $\pre_H(s)=\pre_H(\sem{e})$, $\dW(\sem{e},s)\le\epsilon$, and $|\E[\sem{e}]-\E[s]|\le\delta$.
\end{theorem}
The prefix condition is \Cref{lem:exact-prefix}; the two bounds follow by induction on the derivation, using only (DistErr) and (QueryErr).
The proof is in \Cref*{app:method:thm:generic-soundness}.

Thus (PComplete) confines all approximation to the tails.
Exact prefixes also answer a second kind of query: the probability $\Pr(\sem{e} > D)$ that a run misses a deadline $D$, which \Cref{sec:exp-other} evaluates on fork--join tasks.
For $D \leq H$, the miss probability is determined by the exact prefix, so $S_s(D) = S_{\sem{e}}(D)$.
For $D > H$, $|S_{\sem{e}}(D) - S_s(D)| \leq \dW(\sem{e}, s) \leq \epsilon$ by \Cref{thm:generic-soundness}, because a single term of the sum defining $\dW$ is at most the whole sum.
Hence $S_s(D) \pm \epsilon$ is a sound interval for the miss probability.

\section{Moment-Parametric Analysis}
\label{sec:moment-method}

This section lifts the analysis to raw moments up to order $k$: summaries and the bottom-up semantics keep their form, the bounds become vectors, and the distance is weighted by the queried moment.

\subsection{Moment Domains and Weighted Survival Distances}
\label{sec:moment-domain}

For a fixed order $k$, the concrete domain
consists of the cost distributions
with finite $k^{\textit{th}}$ moment, which 
for a nonnegative cost implies finite moments of all lower orders; we write $M_j(X) = \E[X^j]$ for $j \leq k$.
We assume that atoms have finite $k^{\textit{th}}$ moments and that every $\RepeatOp_N$ has $\E[N^k] < \infty$, which holds for geometric retry; then every operator preserves finiteness of the $k^{\textit{th}}$ moment, so $\sem{e}$ lies in the concrete domain for every cost expression $e$.
The tail-sum formula of \Cref{sec:distance} generalizes to $M_j(X)=\sum_{t\ge0}w_j(t)\,S_X(t)$ for $j\ge1$, with weight $w_j(t)=(t+1)^j-t^j$, so we use the weighted survival distance
\[
  d_j(X,Y)=\sum_{t\ge0}w_j(t)\,|S_X(t)-S_Y(t)|,
  \qquad\text{with}\qquad
  |M_j(X)-M_j(Y)|\le d_j(X,Y).
\]
The case $j=1$ is $\dW$.
Higher orders weigh tail mismatch more, because $w_j(t)$ grows as $jt^{j-1}$.

\subsection{Moment-Parametric Prefix--Tail Summaries}
\label{sec:moment-tail}

A moment-parametric summary $s=(f_0,\ldots,f_H,\theta)$ replaces $\lambda$ by the parameters $\theta$ of a tail family $\tau_\theta$ on $\N$, with $\Pr(s=t)=\rho\,\tau_\theta(t-H-1)$ for $t>H$.
The representation function $\beta_H$ 
retains the prefix of $X$ 
and chooses $\theta$ so that $\tau_\theta$ matches the first $k$ moments of the residual tail of $X$, the conditional distribution of $X-H-1$ given $X>H$; 
this approach preserves all raw moments up to order $k$.

\tool uses the maximum-entropy distribution on $\N$ with prescribed first $k$ moments, which has the form $\tau_\theta(x) \propto \exp\bigl(\sum_{i=1}^{k}\theta_ix^i\bigr)$~\cite{jaynes1957maxent}.
For $k=1$, $\tau_\theta$ is geometric with $\lambda=e^{\theta_1}$, so the expectation analysis of \Cref{sec:method} is the case $k = 1$.
For $k = 2$, \tool fits the quadratic member by damped Newton iteration, summing the masses and moments of $\tau_\theta$ up to a cutoff with a geometric bound on the remainder.
When the moments lie outside the quadratic family's range, the geometric tail is the fallback.
Only inexact fits and the geometric fallback introduce moment error, an exception to the second observation of \Cref{sec:introduction};
the analysis adds a bias bound $\zeta_j\ge|M_j(X)-M_j(\beta_H(X))|$ to $\delta_j$.

\subsection{Moment-Aware Abstract Semantics and Operator Contracts}
\label{sec:moment-semantics}
\label{sec:higher-operator-contracts}

A judgment now has the form
$e\Down(s,\vec\epsilon,\vec\delta)$ with vectors $\vec\epsilon=(\epsilon_1,\ldots,\epsilon_k)$ and $\vec\delta=(\delta_1,\ldots,\delta_k)$; it asserts $d_j(\sem{e},s)\le\epsilon_j$ and $|M_j(\sem{e})-M_j(s)|\le\delta_j$ for $j\le k$.
The \emph{order-$k$ contract} requires that, for each $j\le k$, an operator bound the $d_j$ distance and the $M_j$ error of its output by a nonnegative combination of the inputs' $d_l$ distances and $M_l$ errors, for $l\le j$, with coefficients that may depend on enclosures of the inputs' lower moments.
The rules of \Cref{sec:error-bounded-semantics} carry over, with the order-$k$ contract in place of (DistErr) and (QueryErr);
the analysis adds each local loss $\eta_j\ge d_j(\widetilde X,s)$ to $\epsilon_j$ and the bias bound $\zeta_j$ to $\delta_j$, and it bounds $\eta_j$ as in \eqref{eq:loss-bound}, with weights $w_j(t)$ in the sums and the moments $M_j$ in place of the means.

\mypara{Pointwise operators.}
$\Mix$, $\max$, and $\min$ act pointwise on survival functions, so their bounds survive the weights $w_j(t)$ and keep the coefficients of \Cref{tab:operator-contracts} at every order.
\begin{proposition}[Pointwise operators]
\label{prop:order-uniform}
For independent inputs and $j\ge1$,
$d_j(\Mix_w(\vec X),\Mix_w(\vec Y))\le\sum_i w_i\,d_j(X_i,Y_i)$ and
$d_j(\max(X,Z),\max(Y,Z))\le d_j(X,Y)$, and likewise for $\min$.
\end{proposition}
The proof is in \Cref*{app:moment:prop:order-uniform}.

\mypara{Convolutional operators.}
Sums and random repetitions are convolutional.
They are simple for the mean, but for $j\ge2$ their moment identities and distances mix orders, as in $M_2(X+Y)=M_2(X)+2M_1(X)M_1(Y)+M_2(Y)$.

\begin{proposition}[Convolutional operators]
\label{prop:conv-dist}
Let $Z$ be independent of $X$ and $Y$, and let $M_0(Z)=1$.
For every $j\ge1$,
\[
  d_j(X+Z,Y+Z)\le\sum_{l=1}^{j}\binom{j}{l}M_{j-l}(Z)\,d_l(X,Y).
\]
If $N$ is independent of the copies of $X$ and $Y$, then
$d_2(\RepeatOp_N(X),\RepeatOp_N(Y))\le\E[N]\,d_2(X,Y)+\E[N(N-1)]\,\bigl(M_1(X)+M_1(Y)\bigr)\,d_1(X,Y)$.
\end{proposition}
\noindent
The proof is in \Cref*{app:moment:prop:conv-dist}.
The analysis evaluates these coefficients by interval arithmetic on the children's intervals $M_l(s)\pm\delta_l$.

\begin{theorem}[Soundness for raw moments]
\label{thm:moment-soundness}
Let every operator in $e$ satisfy the order-$k$ contract.
If $e\Down(s,\vec\epsilon,\vec\delta)$, then $d_j(\sem{e},s)\le\epsilon_j$ and $|M_j(\sem{e})-M_j(s)|\le\delta_j$ for every $j\le k$.
\end{theorem}

\Cref{prop:order-uniform,prop:conv-dist} and the moment identities show that all built-in operators satisfy the order-$k$ contract when $k\le2$, and that mixture, maximum, minimum, and sum satisfy it at every $k$.
The proof of \Cref{thm:moment-soundness} is the vectorized form of the induction for \Cref{thm:generic-soundness} and is in \Cref*{app:moment:thm:moment-soundness}.
Prefix completeness does not depend on the query, so \Cref{lem:exact-prefix} applies to all orders at once.

\mypara{Variance.}
\label{sec:variance-query}
The $k=2$ analysis bounds the variance: if $M_1\in[l_1,u_1]$ and $M_2\in[l_2,u_2]$ with $l_1\ge0$, then $M_2-M_1^2\in[\max(0,\,l_2-u_1^2),\ u_2-l_1^2]$, and $u_2-l_1^2\ge0$ because $u_2 \ge M_2 \ge M_1^2 \ge l_1^2$.

\section{Tighter Intervals and Better Tree Shapes}
\label{sec:refinement}

A run of the analyzer returns an interval for one tree and a given query.
This section addresses two problems: how to spend additional computation to tighten that interval (\Cref{sec:refinement-profiles}), and how to use the analysis to search for a better organization of the same components (\Cref{sec:design-search}).

\subsection{Precision Profiles and Greedy Refinement}
\label{sec:refinement-profiles}
\label{sec:error-amplification}

Refinement tightens the interval by changing the \emph{precision profile} $(H,\mathcal O)$
of \Cref{def:err-bounded-summarization}.
Each subtree $f\in\mathcal O$ is solved by an exact solver, such as the Markov-chain model of Shchukin et al.~\cite{repeater-paper}, and its exact distribution $\sem{f}$ is used as an atom (rule \textsc{Atom} of \Cref{sec:error-bounded-semantics});
all losses inside $f$ disappear, and only the loss
of abstracting $\sem{f}$ at the root of $f$ remains.
Refinement may increase $H$ in a 
given
subtree; \Cref{thm:generic-soundness} holds at each node $v$ for $\gamma_{H_v}$ if $H_v$ is at most its children's, because (PComplete) at $v$ uses their prefixes only up to $H_v$.

The error-bound recurrence of \Cref{sec:error-bounded-semantics} can be unrolled into a root-level attribution.
Call each abstraction step of a derivation an abstraction site: every atom, every subtree in $\mathcal O$, and every operator node is one.
Each internal node of the tree in \Cref{fig:overview-numeric-repeater} has the form $\Retry_a(\max(e_L, e_R))$, and so contributes two sites.
Because the recurrence for $\epsilon$ is linear, unrolling it from the root gives
\begin{equation}
\label{eq:amp-generic}
  \epsilon_{\mathrm{root}}=\sum_{\ell}W_\ell\,\eta_\ell,
\end{equation}
where $\ell$ ranges over the abstraction sites, $\eta_\ell$ is the loss bound recorded at $\ell$, and the \emph{amplification} $W_\ell$ is the product of the constants $L_i$ of the operators applied after $\ell$ on the path to the root.
At each internal node of \Cref{fig:overview-numeric-repeater}, $\max$ contributes a factor of $1$ and $\Retry_a$ a factor of $1/a$ to $W_\ell$,
so losses on paths with low success probabilities are amplified most strongly.

\begin{proposition}[Monotonicity]
\label{prop:monotonicity}
For a fixed tree, set $\mathcal O$, and
coefficients $L_i$, $A_i$, $C_i$ of the error-bound transformers,
$\epsilon_{\mathrm{root}}$ and $\delta_{\mathrm{root}}$ are nondecreasing in the local loss bounds $\eta_\ell$.
\end{proposition}

\begin{proposition}[Convergence]
\label{prop:convergence}
Fix $e$ and $\mathcal O$, and let every operator in $e$ satisfy (PComplete), (DistErr), and (QueryErr).
Then, in exact arithmetic and with \eqref{eq:loss-bound} for any $J\ge H$, $\epsilon_{\mathrm{root}}\to0$ and $\delta_{\mathrm{root}}\to0$ as $H\to\infty$.
\end{proposition}
The proofs of Propositions \ref{prop:monotonicity} and \ref{prop:convergence} are in \Cref*{app:refinement:prop:convergence}.

\mypara{Cost model and greedy refinement.}
\label{sec:greedy-refinement}
\label{sec:refinement-cost}
A cost model assigns
to
each subtree in $\mathcal O$ the cost of its exact solution, estimated from its state count, plus that of abstracting its result, and 
to
each remaining node the cost of its abstract computation (\Cref{sec:complexity}).
Given a budget $B$ on the total cost, refinement reduces
the root bound that the query uses: $\delta_{\mathrm{root}}$ for the mean, or $\epsilon_{\mathrm{root}}$ for a deadline-miss probability (\Cref{sec:error-bounded-semantics}).
\tool{} refines greedily, starting from an initial profile such as $(H,\emptyset)$.
Each iteration considers solving one more subtree exactly or increasing the prefix horizon, everywhere or within one subtree, and accepts the candidate with the largest reduction of the bound per unit of added cost within the budget.
Because such changes alter the summaries, refinement reanalyzes each candidate, so the bound it returns is sound.

\subsection{Searching for Trees with Smaller Expected Cost}
\label{sec:design-feedback}
\label{sec:design-search}

For repeater chains, the tree shape is itself a design choice: 
because a failed swap regenerates the whole subtree below it, different order-preserving bracketings of the same links yield different expected completion times, even with homogeneous parameters and among AVL-balanced trees (\Cref{sec:exp-design}).
\Cref{alg:search} performs a beam search over such bracketings, keeping the atoms and the parameter $a$ of every $\Retry_a$ node.
Each round rewrites every tree in the beam at every internal node, analyzes up to $b$ new candidates at a small horizon $H_s$, and keeps the $w$ analyzed trees with the smallest estimated means.
Candidates reuse the summaries of unchanged subtrees.


\begin{algorithm}[!tb]
\caption{Analysis-guided beam search over repeater trees}
\label{alg:search}
\footnotesize
\algtext*{EndWhile}
\begin{algorithmic}[1]
\Require initial tree $e_0$; set $\mathcal B$ of baseline trees to improve on; search and final horizons $H_s,H_f$; beam width $w$; batch size $b$; maximum number $T$ of candidates to analyze at $H_s$ besides $e_0$
\State analyze $e_0$ at $H_s$; $\mathit{beam}\gets\{e_0\}$; $\mathit{seen}\gets\{e_0\}$
\State $U\gets R(\mathit{beam})\setminus\mathit{seen}$, where $R(\mathit{beam})$ lists the order-preserving rewrites (rotations, split moves, subtree rebuilds) of the trees in $\mathit{beam}$ at every internal node, root first
\While{$U\neq\emptyset$ \textbf{and} $|\mathit{seen}|\le T$}
  \State analyze the first $\min(b,\,T+1-|\mathit{seen}|)$ trees of $U$ at $H_s$; add them to $\mathit{seen}$
  \State $\mathit{beam}\gets$ the $w$ trees in $\mathit{seen}$ with the smallest estimated means
  \State $U\gets R(\mathit{beam})\setminus\mathit{seen}$
\EndWhile
\State reanalyze the tree $e^\ast\in\mathit{seen}$ with the smallest estimated mean and each $e_b\in\mathcal B$ at $H_f$
\State report that $e^\ast$ is an improvement only if, for each $e_b\in\mathcal B$, $\widehat K(e^\ast)+\delta(e^\ast)<\widehat K(e_b)-\delta(e_b)$
\end{algorithmic}
\end{algorithm}

The ranking by estimated means at $H_s$ only steers the search; because the intervals at $H_f$ contain the true means, every reported improvement holds for the exact expected costs.

\section{Evaluation}
\label{sec:evaluation}
\label{sec:experiment}

We implemented our framework in a tool called \tool and evaluated it on quantum-repeater completion times, RFID collision resolution, and fork--join response times.
Quantum repeaters~\cite{repeater-paper} stress the analysis hardest: each $\Retry_a$ node multiplies the error bound of its subtree by $1/a$, where $a\in(0,1]$ is the swap-success probability, so the amplification grows as $(1/a)^d$ along a root-to-leaf path with $d$ $\Retry$ nodes.
The experiments address three questions:
\begin{description}
\setlength{\itemsep}{0pt}\setlength{\parsep}{0pt}\setlength{\topsep}{2pt}
  \item[\textsc{RQ1}] How well does \tool scale on homogeneous and heterogeneous repeater chains, compared with exact solvers, hand-derived approximations, and truncated propagation? (\Cref{sec:exp-chains})
  \item[\textsc{RQ2}] How effectively does the analysis guide refinement and the search for better trees? (\Cref{sec:exp-design})
  \item[\textsc{RQ3}] How well does \tool bound second moments, deadline misses, and RFID rounds? (\Cref{sec:exp-other})
\end{description}

\subsection{Experimental Setup}
\label{sec:exp-setup}
\label{subsec:exp-setup}
\label{sec:instantiations}
\label{sec:front-ends}
\label{sec:general-generator-constructions}

\mypara{\textbf{Quantum repeaters.}}
A chain of $n$ elementary links is described by a generator parameterized by a \emph{split function} $c$, which gives each interval $i<j$ a split point $c(i,j)$ with $i\le c(i,j)<j$:
\[
  R(i,i)=\Atom(\Geom_{\ge1}(p_i)),
  \qquad  R(i,j)=\Retry_{a}\bigl(\max(R(i,c),\,R(c+1,j))\bigr),
\]
where $p_i$ is the probability that link $i$ is generated in one time step, $a$ is the success probability of the swap that joins the two halves (\Cref{sec:overview-example}), and 
the split function determines the tree.
\emph{Balanced trees} split every interval of $m$ links into its first $\lfloor m/2\rfloor$ and its last $\lceil m/2\rceil$ links (the doubling scheme of Shchukin et al.~\cite{repeater-paper} for a power of two), \emph{AVL trees} require only that sibling subtrees differ in height by at most one, and \emph{pair-and-carry trees} repeatedly join adjacent pairs from left to right and carry an unmatched subtree to the next round.
The links are homogeneous, $p_i=p$, as in the analyses of Shchukin et al., or heterogeneous, $p_i=p\,U_i$ with independent $U_i$ uniform in $[0.6,1.4]$ ($[0.3,1.7]$ in \Cref{tab:design-search-and-moments}(a)), three draws per size,
with reported values the medians over the draws;
by default, the links
are heterogeneous.
We use $J=4H$ in \eqref{eq:loss-bound}.

\mypara{Baselines.}
The \emph{exact solver} solves the Markov-chain model of Shchukin et al.~\cite{repeater-paper}, whose number of states doubles with every link.
Their approximation \emph{nested\_4} solves sub-chains of $4$ links exactly and combines them through effective link-success probabilities.
For homogeneous balanced trees with $n=2^d$ links, the classical \emph{doubling formula} $(3/(2a))^d/p$~\cite{sangouard2011quantum} approximates the mean when $p$ is small; Coopmans et al.~\cite{coopmans2022improved} show that it is in essence an upper bound.
\emph{Truncated distribution propagation}, which generalizes the algorithm of Brand et al.~\cite{brand2020waitingtime} to arbitrary trees, computes the masses of every subtree up to a truncation point, doubling it until the missing mass is below $10^{-10}$
and the relative change in the mean is below $10^{-7}$;
it reports no bound, and its converged mean is used as the reference
mean when computing relative errors.
In \Cref{fig:chain-scaling} and \Cref{tab:design-search-and-moments}(a), \tool uses its first-moment kernel with FFT-based convolution (\Cref{sec:complexity}), and truncated propagation calls the same routine, so differences in time come from the algorithms.
The intervals do not account for floating-point rounding.
Up to 4096 links, we found that extended-precision convolution changes $\widehat{K}$ and $\delta$ by less than $10^{-9}\widehat{K}$, far below the reported $\delta$.

We do not compare against Monte Carlo simulation. Simulation yields confidence intervals rather than guaranteed bounds, and its cost grows as the inverse square of the target relative error. For a deadline-miss probability of $10^{-3}$, a relative error of 1\% requires about $10^7$ samples.

\mypara{\textbf{RFID collision resolution and fork--join tasks.}}
In RFID collision resolution, a reader identifies $n$ tags whose identifiers are uniformly random bit strings.
It broadcasts an identifier prefix, and the tags whose identifiers begin with that prefix reply~\cite{yan2015memoryless}.
When several tags reply, their replies collide.
The reader then splits them by the next identifier bit and resolves the two parts concurrently, so the queried prefixes form a binary trie: the tree of identifier prefixes.
Over a noisy channel, the reply of a single tag is decoded with probability $a=0.8$, so the reader repeats the query until it succeeds.
The number of query rounds is
\[
  T(0)=1,
  \quad
  T(1)=\Atom(\Geom_{\ge1}(a)),
  \quad
  T(n)=\RepeatOp_{N_n}(1)+1+\Mix_{w}(B_1,\ldots,B_{n-1})\;\;(n\ge2),
\]
where $B_k=\max(T(k),T(n-k))$ resolves a split that sends $k$ tags to the first part;
$w = (w_1,\ldots,w_{n-1})$, with
$w_k=\binom{n}{k}2^{-n}/(1-\sigma_n)$ and $\sigma_n=2^{1-n}$; and $N_n$, with $\Pr(N_n=m)=(1-\sigma_n)\,\sigma_n^m$ for $m\ge0$, counts the unproductive splits that send all tags to one part.
Without channel errors ($a=1$), the number of rounds is $1+H_n$, where the height $H_n$ of a random binary trie has the classical distribution $\Pr(H_n\le k)=\prod_{i=1}^{n-1}(1-i/2^k)$~\cite{flajolet1983performance}.
Fork--join tasks~\cite{lakshmanan2010forkjoin,saifullah2013parallel} are series--parallel task trees $F(16)$ ($F(8)$ to $F(64)$ for the moments), where $F(1)=\Atom(V)$ and $F(\ell)=\textit{op}(F(c),F(\ell-c))$ for $1\le c<\ell$; 
the leaf time $V$ is right-skewed~\cite{davis2019survey}: $1$ with probability $0.99$ and $31+\Geom_{\ge1}(0.055)$ otherwise.
We sample the trees by drawing $c$ uniformly at each node and $\textit{op}$ as $+$, $\max$, or $\Mix_{(u,1-u)}$
with probabilities $0.20$, $0.75$, and $0.05$, where $u$ is uniform in $[0.3,0.7]$.

For fork--join tasks, the baselines
(for higher moments and deadline-miss probabilities, respectively)
are (i) an interval analysis that propagates enclosures of the first two moments using only the operands' moment enclosures (for $\max$, $\max(\E[X],\E[Y])\le\E[\max(X,Y)]\le\E[X]+\E[Y]$), and (ii) the Markov inequality applied to its upper bound on the mean.
The RFID baselines are the classical formula $\E[H_n]+1/a$ and a scalar analysis that replaces $\E[\max(X,Y)]$ by $\max(\E[X],\E[Y])$.

\mypara{\textbf{Metrics and platform.}}
The \emph{relative error} is $100\,|\widehat K-K|/K$ and the \emph{relative bound} is $100\,\delta/K$, where $K$ is the reference value of the queried quantity (by default, the mean), $\widehat K$ the estimate, and $\delta$ the reported bound.
\tool is implemented in Python~3.11 with NumPy and computes in double precision.
Each measurement runs in a fresh process on one core of an Intel Core i7-8700 CPU under Ubuntu~22.04, with a $4$\,GiB memory limit; times are medians of up to three runs.

\mypara{\textbf{Use of AI in the research.}}
LLMs (GPT-6 Astra and Claude Opus 5.5) were used to help implement \tool, conduct the experiments, and draw figures. 


\subsection{RQ1: Accuracy and Cost at Scale}
\label{sec:exp-chains}
\label{sec:exp-scaling}
\label{sec:exp-precision}
\label{subsec:exp-error-bounds}
\label{sec:exp-baselines}
\label{subsec:exp-baselines}

\begin{figure}[!tb]
\centering
\includegraphics[width=\linewidth]{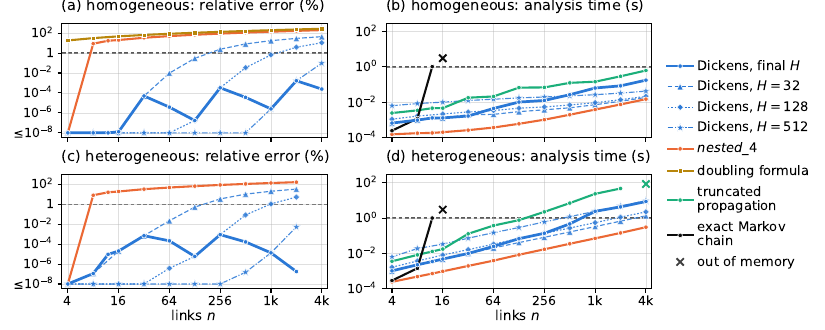}
\caption{
Balanced repeater chains with $4$ to $4096$ links ($a=0.8$, $p=0.37$): (a, c)~relative error; (b, d)~analysis time; dashed lines mark $1\%$ and $1$\,s.
(c) has no points at $4096$ links, where truncated propagation, which supplies the reference mean, runs out of memory.
``\tool, final $H$'' denotes \tool{} in adaptive mode.
The doubling formula only applies to homogeneous chains, so it is only plotted in (a).
}
\label{fig:chain-scaling}
\vspace{-2ex}
\end{figure}

\Cref{fig:chain-scaling} compares all methods on the same trees, from $4$ links (where \textit{nested\_4} is exact) to $4096$.
In ``adaptive mode,''
\tool doubles $H$ from $32$ until the interval at the root guarantees at most $1\%$ relative error, i.e., $\widehat K>\delta$ and $\delta/(\widehat K-\delta)\le1\%$, which implies a relative bound of at most $1\%$ because $K\ge\widehat K-\delta$.
The test uses only $\widehat K$ and $\delta$, so it needs no reference mean; times are for the run at the final $H$.
On heterogeneous chains,
adaptive
\tool meets the $1\%$ target on every tree, and up to $2048$ links its median relative error is at most $10^{-3}\%$.
Its final $H$ grows from $32$ to $1024$; at each fixed $H$, the error grows steeply with $n$, so the adaptive error in \Cref{fig:chain-scaling}(c) drops whenever the final $H$ doubles.
The exact solver runs out of memory at $16$ links, and \textit{nested\_4} errs by $7.8\%$ at $8$ links, $47\%$ at $64$, and more than $100\%$ from $512$ links on.
At $2048$ links, truncated propagation obtains an answer in $47$\,s and $2.0$\,GiB, and \tool in $4.3$\,s and $0.2$\,GiB; at $4096$ links, truncated propagation runs out of memory, and \tool takes $8.6$\,s and $0.2$\,GiB.

On homogeneous chains, for which the Shchukin et al.~\cite{repeater-paper} approximations were designed, \textit{nested\_4} errs about as much as on heterogeneous chains.
The doubling formula
has $18\%$ error already at $4$ links; \tool and truncated propagation analyze identical subtrees once, so both run faster, and \tool meets the target at every size, in $0.18$\,s at $4096$ links.

On heterogeneous chains, harder parameters shrink the chains that truncated propagation can handle, and with $a=0.5$ also those of \tool{}, but \tool{} reaches at least twice as far in every setting we tried.
With $p=0.1$, \tool{} meets the 1\% target up to 4096 links (23\,s), versus 1024 for truncated propagation.
With $a=0.5$, each $\Retry$ node multiplies the error bound of its subtree by~2;
\tool{} reaches 512 links ($p=0.37$) and 128 links ($p=0.1$), versus 256 and 64, respectively, for truncated propagation.

\mypara{Answer to RQ1.}
With $a = 0.8$, \tool{} bounds the mean within 1\% on every chain we tried, up to 4096 links.
The exact solver exceeds the memory limit at 16 links, and truncated propagation, which reports no bound, exceeds it at 4096 heterogeneous links.
The more accurate hand-derived approximation, \textit{nested\_4}, has errors above 100\% from 512 links on, and the doubling formula errs by 18\% already at 4 links.
With $a=0.5$, \tool{} still reaches twice as far as truncated propagation.

\begin{table}[!tb]
\centering
\scriptsize
\setlength{\tabcolsep}{4pt}
\caption{
(a) Performance of \Cref{alg:search} on heterogeneous chains ($a=0.8$, $p=0.3$, and $U_i$ uniform in $[0.3,1.7]$),
starting from the balanced tree, with up to $T = 1024$ candidates at $H_s = 32$.
Three chains per size; times are medians.
(b) Second-moment analysis of a homogeneous eight-link balanced repeater ($p=0.1$, $a=0.5$, four-link subtrees solved exactly) with the $M_1$-directed geometric tail, which preserves only $M_1$, and the $M_2$-directed quadratic tail, which preserves both $M_1$ and $M_2$.
Bold marks the better of each pair of relative errors and bounds.
}
\label{tab:design-search-and-moments}
\vspace{-3ex}
\begin{tabular}{@{\hspace{1.0ex}}c@{\hspace{4.5ex}}c@{\hspace{1.0ex}}}
\begin{tabular}[t]{@{}rlr@{}}
\toprule
\multicolumn{1}{@{}l}{(a)}  & Improvement vs.    & Search \\
                            & balanced tree (\%) & time \\
                        $n$ & median [min,\,max] & (seconds) \\
\midrule 
32 & 8.25 [3.94,\,10.3] & 0.48 \\
64 & 8.52 [7.09,\,8.84] & 0.95 \\
128 & 5.36 [3.58,\,5.65] & 1.36 \\
256 & 1.95 [1.78,\,2.23] & 2.06 \\
\bottomrule
\end{tabular}

  &
\begin{tabular}[t]{r rr rr rr rr}
\toprule
\multicolumn{1}{@{}l}{(b)} & \multicolumn{2}{c}{$M_1$ err.\ (\%)} & \multicolumn{2}{c}{$M_2$ err.\ (\%)} & \multicolumn{2}{c}{$M_2$ bound (\%)} & \multicolumn{2}{c}{Var.\ err.\ (\%)} \\
\cmidrule(lr){2-3}\cmidrule(lr){4-5}\cmidrule(lr){6-7}\cmidrule(lr){8-9}
$H$ & $M_1$-dir. & $M_2$-dir. & $M_1$-dir. & $M_2$-dir. & $M_1$-dir. & $M_2$-dir. & $M_1$-dir. & $M_2$-dir. \\
\midrule
4 & 3.2 & \textbf{1.2} & 21.7 & \textbf{1.2} & 686 & \textbf{662} & 42.7 & \textbf{0.486} \\
16 & 0.702 & \textbf{0.293} & 10.5 & \textbf{0.321} & 305 & \textbf{292} & 23.2 & \textbf{0.0504} \\
64 & 0.00157 & \textbf{0.000689} & 1.14 & \textbf{0.000923} & 24.3 & \textbf{22.2} & 2.71 & \textbf{0.000289} \\
128 & $<10^{-6}$ & $<10^{-6}$ & 0.0824 & {\boldmath$<10^{-6}$} & 1.27 & \textbf{0.889} & 0.197 & {\boldmath$<10^{-6}$} \\
\bottomrule
\end{tabular}

\end{tabular}
\vspace{-1ex}
\end{table}

\subsection{RQ2: Refinement and Searching for Trees with Smaller Expected Cost}
\label{sec:exp-design}
\label{subsec:exp-tree-design}
\label{sec:exp-refinement}
\label{subsec:exp-trajectories}

In this RQ, we focus on quantum repeaters.

\mypara{Refinement.}
On a homogeneous chain with ten links ($p=0.3$, $a=0.5$), the pair-and-carry tree $((((1,2),(3,4)),((5,6),(7,8))),(9,10))$ analyzed at $H=32$ with two-link subtrees solved exactly has root bound $\delta=47.38$.
With a budget of $1.25$ times the initial analysis cost, greedy refinement (\Cref{sec:greedy-refinement}) first solves the eight-link left subtree exactly, which lowers that cost, and then increases $H$ from $32$ to $64$ in that subtree, whose abstraction loss dominates $\delta$; this lowers $\delta$ to $5.845$.

\mypara{Beam search on the ten-link chain.}
For $n$ not a power of two, there is no canonical scheme, and ten links already admit $4862$ trees.
Starting from the same
ten-link
pair-and-carry tree
with mean $153.948\pm0.111$, \Cref{alg:search} with $H_s=32$, $w = 4$, $b=32$, $T=100$, and $H_f=128$
finds at its ninth candidate the tree $(((1,2),(3,4)),(((5,6),(7,8)),(9,10)))$ with mean $126.663\pm0.010$, a $17.7\%$ improvement;
no tree has a smaller mean (exhaustive check).
Another AVL tree, splitting the right subtree $3+3$, not $4+2$, has mean $127.627\pm0.007$, so height balance alone does not determine the best scheme.

\mypara{Longer heterogeneous chains.}
On chains with $32$ to $256$ links, the search starts from the balanced tree.
With $w=4$ and $b=64$, it finds trees with means $1.8\%$ to $10\%$ lower,
and at $H_f = 256$ the interval of every selected tree lies below that of the balanced tree (\Cref{tab:design-search-and-moments}(a)).

\mypara{Answer to RQ2.}
\tool{}'s analysis effectively guides both refinement and the search for better trees.
Refinement targets the subtree that dominates the root bound.
On the ten-link chain of \Cref{sec:exp-refinement}, refinement cut the root bound eightfold within a budget of $1.25\times$ the initial analysis cost.
The beam search finds better trees, including an optimal ten-link tree.

\subsection{RQ3: Generalizing Beyond Mean Completion Times}
\label{sec:exp-moments}
\label{subsec:exp-moment}
\label{sec:exp-other}
\label{subsec:exp-case-studies}
\label{subsec:exp-transfer}


\mypara{On quantum repeaters.}
\Cref{tab:design-search-and-moments}(b) compares the two tails on an eight-link balanced repeater.
At every $H$, the $M_2$-directed quadratic tail (\Cref{sec:moment-tail}) gives $M_2$ and variance errors up to five orders of magnitude lower than those of the $M_1$-directed geometric tail.
The tail choice changes the estimate far more than the bound:
$\delta_2$ is dominated by propagated shape error, which only a larger $H$ or solving subtrees exactly reduces, and the $M_2$-directed bound does not fall below 1\% until $H=128$.
Across the $12$ combinations of $n\in\{4,8\}$, $p\in\{0.1,0.37\}$, and $H\in\{16,32,64\}$ with $a=0.5$ on homogeneous chains (including rows $H=16$ and $H=64$ of \Cref{tab:design-search-and-moments}(b)), all 
solvable by
the exact solver, the $M_2$-directed analysis lowers the $M_2$ error in $11$ and ties in one.

\mypara{On fork--join tasks.}
With $8$ to $64$ leaves (eight sampled trees of each size), \tool{}, doubling $H$ until its mean interval is narrower than $1\%$ of the estimate, gives intervals for the mean and variance whose median widths are $0.013$--$0.11\%$ and $0.13$--$1.4\%$ of the exact values, against $139\%$--$287\%$ and $541\%$--$6456\%$ for the moment-only interval analysis.
See \Cref{app:evaluation:tab:forkjoin-moments}.

\begin{table}[!tb]
\centering
\scriptsize
\setlength{\tabcolsep}{3pt}
\caption{(a)~Six fork--join tasks ($16$ leaves, $H=64$) with deadlines $D$ at exact response-time quantiles ($0.90$ inside the prefix).
(b)~\tool solves populations of $\le4$ tags exactly; bold marks the smallest error per row.}
\label{tab:cases}
\vspace{-3ex}
\begin{tabular}{@{}c@{\hspace{4.5ex}}c@{}}
\begin{tabular}[t]{@{}lrrrrrrr@{}}
\toprule
\multicolumn{8}{@{}l@{}}{(a) Fork--join deadline-miss probability $\Pr(X>D)$} \\
\multicolumn{8}{@{}l@{}}{\phantom{(a) }(medians over six tasks)} \\
\cmidrule(r){1-8}
\multirow{3}{*}{Quantile} & \multicolumn{1}{l}{\multirow{3}{4em}{Truncated propagation}} & \multicolumn{2}{c}{\multirow{3}{*}{Markov}} & \multicolumn{4}{c@{}}{\tool} \\
\cmidrule(l){5-8}
 &  & \multicolumn{2}{c}{} & \multicolumn{2}{c}{$\frac{ \textrm{\tool upper}}{\textrm{trunc.\ prop.}}$} & \multicolumn{2}{c}{$\frac{\textrm{Markov}}{\textrm{\tool upper}}$} \\
\midrule
0.90 & $9.70\!\times\!10^{-2}$ & \multicolumn{2}{c}{0.578} & \multicolumn{2}{c}{1.00} & \multicolumn{2}{c@{}}{5.98} \\
0.99 & $9.74\!\times\!10^{-3}$ & \multicolumn{2}{c}{0.261} & \multicolumn{2}{c}{1.14} & \multicolumn{2}{c@{}}{22.8} \\
0.999 & $9.80\!\times\!10^{-4}$ & \multicolumn{2}{c}{0.171} & \multicolumn{2}{c}{2.49} & \multicolumn{2}{c@{}}{72.5} \\
0.9999 & $9.77\!\times\!10^{-5}$ & \multicolumn{2}{c}{0.128} & \multicolumn{2}{c}{15.8} & \multicolumn{2}{c@{}}{90.1} \\
\bottomrule
\end{tabular}
&
\begin{tabular}[t]{@{}lcrrrr@{}}
\toprule
\multicolumn{6}{@{}l@{}}{(b) RFID collision resolution with $a=0.8$} \\
\multicolumn{6}{@{}l@{}}{\phantom{(b) }(expected number of query rounds $K$ for $n$ tags)} \\
\cmidrule(r){1-6}
\multirow{4}{*}{$n$} & \multirow{4}{*}{\begin{tabular}{@{}c@{}}$K$ via\\truncated\\propagation\end{tabular}} & \multicolumn{2}{c}{Baseline} & \multicolumn{2}{c@{}}{\tool, $H{=}16$} \\
& & \multicolumn{2}{c}{rel.\ error (\%)} & rel.\ error & rel.\ bound \\
\cmidrule(lr){3-4}\cmidrule(l){5-6}
 &  & Formula & Scalar & (\%) & $\delta/K$ (\%) \\
\midrule
\phantom{0}8 & \phantom{0}7.8555 & 4.28 & 13.9 & {\boldmath$1.3\!\times\!10^{-11}$} & $5.6\!\times\!10^{-7}$ \\
16 & \phantom{0}9.9325 & 3.83 & 18.1 & {\boldmath$1.2\!\times\!10^{-9}$} & $6.8\!\times\!10^{-6}$ \\
32 & 11.9765 & 3.42 & 21.8 & {\boldmath$8.0\!\times\!10^{-8}$} & $8.2\!\times\!10^{-5}$ \\
64 & 14.0014 & 3.04 & 24.8 & {\boldmath$4.7\!\times\!10^{-6}$} & 0.00105 \\
\bottomrule
\end{tabular}
\end{tabular}
\vspace{-1ex}
\end{table}

We further study deadline miss probabilities.
In \Cref{tab:cases}(a), deadlines inside the stored prefix yield exact miss probabilities (to within $10^{-12}$), as \Cref{lem:exact-prefix} predicts.
For tail deadlines, all $18$ intervals $S_s(D)\pm\epsilon$ contain the exact value, with upper ends $23$--$90$ times below the Markov bounds.

\mypara{On RFID collision resolution.}
Without channel errors, \tool's root prefix agrees with the classical distribution of $1+H_n$ to within $10^{-14}$ for $8$ to $64$ tags.
With $a=0.8$ (\Cref{tab:cases}(b)), the classical formula and the scalar baseline
have relative errors of
$3\%$--$4\%$ and $14\%$--$25\%$,
respectively.
In contrast, \tool{} at $H=16$ and $J=2000$ has a relative error of at most $5\times10^{-6}\%$, and a relative bound of at most $0.0011\%$.

\mypara{Answer to RQ3.}
\tool{} obtains
tight bounds for RFID rounds, bounds for deadline misses much tighter than Markov bounds, and accurate second moments whose bounds tighten as $H$ grows.
Its mean and variance intervals on fork--join tasks are orders of magnitude narrower than those of a moment-only interval analysis.
Its deadline-miss probabilities are exact inside the prefix and, beyond it, 23--90 times tighter than Markov bounds; and its RFID round counts err by at most $5\times10^{-6}\%$, against 3\%--25\% for the baselines.

\section{Discussion and Extensions}
\label{sec:discussion}
\label{sec:extensions}
\label{sec:unbounded-extension}
\label{sec:threshold-extension}
\label{sec:kernel-limitations}
\label{sec:cutoff}

\mypara{Operators without prefix completeness~(property (PComplete)).}
Some operators have no finite-prefix transformer,
such as the idle time $|X - Y|$ of the faster of two branches that start together
(in a repeater, how long the first-ready link is stored before the swap): two large times can have a small gap, so the prefix of $|X-Y|$ depends on the tails of $X$ and $Y$.
Even without (PComplete), the rules of \Cref{sec:error-bounded-semantics} and the proof of the two bounds in \Cref{thm:generic-soundness} apply, so a derivation $e\Down(s,\epsilon,\delta)$ asserts $d_W(\sem{e},s)\le\epsilon$ and $|\E[\sem{e}]-\E[s]|\le\delta$.
Only exactness of the prefix is lost: the stored prefix is that of the surrogate, and deadlines $D\le H$ also receive the interval $S_s(D)\pm\epsilon$.
For the idle time,
the map $(x,y)\mapsto|x-y|$ is $1$-Lipschitz in each argument, so (DistErr) and (QueryErr) hold with $L_i=A_i=1$ and $C_i=0$.
Each mass of $\widetilde Z=|s_X-s_Y|$ is a sum of products of the children's masses, and the terms beyond $H$ pair two geometric tails and sum in closed form.
Because $|x-y|=x+y-2\min(x,y)$, the mean of $\widetilde Z$ follows from the mean of $\min$, so (FinComp) holds.
The rule abstracts $s=\beta_H(\widetilde Z)$, bounds its loss $\eta$ by \eqref{eq:loss-bound}, and sets $\epsilon=\epsilon_X+\epsilon_Y+\eta$ and $\delta=\epsilon_X+\epsilon_Y$.

\mypara{Unbounded recursive generators.}
A generator with unbounded recursion defines a distributional fixed point $\vec X=\vec F(\vec X)$, and the analysis would iterate the abstract transformer.
If the error-flow matrix $\mathbf L$ of one unfolding (path products of the $L_i$) has spectral radius below $1$, then $\vec F$ is a contraction, and a fixed point $\vec s=\beta_H(\vec F(\vec s))$ with local losses $\vec\eta$ satisfies $\dW(\vec X,\vec s)\le(I-\mathbf L)^{-1}\vec\eta$.

\mypara{Dependence and joint costs.}
If sibling costs share randomness, their summaries no longer determine $\max(X,Y)$ or $X+Y$, so the analysis would need joint summaries.
They would also cover several costs of one system, such as the completion time and fidelity of a repeater~\cite{brand2020waitingtime}.

\section{Related Work}
\label{sec:related}
\label{sec:related-work}

\mypara{Probabilistic program and resource analysis.}
Expected-cost analyses derive bounds through expectation transformers, types, potential functions, or martingales~\cite{kaminski2018weakest,ngo2018bounded,wang2020raising,avanzini2020modular,leutgeb2022probabilistic,batz2023amortized,hoffmann2012resource,kahn2025efficient}, some also for higher moments and tails~\cite{kura2019tail,wang2021central,moosbrugger2022polar}, and algebraic and abstract-interpretation frameworks compose local summaries through an operator algebra~\cite{wang2018pmaf,wang2024newtonian,monniaux2000abstract,bistarelli2001semiring,ahrens2025weighted}.
These analyses handle loops and state, but target additive costs.
Weighted automata and quantitative games handle $\max$ and $\min$ through global fixed points~\cite{vcerny2011quantitative,chatterjee2010expressiveness};
in contrast, our prefix--tail abstract domain provides compositional transformers for them, with sound error bounds.
Contract (DistErr) is a Wasserstein-1 Lipschitz property, as in expected-sensitivity analysis~\cite{barthe2018expected}.

\mypara{Distributional approximation and guaranteed bounds.}
Moment matching, maximum-entropy fitting, expectation propagation, and phase-type fitting yield finite surrogates~\cite{johnson1989matching,asmussen1996em,osogami2006ph,jaynes1957maxent,minka2013ep}, usually without error bounds; ours carry local loss bounds.
Closest to our domain are guaranteed-bound methods, which enclose a program's distribution between lower and upper bounds.
Beutner et al.~\cite{beutner2022guaranteed} partition the space of sampled values into interval traces and bound the posterior probability that the result lies in a given interval.
Zaiser et al.~\cite{zaiser2025guaranteed} extend guaranteed bounds to moments with eventually geometric distributions (EGDs), which also pair an exact initial block with geometric tails.
The two approaches are incomparable and differ in four ways.
First, an EGD is an upper bound on the unnormalized distribution, derived from a contraction invariant that requires a uniform decrease of the distribution; as they note, the bound therefore loses precision either for small values or in the tail.
Our tail need not dominate the true distribution: it is fitted to preserve the queried moments, and its error is bounded separately by $\epsilon$ and $\delta$.
Second, their language has no operator for $\max$ or $\min$ of two variables and no comparisons between variables.
As they note, such constructs can be encoded, but only by making programs more complex:
for example, $Z := \max(X, Y)$ becomes the loop
$Z := 0;\ \mathbf{while}\ \neg(X = 0 \wedge Y = 0)\ \{X \mathrel{\dot{-}}= 1;\ Y \mathrel{\dot{-}}= 1;\ Z \mathrel{+}= 1\}$,
which also sets $X$ and $Y$ to $0$.
Each $\max$ or $\min$ thus needs at least one loop, and each loop needs its own contraction invariant.
Max and min are where shape error turns into moment error, and our analysis treats them as primitive operators whose error-bound transformers are explicit formulas (\Cref{tab:operator-contracts}).
Third, their analysis solves a nonlinear constraint system for each program, which they note does not scale to very large programs.
Each of our transformers is applied locally, without constraint solving, which lets \tool{} analyze repeater trees with thousands of nodes.
Fourth, their approach is more general:
it bounds the joint distribution of all variables of a discrete probabilistic program with conditioning and general loops, whereas hierarchical cost expressions have a single cost, no conditioning, and no loops other than $\Retry_a$ and $\RepeatOp_N$.

\mypara{Recursive stochastic systems.}
Probabilistic pushdown systems and recursive Markov chains yield expectations and variances by solving global nonlinear systems~\cite{esparza2005ppda,etessami2009rmc,brazdil2015runtime}, and
model checkers analyze finite-state models exactly~\cite{prism2011,storm2017}.
Probabilistic recurrence analyses derive asymptotic or tail bounds, also for span recurrences with maxima~\cite{karp1994prr,tassarotti2017span,tassarotti2018verified,sun2023automated}, and the contraction method uses Lipschitz estimates in transport metrics for sum and max recursions~\cite{ruschendorf2006summax}.

\mypara{Quantum repeaters.}
The closest application literature studies quantum-repeater completion times and rates~\cite{briegel1998,dcz2001,sangouard2011quantum,bernardes2011rate}: Shchukin et al.~\cite{repeater-paper} give exact Markov-chain results for small systems and hand-derived approximations for larger homogeneous chains, and later work improves exact completion-time computation, models memory cutoffs, and optimizes repeater policies~\cite{brand2020waitingtime,kamin2023,goodenough2025,inesta2023,jiang2007,shchukin2022,haldar2024}.
These solvers complement our analysis, which can use them for small subtrees.

\mypara{Timing and collision-resolution analysis.}
Probabilistic response-time analysis treats parallel and DAG tasks~\cite{lakshmanan2010forkjoin,saifullah2013parallel,melani2015response,diaz2002stochastic,cucu2012measurement,davis2019survey}, and copulas and probability boxes enclose timing distributions under unknown dependence~\cite{bernat2005copulas,williamson1990parith,ferson2003pbox,markovic2021convolution}.
Tree-splitting analyses solve additive total-cost recurrences~\cite{capetanakis1979tree}, and with concurrent resolution the number of rounds is the height of a random trie, whose distribution is classical~\cite{flajolet1983performance,devroye1984height} and, with channel errors, known only asymptotically~\cite{broutin2008height}.


\bibliographystyle{ACM-Reference-Format}
\bibliography{src/main}

\dkSupplementStart
\section*{About the Appendices}

The appendices give the proofs, constructions, and experimental protocols behind the paper, and the complete tables of its experiments.
They use the notation of the paper, which Table~\ref*{main:tab:notation} of the paper summarizes.
References
without a letter, such as \Cref*{main:sec:method}, Table~\ref*{main:tab:operator-contracts}, or eq.~(\ref*{main:eq:loss-bound}), refer to the paper;
references
with a letter, such as \S A.2 or Table~D.1, refer to the appendices.
They follow the order of the paper:
\begin{itemize}
\item Appendix~\ref{app:method:sec:method}: the mean analysis of \Cref*{main:sec:background,main:sec:method}, with the proofs of \Cref{main:lem:exact-prefix} and \Cref{main:thm:generic-soundness}, the generators, and the three applications of \Cref*{main:sec:exp-setup};
\item Appendix~\ref{app:moment:sec:moment-method}: the moment analysis of \Cref*{main:sec:moment-method}, with the proofs of \Cref*{main:prop:order-uniform,main:prop:conv-dist};
\item Appendix~\ref{app:refinement:sec:refinement}: refinement and design search (\Cref*{main:sec:refinement}), with the proofs of \Cref*{main:prop:monotonicity,main:prop:convergence};
\item Appendix~\ref{app:evaluation:sec:evaluation}: the protocols and complete tables of \Cref*{main:sec:experiment}, in the order of its subsections;
\item Appendix~\ref{app:extension:sec:extensions}: the extensions of \Cref*{main:sec:extensions} and a sharper bound for the deadline queries of \Cref*{main:sec:error-bounded-semantics}.
\end{itemize}

\supplementcontents

\appendix
\section{Details of the Mean Analysis}
\label{app:method:sec:method}
\suppressfloats[t]

This appendix gives the details of the mean analysis of \Cref*{main:sec:method}, of the generators of \Cref*{main:sec:cost-expressions}, and of the three applications of \Cref*{main:sec:exp-setup}.

\begin{figure}[tb]
\centering
\begin{tikzpicture}
  \node[anchor=south west,inner sep=0] (img) {\includegraphics[width=0.82\linewidth]{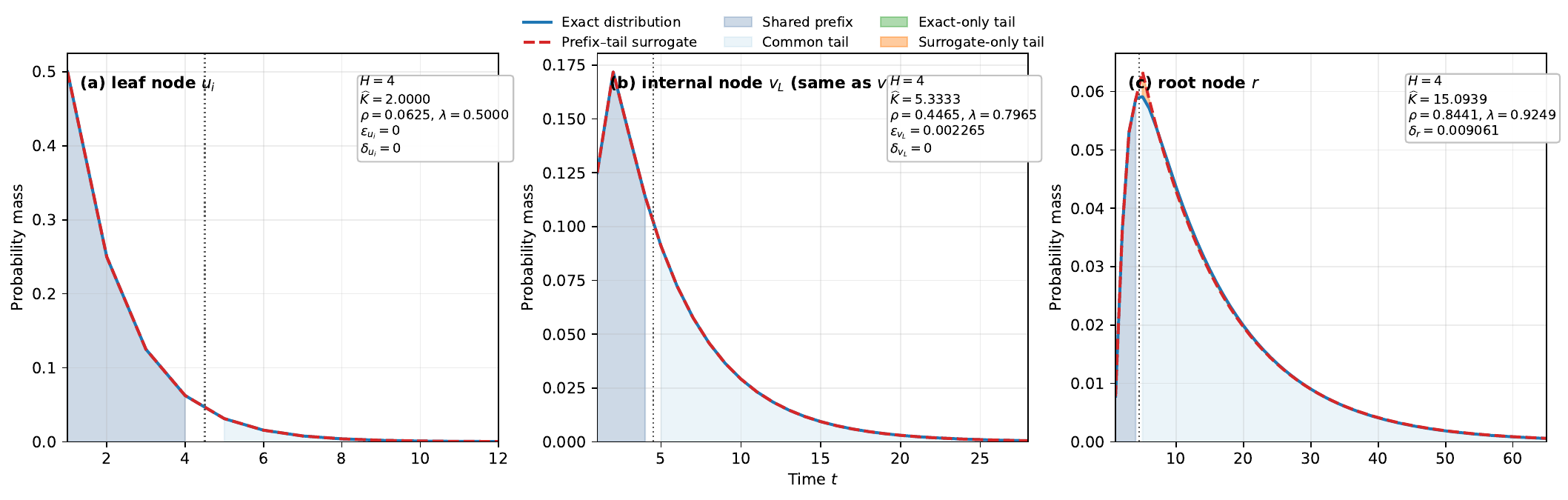}};
  \begin{scope}[x={(img.south east)},y={(img.north west)}]
    \fill[white] (0.5025,0.786) rectangle (0.5651,0.885);
  \end{scope}
\end{tikzpicture}
\caption{The running example of \Cref*{main:sec:overview} ($p=a=1/2$, $H=4$): exact distributions (solid) and prefix--tail surrogates (dashed) at a leaf $u_i$, at $v_L$, and at the root $r$; the shaded prefix is stored exactly.}
\label{app:method:fig:overview-node-surrogates}
\end{figure}

\paragraph{The running example.}
Figure~\ref{app:method:fig:overview-node-surrogates} plots the surrogates of Table~\ref*{main:tab:overview-first-step} against the exact completion-time distributions.
At the leaf, the surrogate equals the exact geometric distribution; at $v_L$ and at $r$, the masses up to $H=4$ are exact (\Cref*{main:lem:exact-prefix}), and the geometric tails stay close to the exact tails, which carry $45\%$ and $84\%$ of the mass, respectively.
The local losses in Table~\ref*{main:tab:overview-first-step} use the loss bound (\ref*{main:eq:loss-bound}) with $J=200$.

\subsection{Cost Expressions and Distances}
\label{app:method:sec:operators-contracts}

A cost expression $\omega(e_1,\ldots,e_m)$ is a finite tree: each operator application is an internal node, and each atom is a leaf.
\emph{Finite} constrains only the tree; the distribution it denotes may have infinite support, as for a geometric atom.
Expressions contain no variables, so every expression denotes a fully determined distribution $\sem{e}$.
Generators whose recursion need not terminate are treated in \Cref{app:extension:sec:unbounded-extension}.

For the mean, all costs live in
\[
  \mathcal{D}_1=\{X\mid X\text{ is an }\N\text{-valued random variable and }\E[X]<\infty\},
\]
which contains $\sem{e}$ whenever atoms have finite means and every $\RepeatOp_N$ has $\E[N]<\infty$.
Shape error is measured by the survival form of the Wasserstein-1 distance, $\dW(X,Y)=\sum_{t\ge0}|S_X(t)-\nobreak S_Y(t)|$ with $S_X(t)=\Pr(X>t)$.
On $\N$, this sum equals the cumulative form $\sum_{t\ge0}|\Pr(X\le t)-\nobreak\Pr(Y\le t)|$, because $|\Pr(X\le t)-\Pr(Y\le t)|=|S_X(t)-S_Y(t)|$; we use the survival form because it matches the tail-sum formula $\E[X]=\sum_{t\ge0}S_X(t)$, which gives $|\E[X]-\E[Y]|\le\dW(X,Y)$ (\cref*{main:eq:mean-dist}).
That inequality is often loose because it charges for every change of shape, including changes that leave the mean unchanged.
The analysis therefore also maintains the moment bound $\delta$.

\subsection{Prefix--Tail Summaries}
\label{app:method:sec:prefix-tail-domain}

A summary $s=(f_0,\ldots,f_H,\lambda)$ denotes the distribution of \cref*{main:Eqn:Concretization}, with tail mass $\rho=1-\sum_{t\le H}f_t$.

\paragraph{Abstraction preserves the mean.}
The representation function $\beta_H$ copies $f_t=\Pr(X=t)$ for $t\le H$, so $\rho=\Pr(X>H)$, and sets $\lambda=\nu_1/(1+\nu_1)$ with $\nu_1=\E[X-H-1\mid X>H]$, the mean of the residual tail, when $\rho>0$ (and $\lambda=0$ otherwise).
Because $\E[X\mid X>H]=H+1+\nu_1$, the tail mean is $\nu_1=\bigl(\E[X]-\sum_{t\le H}t\,f_t\bigr)/\rho-(H+1)$, computed from the mean of $X$ and its prefix.
When $\lambda=0$, the tail places all of $\rho$ at $H+1$; \cref*{main:Eqn:Concretization} reads $0^0$ as $1$.
A geometric distribution on $\N$ with parameter $\lambda$ has mean $\lambda/(1-\lambda)=\nu_1$.
Hence
\[
  \E[\beta_H(X)]
  =\sum_{t\le H}t\,f_t+\rho\Bigl(H+1+\frac{\lambda}{1-\lambda}\Bigr)
  =\sum_{t\le H}t\Pr(X=t)+\Pr(X>H)\,\E[X\mid X>H]
  =\E[X].
\]
The first right-hand side simplifies to the closed form $\E[s]=\sum_{t\le H}t\,f_t+\rho\,(H+1/(1-\lambda))$ that follows \Cref*{main:def:summary}.
The geometric tail thus preserves the mean in closed form, so the mean is the simplest instance of the moment analysis of Appendix~\ref{app:moment:sec:moment-method}.

\paragraph{Local loss.}
For an evaluation horizon $J\ge H$ and $s=\beta_H(X)$, the analysis records the bound
\[
  \eta(X;J)
  =\sum_{t=0}^{J}|S_X(t)-S_s(t)|
  +\Bigl(\E[X]-\sum_{t=0}^{J}S_X(t)\Bigr)
  +\Bigl(\E[s]-\sum_{t=0}^{J}S_s(t)\Bigr)
\]
(\cref*{main:eq:loss-bound}), which \Cref*{main:sec:prefix-tail-domain} shows to be at least the exact loss $\dW(X,s)$.
Moving from $J$ to $J+1$ adds $|S_X(J+1)-S_s(J+1)|$ and removes $S_X(J+1)+S_s(J+1)$, so $\eta(X;J)$ is nonincreasing in $J$ and converges to the exact loss.
All quantities in $\eta(X;J)$ are finite sums together with the two means, which the operator rule obtains from property~(FinComp).

\subsection{Operator Contracts}
\label{app:method:sec:cost-operator-contracts}

An admissible operator provides the four components of \Cref*{main:def:admissible-operator}.

\paragraph{Concrete semantics (CSem)}
\Cref*{main:def:cost-operators} defines each operator on random variables; the analysis computes with distributions, so each operator is also given as a map $\sem{\omega}$ that takes the distributions of $m$ independent operands to the distribution of the result.
For example, $\sem{+}$ is convolution, and $\sem{\max}$ multiplies distribution functions, $\Pr(\max(X,Y)\le t)=\Pr(X\le t)\Pr(Y\le t)$.
The map is exact.
Applied to surrogates, it propagates their shape error, and the abstraction of its result adds the local loss.

\paragraph{Prefix completeness (PComplete)}
With $\pre_H(X)=(\Pr(X=0),\ldots,\Pr(X=H))$, property (PComplete) asks for a finite-prefix transformer $\omega^{\#}_H$ with $\pre_H(\sem{\omega}(X_1,\ldots,X_m))=\omega^{\#}_H(\pre_H(X_1),\ldots,\pre_H(X_m))$.
It is a completeness property of the prefix for the operator.
Operators without (PComplete) are treated in \Cref*{main:sec:extensions} (paragraph ``Operators without prefix completeness'') and in \Cref{app:extension:sec:kernel-limitations}.

\paragraph{Error contracts (DistErr) and (QueryErr)}
The coefficients $C_i$ propagate moment error directly when the mean of the result is determined by the means of the operands.
This holds for sums, mixtures, and random repetitions, by linearity of expectation, the law of total expectation, and Wald's identity:
\[
  \E[X+Y]=\E[X]+\E[Y],
  \quad
  \E[\Mix_w(X_1,\ldots,X_m)]=\textstyle\sum_i w_i\,\E[X_i],
  \quad
  \E[\RepeatOp_N(X)]=\E[N]\,\E[X],
\]
the last when $N$ is independent of the copies of $X$ and has finite mean.
These operators carry nonzero $C_i$.
The means of the maximum and minimum are not determined by the means of the operands (the example with $X$, $Y_1$, and $Y_2$ in \Cref*{main:sec:introduction}), so they carry $C_i=0$ and $A_i=1$.
The coefficients $L_i$ of \Cref*{main:tab:operator-contracts} are the case $j=1$ of \Cref*{main:prop:order-uniform,main:prop:conv-dist}, proved in \Cref{app:moment:sec:higher-operator-contracts}.
For $\mathrm{Quorum}_k$ (\Cref*{main:ex:new-operator}), the $k^{\textit{th}}$ smallest of $m$ numbers satisfies $|x_{(k)}-y_{(k)}|\le\max_i|x_i-y_i|\le\sum_i|x_i-y_i|$.
Taking expectations under independent optimal couplings of each $X_i$ with $Y_i$ gives (DistErr) with $L_i=1$, and \cref*{main:eq:mean-dist} then gives (QueryErr) with $A_i=1$ and $C_i=0$.

\paragraph{Finite computability (FinComp)}
The surrogates have infinite support, so the analysis must evaluate $\sem{\omega}$ on their finite parameters.
For sums, mixtures, maxima, and minima, each mass of the result is a finite sum over masses of the operands (convolution, products of distribution or survival functions, and weighted sums).
For $\RepeatOp_N(X)$, the masses are the coefficients of $G_N(G_X(z))$ (\Cref*{main:tab:operator-contracts}), where $G_N$ and $G_X$ are the probability generating functions of $N$ and $X$; for $\Retry_a$, $G_N(y)=ay/(1-(1-a)y)$, which gives the recurrence below.
When $\Pr(X=0)=0$, the coefficient of $z^t$ involves only $\Pr(N=n)$ for $n\le t$, a finite sum for every $N$.
When $\Pr(X=0)=f_0>0$, it also involves the derivatives $G_N^{(n)}(f_0)$ for $n\le t$, which are infinite sums when $N$ has unbounded support; they are available when $G_N$ has a closed form, as for $\Retry_a$.
Every repeated cost in \Cref*{main:sec:evaluation} is positive.
Each mean follows from the identities above or, for $\max$ and $\min$, from the tail-sum formula, because beyond $H$ the survival function of the result is a linear combination of geometric sequences.
For an $m$-ary $\max$, this combination has $2^m-1$ terms (\Cref*{main:sec:complexity}); every $\max$ in \Cref*{main:sec:evaluation} is binary.
This gives everything that the abstraction and the loss bound need: the masses up to $J$ and the mean.

\paragraph{Retry.}
For $N\sim\Geom_{\ge1}(a)$, $\E[N]=1/a$, so $\Retry_a(X)=\sum_{j=1}^{N}X^{(j)}$ satisfies
\[
  \dW(\Retry_a(X),\Retry_a(Y))\le\tfrac1a\,\dW(X,Y),
  \qquad
  \E[\Retry_a(X)]=\E[X]/a .
\]
Both factors $1/a$ are the expected number of attempts: one enters the mean by Wald's identity, and the other enters the shape bound through (DistErr).
Let $b_t=\Pr(X=t)$ and $q_t=\Pr(\Retry_a(X)=t)$.
The first attempt either succeeds (probability $a$) or fails and restarts (probability $1-a$), so $q_t=a\,b_t+(1-a)\sum_{u=0}^{t}b_u\,q_{t-u}$, that is,
\[
  q_t=\frac{a\,b_t+(1-a)\sum_{u=1}^{t}b_u\,q_{t-u}}{1-(1-a)\,b_0}\qquad(t\ge0).
\]
For a positive attempt cost ($b_0=0$, hence $q_0=0$), this is \cref*{main:Eq:RetryRecurrence}.
The first $H+1$ masses of the result depend only on the first $H+1$ masses of the attempt, so the recurrence is a finite-prefix transformer, and $\Retry_a$ satisfies (PComplete).

\subsection{Abstract Semantics and Soundness}
\label{app:method:sec:error-bounded-semantics}
\label{app:method:sec:exact-prefix}

The operator rule of \Cref*{main:sec:error-bounded-semantics} proceeds in three steps.
It applies the exact map to the child surrogates, $\widetilde X=\sem{\omega}(s_1,\ldots,s_m)$, which property (FinComp) makes computable.
Because each $s_i$ may differ from $\sem{e_i}$ by up to $\epsilon_i$ in $\dW$, $\widetilde X$ generally differs from $\sem{e}$, and contract (DistErr) bounds the difference.
It then abstracts $s=\beta_H(\widetilde X)$ and records the loss bound $\eta=\eta(\widetilde X;J)$.
Finally, it combines the bounds:
\[
  \epsilon=\underbrace{\sum_iL_i\,\epsilon_i}_{\text{child error, by (DistErr)}}+\underbrace{\eta}_{\text{local}},
  \qquad
  \delta=\underbrace{\sum_iA_i\,\epsilon_i+\sum_iC_i\,\delta_i}_{\text{child error, by (QueryErr)}} .
\]
The abstraction preserves the mean, so it adds no term to $\delta$.
If an abstraction does not preserve the queried moment, as the geometric tail does for $M_2$ (\Cref*{main:sec:moment-tail}), it records a local bias bound $\zeta_j\ge|M_j(\widetilde X)-M_j(s)|$ for the queried moment $M_j$, which is added to $\delta$.
\Cref*{main:thm:generic-soundness}, restated and proved below, states that the rule is sound.
The prefix condition of the theorem is \Cref*{main:lem:exact-prefix}, which uses only (PComplete); the two bounds use only (DistErr) and (QueryErr) (\Cref*{main:sec:extensions}).

\paragraph{The max--retry rule.}
For a repeater node $\Retry_a(\max(e_L,e_R))$, the analysis abstracts after the maximum and again after the retry.
With losses $\eta_1$ and $\eta_2$ after the two steps, the rule gives $\epsilon=(\epsilon_L+\epsilon_R+\eta_1)/a+\eta_2$ and $\delta=(\epsilon_L+\epsilon_R)/a$, the bounds computed for the running example in \Cref*{main:sec:Overview:Bounds}.

\begin{restated}{Lemma~\ref*{main:lem:exact-prefix}}{Exact prefixes}
\label{app:method:lem:exact-prefix}
Let every operator in $e$ satisfy (PComplete).
If $e\Down(s,\epsilon,\delta)$, then $\pre_H(s)=\pre_H(\sem{e})$.
\end{restated}
\begin{proof}
By induction on the structure of the derivation.
For an atom with distribution $X$, $s = \beta_H(X)$ keeps the prefix of $X$.
At an operator node, (PComplete) and the induction hypothesis
$\mathrm{pre}_H(s_i) = \mathrm{pre}_H(\llbracket e_i \rrbracket)$ give
\[
  \mathrm{pre}_H(\widetilde{X})
  = \omega^\#_H(\mathrm{pre}_H(s_1), \ldots, \mathrm{pre}_H(s_m))
  = \omega^\#_H(\mathrm{pre}_H(\llbracket e_1 \rrbracket), \ldots,
                \mathrm{pre}_H(\llbracket e_m \rrbracket))
  = \mathrm{pre}_H(\llbracket e \rrbracket),
\]
and $s = \beta_H(\widetilde{X})$ keeps that prefix.
\end{proof}

\begin{restated}{Theorem~\ref*{main:thm:generic-soundness}}{Soundness}
\label{app:method:thm:generic-soundness}
Let every operator in $e$ be admissible.
If $e\Down(s,\epsilon,\delta)$, then $\sem{e}\in\gamma_H(s,\epsilon,\delta)$, that is, $\pre_H(s)=\pre_H(\sem{e})$, $\dW(\sem{e},s)\le\epsilon$, and $|\E[\sem{e}]-\E[s]|\le\delta$.
\end{restated}
\begin{proof}
The prefix condition is \Cref*{main:lem:exact-prefix}; the two bounds follow by induction on the derivation, using only (DistErr) and (QueryErr).
Atoms and solved subtrees satisfy both claims by the choice of $\eta$ and by eq.~(\ref*{main:eq:mean-preserving}).
For $e=\omega(e_1,\ldots,e_m)$ with $X_i=\sem{e_i}$, the induction hypothesis, the triangle inequality, (DistErr), and the choice of $\eta$ give
$\dW(\sem{e},s)\le\dW(\sem{\omega}(\vec X),\widetilde X)+\dW(\widetilde X,s)\le\sum_iL_i\epsilon_i+\eta=\epsilon$.
Because $\E[s]=\E[\widetilde X]$ by eq.~(\ref*{main:eq:mean-preserving}), (QueryErr) gives
$|\E[\sem{e}]-\E[s]|\le\sum_iA_i\epsilon_i+\sum_iC_i\delta_i=\delta$.
\end{proof}

\subsection{Generators}
\label{app:method:sec:front-ends}
\label{app:method:sec:general-generator-constructions}

This subsection details the generators of \Cref*{main:sec:cost-expressions}.
An application is instantiated by building a cost expression that models the system, typically from the library of built-in operators.
Large expressions are built by generators, recursive programs that pass strictly smaller parameters to each recursive call.

\paragraph{Size-indexed generators.}
A size-indexed generator $\mathsf{Gen}(c,n)$ has a size budget $n$ and further parameters $c$.
The base case returns an atom, $\mathsf{Gen}(c,1)=\Atom(\mu_c)$.
For $n>1$, a split policy chooses sizes $n_1,\ldots,n_m\ge1$ with $n_i<n$ and $\sum_in_i=n$ and parameters $c_1,\ldots,c_m$, and the generator returns $\omega(\mathsf{Gen}(c_1,n_1),\ldots,\mathsf{Gen}(c_m,n_m))$.
Generators are deterministic (\Cref*{main:sec:cost-expressions}): if the system splits at random, choosing rule $r$ with probability $w_r$, the generator returns the mixture $\Mix_{(w_r)_r}(e_r)_r$ of the expressions $e_r$ of the rules, as the RFID generator of \Cref*{main:sec:exp-setup} does.
Finite families can be analyzed by memoizing summaries for the reachable parameters $(c,n)$.

\paragraph{Interval generators.}
An interval generator $\mathsf{IGen}(i,j)$ is used when subproblems are contiguous intervals of an ordered structure.
The base case $\mathsf{IGen}(i,i)$ returns an atom, and a split policy returns an ordered partition $i=t_0<t_1<\cdots<t_m=j+1$ and the expression $\omega(\mathsf{IGen}(t_0,{t_1-1}),\ldots,\mathsf{IGen}(t_{m-1},t_m-1))$.
The repeater generator $R(i,j)$ of \Cref*{main:sec:exp-setup} is a binary interval generator whose operator is the composite $\Retry_a(\max(\cdot,\cdot))$.

\subsection{The Three Applications}
\label{app:method:sec:instantiations}

This subsection gives the cost expressions of the three applications of \Cref*{main:sec:exp-setup}.

\paragraph{Quantum repeaters.}
Leaves are elementary links, $\Atom(\Geom_{\ge1}(p_i))$, and each internal node is the composite $\Retry_a(\max(e_L,e_R))$: it waits for its two sub-repeaters and retries on a failed swap.
The query is the expected completion time.
Subtrees in $\mathcal O$ are solved by the exact solver of \Cref*{main:sec:exp-setup}, the Markov-chain model of Shchukin et al.~\cite{repeater-paper}, and their exact distributions are abstracted into summaries (\Cref{app:refinement:sec:refinement-profiles}).

\paragraph{RFID collision resolution.}
A reader must identify a batch of tags that respond simultaneously.
When several tags respond at once, their signals collide, and the reader splits the population by the next identifier bit and recurses; the two parts of a split are resolved concurrently, and the cost is the number of query rounds until all tags are identified.
Under uniformly distributed identifiers, a split of $n$ colliding tags sends $k$ tags to the first part with probability $\binom{n}{k}2^{-n}$.
With probability $\sigma_n=2^{1-n}$, all tags go to one part.
Such an unproductive split costs one round, because the empty part is queried in the same round as the colliding one, and the reader splits again.
The number $N_n$ of unproductive splits before the first productive one therefore has $\Pr(N_n=m)=(1-\sigma_n)\,\sigma_n^m$ for $m\ge0$, and the productive split sends $k$ tags to the first part with probability $w_k=\binom{n}{k}2^{-n}/(1-\sigma_n)$ for $1\le k\le n-1$.
Over a noisy channel, the reply of a single tag is decoded with probability $a$, and the reader repeats the query until it succeeds.
These facts give the generator $T$ of \Cref*{main:sec:exp-setup}.
In $T(n)=\RepeatOp_{N_n}(1)+1+\Mix_{w}(B_1,\ldots,B_{n-1})$, the first term counts the unproductive splits, the constant $1$ is the round of the productive split, and $B_k=\max(T(k),T(n-k))$ resolves its two parts concurrently; $T(1)=\Atom(\Geom_{\ge1}(a))$ counts the decoding attempts of a single tag.
It uses atoms and the operators $\RepeatOp_N$, $+$, $\Mix$, and $\max$.
For $a=1$, $T(n)-1$ is the height $H_n$ of the random binary trie of the identifiers: $H_n\le k$ exactly when the $k$-bit prefixes of the $n$ identifiers are distinct, so $\Pr(T(n)\le k+1)=\prod_{i=1}^{n-1}(1-i/2^k)$~\cite{flajolet1983performance}.
For $a<1$, $T(n)$ is the maximum over the tags of the depth of the tag's leaf plus its number of attempts, the height of a trie with random edge weights~\cite{broutin2008height}.
The accuracy experiment of \Cref*{main:sec:exp-other} solves populations of at most four tags exactly (protocol in \Cref{app:evaluation:subsec:exp-case-studies}).
Like a repeater node, each internal node combines a $\max$ with a repetition, and a mixture node expresses the random split.

\paragraph{Fork--join response times.}
A parallel real-time task has an end-to-end response time $X$; the experiments of \Cref*{main:sec:exp-other} query its deadline-miss probability $\Pr(X>D)=S_X(D)$ and its mean and variance.
The input is a series--parallel task tree.
Leaves are code segments with right-skewed execution-time distributions that put most mass on a typical cost and have a thin tail of rare slow runs, as in measurement-based probabilistic timing analysis~\cite{cucu2012measurement,davis2019survey}.
The leaf time $V$ of \Cref*{main:sec:exp-setup} has this form: it is $1$ with probability $0.99$, and otherwise $31+\Geom_{\ge1}(0.055)$, a slow run of at least $32$ units with a geometric excess.
An internal node with child expressions $e_L$ and $e_R$ generates $e_L+e_R$ for sequential composition, $\max(e_L,e_R)$ for a fork--join barrier, and $\Mix_{(u,1-u)}(e_L,e_R)$ for a data-dependent branch taken with probability $u$.
The tasks use no $\RepeatOp$, $\Retry$, or $\min$.
The miss probability is read from the survival function of the root summary: it is exact if $D\le H$, and otherwise the analysis returns the interval $S_s(D)\pm\epsilon$ of \Cref*{main:sec:error-bounded-semantics}, clipped to $[0,1]$.
The exact references propagate finite distribution prefixes: convolution for $+$, CDF products for $\max$, and weighted sums for $\Mix$.

\section{Details of the Moment Analysis}
\label{app:moment:sec:moment-method}

This appendix gives the details and proofs for \Cref*{main:sec:moment-method}.
For a fixed order $k$, the concrete domain is restricted to costs with finite $k$-th moments, the scalar bounds become vectors, $\dW$ is replaced by the weighted survival distances $d_j$, and the abstraction either preserves the moments up to order $k$ or reports local bias bounds.
The prefix--tail shape, exact prefixes, and exactly solved subtrees carry over unchanged; the contracts of convolutional operators acquire terms that mix moment orders.
The analysis of the mean is the case $k=1$: $d_1=\dW$, and the geometric tail preserves the first moment in closed form.

\subsection{Distances for Moments}
\label{app:moment:sec:moment-domain}

For $k\ge1$, let $\mathcal D_k=\{X\mid X\text{ is an }\N\text{-valued random variable and }\E[X^k]<\infty\}$, and let $M_j(X)=\E[X^j]$.
Because costs are non-negative, $\mathcal D_k\subseteq\mathcal D_{k-1}\subseteq\cdots\subseteq\mathcal D_1$, so a finite $k$-th moment ensures finite lower moments.
For $j\ge1$, the discrete tail-sum formula
\[
  M_j(X)=\sum_{t\ge0}\bigl((t+1)^j-t^j\bigr)S_X(t)
\]
follows from $X^j=\sum_{t=0}^{X-1}\bigl((t+1)^j-t^j\bigr)$ by taking expectations, and it implies $|M_j(X)-M_j(Y)|\le d_j(X,Y)$ for $d_j(X,Y)=\sum_{t\ge0}((t+1)^j-t^j)|S_X(t)-S_Y(t)|$.
For $j=1$ the weight is $1$ and $d_1=\dW$.
Because $x\mapsto x^j$ is increasing on $\N$, $\Pr(X^j>u)=S_X(t)$ for $t^j\le u<(t+1)^j$, so $d_j(X,Y)$ is the Wasserstein-1 distance between $X^j$ and $Y^j$.
It differs from the Wasserstein-$j$ distance for $j\ge2$, and it bounds the raw-moment error directly.

\subsection{Tails That Preserve Moments}
\label{app:moment:sec:moment-tail}

A summary at horizon $H$ has the form $s=(f_0,\ldots,f_H,\theta)$ and denotes the distribution with $\Pr(s=\nobreak t)=f_t$ for $t\le H$ and $\Pr(s=t)=\rho\,\tau_\theta(t-H-1)$ for $t>H$, where $\tau_\theta$ is a distribution on $\N$ from a tail family.
The geometric tail of the mean analysis is the family $\tau_\lambda(x)=(1-\lambda)\lambda^x$.
Given $X\in\mathcal D_k$, the abstraction keeps $f_t=\Pr(X=t)$ for $t\le H$, hence $\rho=\Pr(X>H)$.
Let $R_X$ be the residual tail, the conditional distribution of $X-H-1$ given $X>H$.
A \emph{moment-preserving} abstraction chooses $\theta$ with
\begin{equation}
\label{app:moment:Eqn:MomentMatchingConstraints}
  \E_{\tau_\theta}[R^j]=\E[R_X^j],\qquad 1\le j\le k,
\end{equation}
so that $M_j(\beta_H(X))=M_j(X)$ for $j\le k$, because the prefix and the tail mass are exact and the conditional tail moments match.
Alternatively, the abstraction keeps a simpler tail and reports local bias bounds $\zeta_j(X)\ge|M_j(X)-M_j(\beta_H(X))|$, which are added to the moment bounds.
The geometric tail can thus be reused for higher moments at the cost of a bounded local bias for $j\ge2$.

\paragraph{Maximum-entropy tails.}
When it exists, the maximum-entropy distribution on $\N$ under the constraints \eqref{app:moment:Eqn:MomentMatchingConstraints} has the form
\[
  \tau_\theta(x)=\frac{1}{Z(\theta)}\exp\Bigl(\sum_{i=1}^{k}\theta_ix^i\Bigr),\qquad x\in\N,
\]
for parameters $\theta$ with a finite normalizer $Z(\theta)$.
Fitting $\theta$ is a convex problem.
It minimizes $\Phi(\theta)=\log Z(\theta)-\sum_{j=1}^{k}\theta_j\nu_j$ with $\nu_j=\E[R_X^j]$.
The gradient of $\Phi$ is $(\E_{\tau_\theta}[R^j]-\nu_j)_{j=1}^{k}$, and its Hessian is the covariance matrix of $(R,\ldots,R^k)$ for $R\sim\tau_\theta$, which is positive definite.
For $k=1$, $\tau_{\theta_1}(x)\propto e^{\theta_1x}$; setting $\lambda=e^{\theta_1}$ gives the geometric tail, and matching the residual mean gives $\lambda=\nu_1/(1+\nu_1)$.
For $k=2$, \tool uses a damped Newton iteration (\Cref*{main:sec:moment-tail}); $Z(\theta)$ and the moments of $\tau_\theta$ are evaluated by truncated sums with a sound bound on the discarded part.
The analysis only requires that a tail-fitting routine either match the first $k$ residual moments or return sound bias bounds $\zeta_j$.

\paragraph{Fitting the quadratic tail.}
For $k=2$, $\tau_\theta(x)\propto\exp(\theta_1x+\theta_2x^2)$ is normalizable for $\theta_2<0$, and the Newton step uses a backtracking line search to stay in that domain.
With $u_x=\exp(\theta_1x+\theta_2x^2)$, the normalizer and moments are $Z(\theta)=\sum_{x\ge0}u_x$ and $\E_{\tau_\theta}[R^j]=Z(\theta)^{-1}\sum_{x\ge0}x^ju_x$, each summed up to a cutoff with a geometric bound on the remainder, which applies because the ratio $u_{x+1}/u_x$ is eventually below $1$ and decreasing.
The geometric tail is the boundary case $\theta_2=0$, reached when $\nu_2=\nu_1+2\nu_1^2$.
When the target residual moments lie outside the range of the family (for example when $\nu_2>\nu_1+2\nu_1^2$, where the fit would need $\theta_2>0$), the analysis falls back to the geometric tail and reports the biases $\zeta_j$; a numerical fit also reports its residual moment mismatch in $\zeta_j$.

\subsection{Abstract Semantics for Moments}
\label{app:moment:sec:moment-semantics}

The judgment $e\Down(s,\vec\epsilon,\vec\delta)$, with $\vec\epsilon=(\epsilon_1,\ldots,\epsilon_k)$ and $\vec\delta=(\delta_1,\ldots,\delta_k)$, asserts $d_j(\sem{e},s)\le\epsilon_j$ and $|M_j(\sem{e})-M_j(s)|\le\delta_j$ for $1\le j\le k$.
Rule \textsc{Atom} of \Cref*{main:sec:error-bounded-semantics} carries over componentwise: $s=\beta_H(X)$, $\epsilon_j=\eta_j(X)$, and $\delta_j=\zeta_j(X)$, where $\eta_j(X)$ is a sound upper bound on $d_j(X,\beta_H(X))$, bounded as in eq.~(\ref*{main:eq:loss-bound}) with the weights $w_j(t)$ and the moments $M_j$ (\Cref*{main:sec:moment-semantics}).
If the abstraction preserves the moments up to order $k$, all $\zeta_j$ are zero.

An order-$k$ contract for an operator $\omega$ of arity $m$ supplies, for $1\le j\le k$, nonnegative coefficients $L^{j,l}_i$, $A^{j}_i$, and $C^{j,l}_i$ ($1\le l\le j$, $1\le i\le m$) with
\[
  d_j(\sem{\omega}(\vec X),\sem{\omega}(\vec Y))\le\sum_{i=1}^{m}\sum_{l=1}^{j}L^{j,l}_i\,d_l(X_i,Y_i),
\]
\[
  |M_j(\sem{\omega}(\vec X))-M_j(\sem{\omega}(\vec Y))|\le\sum_{i=1}^{m}A^{j}_i\,d_j(X_i,Y_i)+\sum_{i=1}^{m}\sum_{l=1}^{j}C^{j,l}_i\,|M_l(X_i)-M_l(Y_i)|.
\]
The coefficients may depend on enclosures of lower moments of the operands, as for the convolutional operators below.
This is the form that the built-in operators need.
The order-$k$ contract of \Cref*{main:sec:moment-semantics} also allows both bounds to use every $d_l$ and every moment error $|M_l(X_i)-M_l(Y_i)|$ with $l\le j$, and the proof below covers that form unchanged.
For an operator node, the analysis computes $\widetilde X=\sem{\omega}(s_1,\ldots,s_m)$ and $s=\beta_H(\widetilde X)$, records the losses $\eta_j\ge d_j(\widetilde X,s)$ and the biases $\zeta_j$, and sets
\[
  \epsilon_j=\sum_{i}\sum_{l\le j}L^{j,l}_i\,\epsilon_{i,l}+\eta_j,
  \qquad
  \delta_j=\sum_iA^{j}_i\,\epsilon_{i,j}+\sum_i\sum_{l\le j}C^{j,l}_i\,\delta_{i,l}+\zeta_j,
\]
where $\epsilon_{i,l}$ and $\delta_{i,l}$ are the order-$l$ bounds of child $i$.
For pointwise operators the coefficients are diagonal ($l=j$ only) and constant (\Cref*{main:prop:order-uniform}); for convolutional operators they mix orders (\Cref{app:moment:sec:higher-operator-contracts}).

\begin{restated}{Theorem~\ref*{main:thm:moment-soundness}}{Soundness for raw moments}
\label{app:moment:thm:moment-soundness}
Fix $k\ge1$, and let every operator in $e$ have an order-$k$ contract.
If $e\Down(s,\vec\epsilon,\vec\delta)$, then $d_j(\sem{e},s)\le\epsilon_j$ and $|M_j(\sem{e})-M_j(s)|\le\delta_j$ for every $1\le j\le k$.
\end{restated}
\begin{proof}
By induction on $e$, as for \Cref*{main:thm:generic-soundness}.
The induction hypothesis gives the child bounds for all orders up to $k$.
The $d_j$ part of the contract propagates the children's $d_l$ bounds to a $d_j$ bound for $\sem{\omega}(\vec X)$ versus $\widetilde X$, and the triangle inequality adds $\eta_j$.
The moment contract propagates the children's errors, and $\zeta_j$ accounts for the bias of the abstraction.
Where a coefficient depends on lower moments of the operands, these moments lie in the intervals $M_l(s_i)\pm\delta_{i,l}$ by the induction hypothesis, so the coefficient evaluated by interval arithmetic bounds the true one.
\end{proof}

\paragraph{Exact prefixes are independent of the moment order.}
Property (PComplete) does not mention the query, so \Cref*{main:lem:exact-prefix} holds verbatim: under (PComplete), every summary stores the exact prefix, for every order $j\le k$.
Hence all $d_j$-losses originate in the tails.
Raising the order only increases the weight on the approximated tail, so higher moments are more sensitive to it.
When the geometric tail is kept with bias bounds, each bias $\zeta_j$ is also a tail-only quantity.

\subsection{Contracts for Higher Moments}
\label{app:moment:sec:higher-operator-contracts}

The operators split into two classes.
Pointwise operators (mixture, $\max$, $\min$) act on survival functions at a single time, and their $d_j$ contracts do not depend on the order.
Convolutional operators ($+$ and random repetition) combine survival values across times, and their contracts involve lower moments of the operands once $j\ge2$.

\begin{restated}{Proposition~\ref*{main:prop:order-uniform}}{Pointwise operators}
\label{app:moment:prop:order-uniform}
For independent inputs and $j\ge1$,
\[
  d_j(\Mix_w(\vec X),\Mix_w(\vec Y))\le\sum_i w_i\,d_j(X_i,Y_i)
  \quad\text{and}\quad
  d_j(\max(X,Z),\max(Y,Z))\le d_j(X,Y),
\]
and likewise for $\min$.
\end{restated}
\begin{proof}
At every $t\ge0$, independence gives $S_{\max(X,Z)}(t)=1-(1-S_X(t))(1-S_Z(t))$ and $S_{\min(X,Z)}(t)=S_X(t)\,S_Z(t)$.
Hence $|S_{\max(X,Z)}(t)-S_{\max(Y,Z)}(t)|=(1-S_Z(t))\,|S_X(t)-S_Y(t)|$ and $|S_{\min(X,Z)}(t)-S_{\min(Y,Z)}(t)|=S_Z(t)\,|S_X(t)-S_Y(t)|$, where both leading factors are at most~1.
For $\Mix_w$, $S_{\Mix_w(\vec X)}(t)=\sum_i w_i\,S_{X_i}(t)$, so the triangle inequality bounds the difference at $t$ by $\sum_i w_i\,|S_{X_i}(t)-S_{Y_i}(t)|$.
Multiplying by $w_j(t)\ge0$ and summing over $t$ establishes the claims.
\end{proof}

The forms with several perturbed arguments, such as $d_j(\max(X,Y),\max(X',Y'))\le d_j(X,X')+d_j(Y,Y')$, follow by the triangle inequality.

The moments of a mixture are the weighted moments of its components, so mixtures propagate moment errors with $C^{j,j}_i=w_i$.
The moments of a maximum or minimum are not determined by the moments of its operands at any order, so these operators bound moment errors through $d_j$, with $A^j_i=1$, as for $j=1$.

\paragraph{Convolutional operators.}
The means of sums and random repetitions are determined by the means of their operands, but their higher moments mix orders for $j\ge2$.
For independent $X$ and $Y$, $M_2(X+Y)=M_2(X)+2M_1(X)M_1(Y)+M_2(Y)$, and for $T=\RepeatOp_N(X)$ with $N$ independent of the copies of $X$,
\begin{align*}
  M_2(T)&=\E[N]\,M_2(X)+\bigl(\E[N^2]-\E[N]\bigr)M_1(X)^2,\\
  M_2(\Retry_a(X))&=\frac{1}{a}M_2(X)+\frac{2(1-a)}{a^2}M_1(X)^2 .
\end{align*}
In general, $M_j$ of a sum or a random sum is a polynomial in the lower moments of the operands.
These identities are not linear in the operands' moments, so the moment contract of a convolutional operator is not a fixed linear combination of child errors.
At order two, with $U_1^{X}$ and $U_1^{Y}$ upper bounds on the means of the operands, for example,
\[
\begin{aligned}
  |M_2(X'+Y')-M_2(X+Y)|\le{}&|M_2(X')-M_2(X)|+|M_2(Y')-M_2(Y)|\\
  &+2\bigl(U_1^{Y}|M_1(X')-M_1(X)|+U_1^{X}|M_1(Y')-M_1(Y)|\bigr),
\end{aligned}
\]
where $U_1^X$ bounds the means of both $X$ and $X'$, and likewise for $Y$.
The analysis therefore evaluates the coefficients $C^{j,l}_i$ of a convolutional operator by interval arithmetic on the moment polynomial, over the operands' intervals $M_l(s_i)\pm\delta_{i,l}$.
The same coupling appears in the distances: convolution never increases $d_1$, and it can increase $d_j$ for $j\ge2$.

The following statement contains \Cref*{main:prop:conv-dist} and adds the first-order bound for $\RepeatOp_N$ and the geometric case.
\begin{restated}{Proposition~\ref*{main:prop:conv-dist}}{Convolutional operators}
\label{app:moment:prop:conv-dist}
Let the inputs be independent, and let $M_0\equiv1$.
\begin{enumerate}
\item If $Z$ is independent of $X$ and $Y$, then for every $j\ge1$,
\(
  d_j(X+Z,Y+Z)\le\sum_{l=1}^{j}\binom{j}{l}M_{j-l}(Z)\,d_l(X,Y).
\)
\item If $N$ is independent of the copies of $X$ and $Y$, then $d_1(\RepeatOp_N(X),\RepeatOp_N(Y))\le\E[N]\,d_1(X,Y)$, and
\(
  d_2(\RepeatOp_N(X),\RepeatOp_N(Y))\le\E[N]\,d_2(X,Y)+\E[N(N-1)]\bigl(M_1(X)+M_1(Y)\bigr)d_1(X,Y);
\)
for $N\sim\Geom_{\ge1}(a)$, this is $\frac1a\,d_2(X,Y)+\frac{2(1-a)}{a^2}\bigl(M_1(X)+M_1(Y)\bigr)d_1(X,Y)$.
\end{enumerate}
In both bounds, the coefficient of the same-order distance is the constant of the mean case ($1$ and $\E[N]$); the remaining terms have lower order and carry raw moments of the operands.
\end{restated}
\begin{proof}
Because $X,Z\ge0$, set $S_X(u)=1$ for $u<0$.
By independence, $S_{X+Z}(t)=\sum_z\Pr(Z=z)S_X(t-z)$, and likewise for $Y+Z$; the terms with $t-z<0$ cancel because both survival values equal $1$.
With the weights $w_l(u)=(u+1)^l-u^l$ of \Cref*{main:sec:moment-domain}, the binomial theorem gives $w_j(u+z)=(u+z+1)^j-(u+z)^j=\sum_{l=1}^{j}\binom{j}{l}z^{j-l}w_l(u)$ (the term $l=0$ vanishes).
Substituting $u=t-z$,
\[
\begin{aligned}
  d_j(X+Z,Y+Z)
  &=\sum_{t\ge0}w_j(t)\,\Bigl|\sum_z\Pr(Z=z)\bigl(S_X(t-z)-S_Y(t-z)\bigr)\Bigr|\\
  &\le\sum_z\Pr(Z=z)\sum_{u\ge0}w_j(u+z)\,|S_X(u)-S_Y(u)|\\
  &=\sum_z\Pr(Z=z)\sum_{l=1}^{j}\binom{j}{l}z^{j-l}\sum_{u\ge0}w_l(u)\,|S_X(u)-S_Y(u)|
  =\sum_{l=1}^{j}\binom{j}{l}M_{j-l}(Z)\,d_l(X,Y),
\end{aligned}
\]
which proves (1).
For a fixed count $n$ (for $n=0$, both sums are $0$), replace the copies of $X$ by copies of $Y$ one at a time.
For $j=1$, each replacement costs at most $d_1(X,Y)$ by (1), so the sum of $n$ copies moves by at most $n\,d_1(X,Y)$, and averaging over $N$ gives the bound $\E[N]\,d_1(X,Y)$.
The $r$-th replacement perturbs one copy while the other $n-1$ copies form a fixed operand $Z_r$ with mean $(r-1)M_1(Y)+(n-r)M_1(X)$.
Summing the case $j=2$ of (1), $d_2(X+Z_r,Y+Z_r)\le d_2(X,Y)+2M_1(Z_r)\,d_1(X,Y)$, over $r=1,\ldots,n$, and using $\sum_{r=1}^{n}(r-1)=\sum_{r=1}^{n}(n-r)=n(n-1)/2$, gives
\[
  d_2\Bigl(\sum_{i=1}^{n}X^{(i)},\sum_{i=1}^{n}Y^{(i)}\Bigr)\le n\,d_2(X,Y)+n(n-1)\bigl(M_1(X)+M_1(Y)\bigr)d_1(X,Y).
\]
Averaging over $N$, through $S_{\RepeatOp_N(X)}=\sum_n\Pr(N=n)\,S_{X^{(1)}+\cdots+X^{(n)}}$ and the triangle inequality, gives the bound for $\RepeatOp_N$; substituting $\E[N]=1/a$ and $\E[N(N-1)]=2(1-a)/a^2$ gives the geometric case.
\end{proof}

\paragraph{Discharging the coefficients.}
In the order-$j$ rule, \Cref*{main:prop:conv-dist} supplies the coefficients $L^{j,l}_i$ of $+$ and $\RepeatOp$: the diagonal coefficient $L^{j,j}_i$ is the constant of the mean case, and each off-diagonal coefficient $L^{j,l}_i$ with $l<j$ is a combination of raw moments of the operands with nonnegative weights.
The analysis encloses these moments in the error-bounded intervals that it already maintains for the children ($M_0=1$) and evaluates the coefficients by the same interval arithmetic as the moment coefficients $C^{j,l}_i$.

\paragraph{The difficulty shifts with the moment order.}
The two classes of operators exchange roles relative to the mean.
The means of the maximum and minimum are not determined by the operands' means, but their $d_j$ contracts are order-uniform: one coefficient serves every order.
The means of sums and random repetitions are determined by the operands' means, but for $j\ge2$ these operators make the orders depend on each other.

\subsection{Variance and the Cost of Higher Orders}
\label{app:moment:sec:variance-query}
\label{app:moment:sec:moment-cost}

The variance is a derived query of the analysis with $k=2$, and its surrogate value is $M_2(s)-M_1(s)^2$.
The variance interval of \Cref*{main:sec:variance-query} combines the $M_1$ and $M_2$ bounds at their worst-case corners, so a variance query also yields a mean bound.

The cost grows with the order in three ways.
First, a moment-preserving tail needs enough parameters to match $k$ residual moments; the case $k=1$ has a closed form, whereas $k\ge2$ requires numerical fitting.
Second, the weight of $d_j$ grows like $jt^{j-1}$, so a small mismatch far in the tail is cheap for the mean but expensive for the variance and higher moments; useful higher-moment bounds may require larger prefixes, tighter local bounds, exactly solved subtrees, or richer tails.
Third, sums and repetitions need vector-valued propagation with interval coefficients, whereas pointwise operators generalize directly.
The refinement procedure of \Cref*{main:sec:refinement} is unchanged, except that it is applied to the bound of the selected moment or derived query.

\section{Details of Refinement and Design Search}
\label{app:refinement:sec:refinement}

This appendix gives the details of \Cref*{main:sec:refinement}: precision profiles, the attribution of the root bound to local losses, the cost model, the greedy refinement procedure, and the use of error bounds in the beam search.

\subsection{Precision Profiles}
\label{app:refinement:sec:refinement-profiles}

A precision profile $(H,\mathcal O)$ consists of a horizon $H$ and a set $\mathcal O$ of subtrees solved exactly; no subtree in $\mathcal O$ contains another.
One policy for choosing $\mathcal O$ solves exactly every maximal subtree with at most a given number of Markov-chain states; the scaling experiments of \Cref*{main:sec:exp-scaling} use $\mathcal O=\emptyset$ to isolate abstraction from exact-subtree solving.
Each subtree $f\in\mathcal O$ is solved exactly, with exact distribution $X_f=\sem{f}$, abstracted, and then treated as an atom, with summary and bounds
\[
  s_f=\beta_H(X_f),\qquad \epsilon_f=\eta_f=\eta(X_f;J)\ge\dW(X_f,s_f),\qquad \delta_f=0 .
\]
The equality $\delta_f=0$ uses the mean-preserving abstraction.
Solving $f$ exactly removes all losses that would arise inside it and leaves only the loss $\eta_f$ of abstracting $X_f$ at the root of $f$.
The profile $(H,\emptyset)$ is the fully abstract analysis.
If the root is in $\mathcal O$, the whole expression is solved exactly and abstracted once, at the cost of an exact solution.
Useful profiles lie between these extremes.

Refinement may also raise the horizon within a subtree.
As \Cref*{main:sec:refinement-profiles} states, \Cref*{main:thm:generic-soundness} then holds at each node $v$ for $\gamma_{H_v}$ as long as $H_v$ is at most the horizons of the children of $v$.
Raising the horizon within a whole subtree keeps this condition, because inside the subtree the horizon is uniform and its root has a larger horizon than its parent.

\subsection{Error Attribution}
\label{app:refinement:sec:error-amplification}

The recurrence for $\epsilon$ determines the root bound but accumulates it bottom-up, so it does not show which local loss is responsible for how much of the root bound.
Refinement therefore unrolls the same recurrence into a sum over local losses.

Let the abstraction sites of a derivation be its abstraction steps: one per atom, one per exactly solved subtree, and one per operator node, with two for an internal node $\Retry_a(\max(e_L,e_R))$.
Site $\ell$ records a loss bound $\eta_\ell$.
If the path from $\ell$ to the root passes through operators $\omega_1,\ldots,\omega_d$ at child positions $i_1,\ldots,i_d$, the amplification of $\ell$ is $W_\ell=\prod_{h=1}^{d}L^{(\omega_h)}_{i_h}$, the product of the coefficients $L_i$ along the path.
Inside a composite node, the product starts where the loss is introduced: a loss after the maximum and before the retry includes the factor of the retry, and a loss after the retry does not.
Because the recurrence is linear in the child bounds and the local losses, unrolling it gives eq.~(\ref*{main:eq:amp-generic}) of the paper,
\begin{equation*}
\label{app:refinement:eq:amp-generic}
  \epsilon_{\mathrm{root}}=\sum_{\ell}W_\ell\,\eta_\ell .
  \tag{\ref*{main:eq:amp-generic}}
\end{equation*}
For equal local losses, a larger amplification gives a larger contribution to the root shape bound.

\paragraph{The max--retry instance.}
For a repeater node $v=\Retry_{a_v}(\max(e_L,e_R))$, let $\eta^{1}_v$ be the loss after the maximum and $\eta^{2}_v$ the loss after the retry, and let $\mathrm{Anc}(v)$ be the proper ancestors of $v$.
The amplifications are
\[
  W^{1}_v=\frac{1}{a_v}\prod_{x\in\mathrm{Anc}(v)}\frac{1}{a_x},
  \qquad
  W^{2}_v=\prod_{x\in\mathrm{Anc}(v)}\frac{1}{a_x},
  \qquad
  W_f=\prod_{x\in\mathrm{Anc}(f)}\frac{1}{a_x}\quad(f\in\mathcal O),
\]
and \eqref{app:refinement:eq:amp-generic} specializes to
\begin{equation}
\label{app:refinement:eq:amp-maxretry}
  \epsilon_{\mathrm{root}}=\sum_{v}\bigl(W^{1}_v\eta^{1}_v+W^{2}_v\eta^{2}_v\bigr)+\sum_{f\in\mathcal O}W_f\,\eta_f ,
\end{equation}
where $v$ ranges over the internal nodes outside exactly solved subtrees.
The geometric leaves are represented exactly ($\eta=0$), so they contribute no term.
This is \eqref{app:refinement:eq:amp-generic} with the constants $1$ for $\max$ and $1/a$ for $\Retry_a$.

\begin{restated}{Proposition~\ref*{main:prop:monotonicity}}{Monotonicity}
\label{app:refinement:prop:monotonicity}
For a fixed tree, set $\mathcal O$, and
coefficients $L_i$, $A_i$, $C_i$ of the error-bound transformers,
$\epsilon_{\mathrm{root}}$ and $\delta_{\mathrm{root}}$ are nondecreasing in the local loss bounds $\eta_\ell$.
\end{restated}
\begin{proof}
By \eqref{app:refinement:eq:amp-generic}, $\epsilon_{\mathrm{root}}=\sum_\ell W_\ell\,\eta_\ell$, where each amplification $W_\ell$ is a product of coefficients $L_i$, and hence nonnegative and independent of the $\eta_\ell$.
At every node, $\delta$ is a combination of the children's $\epsilon_i$ and $\delta_i$ with nonnegative coefficients $A_i$ and $C_i$, and atoms have $\delta=0$; by induction on the tree, every $\epsilon$ and $\delta$ is therefore a nonnegative combination of the $\eta_\ell$.
\end{proof}

\begin{restated}{Proposition~\ref*{main:prop:convergence}}{Convergence}
\label{app:refinement:prop:convergence}
Fix $e$ and $\mathcal O$, and let every operator in $e$ satisfy (PComplete), (DistErr), and (QueryErr).
Then, in exact arithmetic and with (\ref*{main:eq:loss-bound}) for any $J\ge H$, $\epsilon_{\mathrm{root}}\to0$ and $\delta_{\mathrm{root}}\to0$ as $H\to\infty$.
\end{restated}
\begin{proof}
We show by induction on $e$ that $\epsilon\to0$ and $\delta\to0$ at every node as $H\to\infty$.
At each abstraction step, let $\widetilde X$ be the distribution that the rule abstracts (for an atom or a subtree in $\mathcal O$, its exact distribution), and let $s=\beta_H(\widetilde X)$.
The survival functions of $s$ and $\widetilde X$ agree up to $H$, and $\E[s]=\E[\widetilde X]$.
The first term of the loss bound (\ref*{main:eq:loss-bound}) therefore sums $|S_{\widetilde X}(t)-S_s(t)|\le S_{\widetilde X}(t)+S_s(t)$ over $H<t\le J$, and the other two terms are the sums of $S_{\widetilde X}(t)$ and $S_s(t)$ over $t>J$, so
\[
  \eta\le\sum_{t>H}\bigl(S_{\widetilde X}(t)+S_s(t)\bigr)=2\sum_{t>H}S_{\widetilde X}(t),
\]
because the two tail sums equal the two means minus the same survival values up to $H$.
By (PComplete), $\widetilde X$ has the prefix of the exact distribution $\sem{e'}$ of the subexpression $e'$ at this step (\Cref*{main:lem:exact-prefix}), so
\[
  \sum_{t>H}S_{\widetilde X}(t)=\E[\widetilde X]-\E[\sem{e'}]+\sum_{t>H}S_{\sem{e'}}(t)\le\delta+\sum_{t>H}S_{\sem{e'}}(t),
\]
where $\delta$ bounds $|\E[\sem{e'}]-\E[s]|$ by \Cref*{main:thm:generic-soundness}.
The last sum tends to $0$ because $\sem{e'}$ has a finite mean.
The bound $\delta=\sum_iA_i\epsilon_i+\sum_iC_i\delta_i$ tends to $0$ by the induction hypothesis, because the coefficients do not depend on $H$.
Hence $\eta\to0$, and $\epsilon=\sum_iL_i\epsilon_i+\eta\to0$.
\end{proof}

\subsection{Cost Model and Objectives}
\label{app:refinement:sec:refinement-cost}

Different profiles solve different subtrees exactly at different horizons and need different amounts of work.
For a node $v$ outside the exactly solved subtrees, let $c_v(H)$ be the cost of its abstract computation, with $J(H)\ge H$ the evaluation horizon of the loss bound.
By \Cref*{main:sec:complexity}, $c_v(H)=O(m\,J(H))$ for a mixture, maximum, or minimum, and $c_v(H)=\Theta(J(H)^2)$ for a sum or a retry evaluated directly.
For $f\in\mathcal O$, let $c_f$ be the cost of the exact solution, estimated from the size of the exact state space or measured, plus the cost of abstracting its result at $J(H)$.
The cost of a profile is
\[
  \mathrm{Cost}(H,\mathcal O)=\sum_{f\in\mathcal O}c_f+\sum_{v}c_v(H),
\]
with $v$ ranging over the nodes outside exactly solved subtrees.
Refinement approximately solves
\begin{equation}
\label{app:refinement:eq:opt-problem}
  \min_{(H,\mathcal O)}\ \mathrm{Obj}(H,\mathcal O)
  \quad\text{subject to}\quad
  \mathrm{Cost}(H,\mathcal O)\le B ,
\end{equation}
where the objective $\mathrm{Obj}$ is the root bound that the query uses (\Cref*{main:sec:greedy-refinement}): $\delta_{\mathrm{root}}$ for the mean and $\epsilon_{\mathrm{root}}$ for a deadline-miss probability.

\subsection{Greedy Refinement}
\label{app:refinement:sec:greedy-refinement}

\Cref{app:refinement:alg:greedy} shows the greedy refinement used by \tool.
Starting from an initial profile, such as $(H,\emptyset)$, each iteration forms candidate profiles and reruns the abstract semantics on each.
The main action solves one more eligible subtree exactly (in the experiment of \Cref{app:evaluation:subsec:exp-trajectories}, any subtree other than the root with at most $1024$ states), dropping its descendants from the abstract computation; a second action increases the horizon, globally or within a subtree.
Each candidate is scored by its decrease $\Delta\mathrm{Obj}$ of the objective per unit of added cost $\Delta C$; candidates over budget or without improvement are discarded, and the best ratio wins.
For an improving candidate with $\Delta C\le0$, the score is $+\infty$; stable candidate order breaks ties.
The loop ends when no improving candidate remains.

\begin{algorithm}[htbp]
\caption{Greedy refinement}
\label{app:refinement:alg:greedy}
\begin{algorithmic}[1]
\Require initial profile $\kappa_0$, such as $(H,\emptyset)$, budget $B$, objective $\mathrm{Obj}$
\Ensure a profile $\kappa$ with $\mathrm{Cost}(\kappa)\le B$
\State $\kappa\gets\kappa_0$
\Loop
  \State $\mathit{Cand}\gets$ the profiles that solve one more eligible subtree of $\kappa$ exactly or increase its horizon
  \State $\mathit{best}\gets\textsc{None}$; $r^\ast\gets0$
  \For{$\kappa'\in\mathit{Cand}$}
    \State rerun the abstract semantics under $\kappa'$
    \State $\Delta C\gets\mathrm{Cost}(\kappa')-\mathrm{Cost}(\kappa)$; $\Delta\mathrm{Obj}\gets\mathrm{Obj}(\kappa)-\mathrm{Obj}(\kappa')$
    \State $r\gets\Delta\mathrm{Obj}/\Delta C$ if $\Delta C>0$, otherwise $+\infty$
    \If{$\mathrm{Cost}(\kappa')\le B$ \textbf{and} $\Delta\mathrm{Obj}>0$ \textbf{and} $r>r^\ast$}
      \State $r^\ast\gets r$; $\mathit{best}\gets\kappa'$
    \EndIf
  \EndFor
  \If{$\mathit{best}=\textsc{None}$} \textbf{break} \EndIf
  \State $\kappa\gets\mathit{best}$
\EndLoop
\State \Return $\kappa$
\end{algorithmic}
\end{algorithm}

Soundness does not depend on the refinement policy: every profile it considers is analyzed by the abstract semantics, so every returned interval is sound.

\Cref*{main:prop:monotonicity} compares profiles with the same tree and set $\mathcal O$; solving a subtree exactly and changing the horizon fall outside it.
Solving a subtree exactly removes all losses inside it but adds the loss $\eta_f$ at its root, so its net effect depends on their sizes; a change of horizon changes the intermediate distributions and hence the losses at many sites.
The implementation therefore reanalyzes every candidate and accepts it only if the objective improves, as in \Cref{app:refinement:alg:greedy}.

\subsection{Analysis-Guided Beam Search}
\label{app:refinement:sec:design-feedback}

The attribution \eqref{app:refinement:eq:amp-generic} shows which parts of a tree are responsible for the remaining uncertainty.
In max--retry trees, losses below swaps with small success probabilities are amplified by large products of $1/a$, so these regions are natural candidates for exact solution.
Error-bounded intervals also compare trees: if two trees yield separated intervals $\widehat K_1\pm\delta_1$ and $\widehat K_2\pm\delta_2$, the analysis establishes which has the smaller expected cost; if they overlap, the attribution shows where more analysis budget should go.
\Cref*{main:alg:search} uses estimated means to rank repeater trees and checks interval separation after the search (\Cref{app:evaluation:subsec:exp-tree-design}).

\section{Evaluation Details}
\label{app:evaluation:sec:evaluation}
\label{app:evaluation:sec:experiment}

This appendix gives the protocols and the complete tables behind \Cref*{main:sec:experiment}, in the order of its subsections: the setup of \Cref*{main:sec:exp-setup} in \Cref{app:evaluation:subsec:exp-setup}, two additional repeater experiments in \Cref{app:evaluation:subsec:exp-error-bounds} (the precision controls of \Cref*{main:sec:refinement-profiles}) and \Cref{app:evaluation:subsec:exp-baselines} (hand-derived approximations on trees that the exact solver can solve), RQ1 and \Cref*{main:fig:chain-scaling} in \Cref{app:evaluation:subsec:exp-chains}, RQ2 in \Cref{app:evaluation:subsec:exp-tree-design}, and RQ3 in \Cref{app:evaluation:subsec:exp-moment,app:evaluation:subsec:exp-case-studies}.

\subsection{Reproducibility and Measurement}
\label{app:evaluation:subsec:exp-setup}

\paragraph{Artifacts.}
The experiments come from the two top-level directories of the artifact, \path{exp_pldi} and \path{current_experiments_20261012}; paths below are relative to the artifact root.
\begin{itemize}
\item \path{exp_pldi} contains the horizon sweep and the implementation of the beam search on large chains.
It holds the input manifests, implementations, isolated-process runner, validation scripts, and raw measurements, and the table it generates for the horizon sweep is the copy in the paper sources.
\item \path{current_experiments_20261012/chain_scaling} contains the chain experiment of \Cref{app:evaluation:subsec:exp-chains}: the harness, every job specification and raw record, and the summaries from which its tables and \Cref*{main:fig:chain-scaling} are generated (\path{plot_chain_2x2.py}).
\item \path{current_experiments_20261012/new_exp} contains one directory per table for the state limit, ten-link beam search, second-moment study, and fork--join deadlines.
\item \path{current_experiments_20261012/table4a_timing_ext} contains the comparison with hand-derived approximations on exactly solvable trees (Table~\ref{app:evaluation:tab:baselines-full}).
\item \path{current_experiments_20261012/critique_followups} contains the nested approximations on homogeneous chains, the floating-point check of the chain experiment, the RFID accuracy experiment (\path{x4_rfid_channel/rfid_channel.py}), the greedy refinement (\path{x5_refinement_subtree_H}), the Kolmogorov bounds for deadline queries (\path{x6_forkjoin_kolmogorov}), and the runs, tables, and figure of the beam search on large chains (\path{x8_search_pow2}).
\end{itemize}
The script \path{recompute_candidate_eta.py} in \path{current_experiments_20261012/figures/h_vs_error_horizon/} computes the data of Figure~\ref{app:evaluation:fig:h-vs-error}, and \path{scripts/verify_tree_optimum.py} performs the exhaustive check of the ten-link beam search.

\paragraph{Units.}
The state limit counts Markov-chain states, and a block of $b$ links has $2^b$ states: a limit of $4$ solves two-link blocks exactly, and a limit of $16$ solves four-link blocks exactly.
The repeater and fork--join analyses use $J=4H$ (\Cref*{main:sec:exp-setup}), and the RFID analyses use the fixed evaluation horizon $J=2000$.

\paragraph{Implementations.}
Two implementations of \tool produce the results.
The original implementation evaluates convolutions and retries directly in $O(J^2)$ time per node and solves the subtrees within the state limit with the Markov-chain solver.
It produces the refinement and the ten-link search of \Cref*{main:sec:exp-design}, \Cref*{main:tab:design-search-and-moments}(b) and \Cref*{main:tab:cases}, and the results of \tool in this appendix except the horizon sweep, the chain experiment, and the large-chain search.
The vectorized first-moment kernel of \Cref*{main:sec:exp-setup} solves no subtree exactly and produces these three, \Cref*{main:fig:chain-scaling}, and \Cref*{main:tab:design-search-and-moments}(a).
For $J+1>128$, the kernel computes a retry by Newton power-series inversion with linear (non-circular) convolution, which is FFT-based for large arrays.
The vectorized kernel always adds the floating-point mismatch of the abstraction $\beta_H$ to the moment bound $\delta$, whereas the original implementation sets mismatches below $10^{-10}\max(1,K)$ to zero; a unit test checks that both agree within a relative tolerance of $10^{-8}$ on six-link trees at $H\in\{8,32\}$.
The comparison of \Cref{app:evaluation:subsec:exp-baselines} runs the original implementation for all methods; the chain experiment runs the original implementation for the exact solver and the hand-derived approximations and the vectorized kernel for \tool and truncated propagation.

\paragraph{Inputs and repetition of the large-scale experiments.}
The manifests specify every size, parameter, seed, method, and horizon.
The horizon sweep and the large-chain search use seeds $0,1,2$; the seeds of the fork--join suites are given in \Cref{app:evaluation:subsec:exp-case-studies}.
The RFID population model is deterministic, and the moment study uses homogeneous trees, so neither has an artificial seed dimension.
Each timed input/method pair runs three times in fresh processes; separate reference jobs supply numerical comparisons.
A seeded shuffle fixes the job order, and jobs execute serially.

The recorded environment is Linux on an Intel Core i7-8700 CPU, Python 3.11.5, NumPy 1.24.3, and SciPy 1.11.1.
Numerical-library thread counts are one and each worker is pinned to logical CPU~0.
The runner imposes a $90$\,s wall-time limit per process and a $4$\,GiB address-space limit; the chain experiment has its own runner with the same memory limit and $600$\,s per call (\Cref{app:evaluation:subsec:exp-chains}).
The recorded peak resident set includes interpreter and library memory, unlike the address-space cap.
Source hashes and the exact configuration are stored with the results.

\paragraph{Timing boundary.}
The reported timer encloses one complete invocation of the algorithm, including all adaptive-horizon attempts of \tool{} and all doublings of the truncation point in truncated propagation; for the adaptive mode, the chain tables also give the time of the last attempt alone (``final'').
Imports, startup, and serialization are excluded.
Single-tree inputs are constructed before timing.
Search timing includes parameter and candidate generation, candidate analysis, and the final reanalysis of the starting and selected trees.
Its component timers also record candidate-analysis and final-check time.
Maximum-entropy fits start with an empty global fit cache.
Cross-candidate summary reuse is enabled only in the explicitly named search configuration.

\paragraph{Aggregation and failures.}
A table first takes the median time of the completed repetitions for each input, then the median of those input medians.
Displayed numerical errors and bounds are also medians across distinct inputs with a reference, except in \Cref{app:evaluation:tab:baselines-full}, which reports means.
Brackets $[\min,\max]$ give the range of the per-input medians.
Per-input results retain every timing sample and status.
An unfinished process is reported as a timeout or failure; a horizon cap reached without satisfying the requested bound is reported as an unmet target.
A speedup or saved-time figure compares corresponding completed inputs.
All percentages of successful precision targets count distinct inputs rather than repetitions.
\texttt{MEM} denotes an allocation failure under the $4$\,GiB address-space cap, and a dash an unavailable completed value or reference.
Search speedups and saved seconds are medians of paired-input ratios and differences; the gain in Table~\ref{app:evaluation:tab:search} is relative to the starting tree.

\paragraph{Error diagnostics.}
Relative error and relative bound are as in \Cref*{main:sec:exp-setup}, with $K$ the reference value, and coverage is the fraction of instances with $|K-\widehat K|\le\delta$.
Unless a caption states otherwise, coverage allows an additive tolerance of $10^{-8}\max(1,|K|)$ for floating-point roundoff.
A reference obtained by truncated floating-point distribution propagation is a numerically converged value without a bound on the omitted mass; convergence diagnostics and coverage are recorded separately.
The implementation computes in double-precision floating point; the coverage checks validate it against the references, and the soundness theorem holds in real arithmetic.

\paragraph{Timing protocol of the chain experiment}
The chain experiment ran on the same machine and software as \path{exp_pldi}.
NumPy was linked against Intel MKL, numerical libraries used one thread, every timed process was pinned to logical CPU~5, and jobs ran sequentially.
The frequency governor was \texttt{powersave} with turbo enabled.
\Cref{app:evaluation:subsec:exp-chains} gives the warm-up rule and the number of processes per input.
The tables from the original implementation report accuracy and bounds only.

\subsection{Precision Controls}
\label{app:evaluation:subsec:exp-error-bounds}
\label{app:evaluation:subsec:exp-certificates}

This subsection varies the two precision controls of \Cref*{main:sec:refinement-profiles}, the horizon $H$ and the set $\mathcal O$ of exactly solved subtrees, one at a time.

\paragraph{Horizon.}
The inputs are three eight-link heterogeneous balanced trees ($p=0.1$, $a=0.5$, seeds $0,1,2$) with $p_i=pU_i$ and independent $U_i\sim\mathrm{Uniform}(0.6,1.4)$.
Their exact means come from the absorbing Markov chain of Shchukin et al.~\cite{repeater-paper} with $256$ states.
The \tool sweep uses $H\in\{8,32,128,512\}$ with no subtree solved exactly.
This changes one precision parameter while holding the input and implementation fixed.
The Markov-chain computation is a separate reference job; it takes $3.38$\,ms (median) and gives $K=274$, $247$, and $292$.
Truncated propagation (\Cref{app:evaluation:subsec:exp-chains}) reaches the truncation point $T=16384$ and differs from the Markov-chain means by about $1.5\times10^{-13}$ (relative).
Kernel tests also compare the retry masses computed by power-series inversion against the direct recurrence and the propagated moments against small Markov-chain solutions.

Table~\ref{app:evaluation:tab:firstpass-convergence} reports medians over the three trees.
Across the trees, the relative error ranges over $[1.99\%,2.38\%]$ at $H=8$, $[0.072\%,0.108\%]$ at $H=32$, and $[1.3\times10^{-4}\%,4.7\times10^{-4}\%]$ at $H=128$, and the relative bound over $[435\%,505\%]$, $[58.1\%,82.0\%]$, and $[0.278\%,1.07\%]$.

\begin{table}[tb]
\centering
\scriptsize
\setlength{\tabcolsep}{3pt}
\caption{Horizon sweep on three eight-link heterogeneous balanced repeaters ($p=0.1$, $a=0.5$), with no subtree solved exactly (\texttt{exp\_pldi}).
Times are in seconds, medians with the range over the trees; errors and bounds are relative to the Markov-chain solution, medians over the trees.}
\label{app:evaluation:tab:firstpass-convergence}
\begin{adjustbox}{max width=\linewidth}
\begin{tabular}{rrrrr}
\toprule
$H$ & Time (s) & Error (\%) & Bound (\%) & Coverage \\
\midrule
8 & 0.00147 [0.00147,0.00148] & 2.14 & 471 & 3/3 \\
32 & 0.00262 [0.00257,0.00263] & 0.0848 & 77.2 & 3/3 \\
128 & 0.00415 [0.00412,0.00419] & 0.000355 & 0.827 & 3/3 \\
512 & 0.0151 [0.0151,0.0154] & $5.84\!\times\!10^{-13}$ & $2.33\!\times\!10^{-8}$ & 3/3 \\
\bottomrule
\end{tabular}

\end{adjustbox}
\end{table}

\begin{figure}[tb]
\centering
\includegraphics[width=0.48\linewidth]{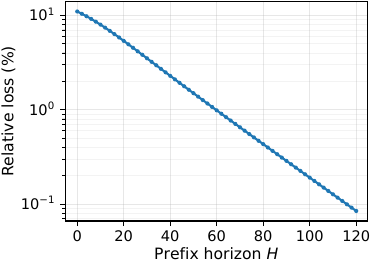}
\caption{Relative loss $100\,\dW(X,\beta_H(X))/K$ of abstracting the exact root distribution $X$ of the homogeneous balanced repeater with $n=8$, $p=0.1$, and $a=0.5$ ($K=242.35$), for $H=0,2,\ldots,120$, on a logarithmic scale.
This is the loss of one abstraction; the propagated bounds $\epsilon$ and $\delta$ of the analysis also accumulate the amplified losses of all nodes.}
\label{app:evaluation:fig:h-vs-error}
\end{figure}

Figure~\ref{app:evaluation:fig:h-vs-error} isolates the effect of the horizon on one abstraction.
It abstracts the exact root distribution of a homogeneous balanced repeater, computed on the grid $t\le6000$ (omitted mass below $10^{-12}$), and plots the exact loss $\dW(X,\beta_H(X))$ relative to $K$.
The loss falls from $10.99\%$ of $K$ at $H=0$ to $3.22\%$ at $H=32$, $0.836\%$ at $H=64$, and $0.0849\%$ at $H=120$, by a factor of about $11$ for every $60$ units of $H$, so the curve is close to a straight line on the logarithmic scale.
It falls much faster than the tail mass $\Pr(X>H)$, which is still $0.67$ at $H=120$: as $H$ grows, the remaining tail approaches a geometric shape, which the mean-preserving geometric tail fits increasingly well.

\paragraph{State limit.}
The state limit is a second control, independent of the horizon: every maximal subtree whose Markov chain has at most that many states is solved exactly and abstracted at its root.
Table~\ref{app:evaluation:tab:frontier-threshold-ablation} fixes $H=12$ ($J=48$) and varies only the limit, on one exactly solvable homogeneous balanced repeater with $n=8$, $p=0.1$, and $a=0.9$; the root ($256$ states) is never solved exactly.
Raising the limit from $2$ to $16$ states moves six of the seven internal nodes into exactly solved subtrees; the relative error falls from $1.1852\%$ to $1.0014\%$, and the relative bound from $45.43\%$ to $19.35\%$.
Rows coincide when no further subtree falls within the limit: a limit of $4$ or $8$ solves the four two-link blocks exactly, and a limit of $16$, $32$, or $64$ solves the two four-link subtrees exactly.

The same scan covers $30$ exactly solvable homogeneous balanced repeaters with eight links: $H\in\{12,16,20\}$ combined with $(a,p)\in\{0.5,0.7\}\times\{0.3,0.5\}$ and with $(a,p)\in\{0.8,0.9\}\times\{0.1,0.2,0.3\}$.
Table~\ref{app:evaluation:tab:frontier-threshold-ablation} shows the configuration with the largest decrease of the error.
In all $30$ configurations, the relative bound does not grow as the limit rises, every bound covers the exact mean, and raising the limit from $2$ to $64$ states shrinks the relative bound by a median of $35\%$ and by up to a factor of $27$.

\begin{table}[tb]
\centering
\scriptsize
\caption{Effect of the state limit at $H=12$ on one balanced repeater ($n=8$, $p=0.1$, $a=0.9$; original implementation).
The state limit counts Markov-chain states; the second column counts the internal nodes inside exactly solved subtrees.
Errors and bounds are relative to the exact mean; coverage allows an additive tolerance of $10^{-9}+10^{-10}K$.}
\label{app:evaluation:tab:frontier-threshold-ablation}
\begin{tabular}{rrrrc}
\toprule
State limit & Nodes solved exactly & Rel.\ error & Rel.\ bound & Coverage \\
\midrule
$2$  & $0$ & $1.1852\%$ & $45.43\%$ & yes \\
$4$  & $4$ & $1.1852\%$ & $40.62\%$ & yes \\
$8$  & $4$ & $1.1852\%$ & $40.62\%$ & yes \\
$16$ & $6$ & $1.0014\%$ & $19.35\%$ & yes \\
$32$ & $6$ & $1.0014\%$ & $19.35\%$ & yes \\
$64$ & $6$ & $1.0014\%$ & $19.35\%$ & yes \\
\bottomrule
\end{tabular}
\end{table}

\subsection{Hand-Derived Approximations}
\label{app:evaluation:subsec:exp-baselines}

This subsection compares \tool with the hand-derived approximations of \Cref*{main:sec:exp-setup} on trees that the exact solver can still solve.

\paragraph{Methods.}
The doubling formula is $(3/(2a))^d/p$ for $n=2^d$ links~\cite{sangouard2011quantum}.
\textit{nested\_}$b$ is the effective-probability approximation of Shchukin et al.~\cite{repeater-paper}: it solves sub-chains of $b$ links exactly, replaces each by an effective link-success probability, and combines them by rate formulas.
\tool with $b$-link blocks uses the state limit $2^b$, so it solves exactly the same blocks as \textit{nested\_}$b$ and differs only in how it combines them: it propagates the prefix--tail summary of each block with error bounds, whereas \textit{nested\_}$b$ collapses each block to a single probability.
The exact solver is the Markov chain of the whole tree; in the artifact, it is the nested method with $32$-link blocks, which solves these trees whole.

\paragraph{Exactly solvable balanced repeaters.}
The inputs are the $24$ homogeneous balanced trees with $n\in\{4,8\}$, $p\in\{0.1,0.37,0.63\}$, and $a\in\{0.3,0.5,0.7,0.9\}$; the references are Markov-chain solutions, which we recomputed independently (largest relative difference $4.5\times10^{-15}$).
For $n=4$, a four-link block is the whole tree, so \textit{nested\_4} and \tool with four-link blocks are exact; the $12$ inputs with $n=8$ are the informative subset.
Table~\ref{app:evaluation:tab:baselines-full} reports all methods, including \tool with two-link blocks, the setting of \textit{nested\_2}.

\begin{table}[tb]
\centering
\scriptsize
\setlength{\tabcolsep}{3pt}
\caption{Hand-derived approximations and \tool on the $24$ exactly solvable homogeneous balanced repeaters (original implementation).
``Blocks'' is the number of links in the sub-chains solved exactly.
Errors and bounds are means over the inputs, relative to the exact mean.
Coverage counts the inputs with $|K-\widehat K|\le\delta+10^{-12}K$.}
\label{app:evaluation:tab:baselines-full}
\begin{tabular}{@{}lrr rr rr r@{}}
\toprule
& & & \multicolumn{2}{c}{Rel.\ error} & \multicolumn{2}{c}{Rel.\ bound} & \\
\cmidrule(lr){4-5}\cmidrule(lr){6-7}
Method & Blocks & $H$ & All & $n=8$ & All & $n=8$ & Coverage \\
\midrule
doubling formula & -- & -- & $21.60\%$ & $26.18\%$ & -- & -- & -- \\
\textit{nested\_2} & $2$ & -- & $6.554\%$ & $9.700\%$ & -- & -- & -- \\
\textit{nested\_4} & $4$ & -- & $2.986\%$ & $5.971\%$ & -- & -- & -- \\
\addlinespace
\tool & $4$ & $4$  & $0.7596\%$ & $1.519\%$ & $57.13\%$ & $114.26\%$ & $24/24$ \\
\tool & $4$ & $8$  & $0.3677\%$ & $0.7353\%$ & $37.22\%$ & $74.43\%$ & $24/24$ \\
\tool & $4$ & $16$ & $0.1000\%$ & $0.2000\%$ & $20.72\%$ & $41.43\%$ & $24/24$ \\
\tool & $4$ & $32$ & $0.0165\%$ & $0.0330\%$ & $9.92\%$ & $19.85\%$ & $24/24$ \\
\addlinespace
\tool & $2$ & $4$  & $1.196\%$ & $1.975\%$ & $157.00\%$ & $260.80\%$ & $24/24$ \\
\tool & $2$ & $8$  & $0.4890\%$ & $0.8627\%$ & $84.51\%$ & $144.40\%$ & $24/24$ \\
\tool & $2$ & $16$ & $0.1096\%$ & $0.2104\%$ & $36.62\%$ & $65.18\%$ & $24/24$ \\
\tool & $2$ & $32$ & $0.0165\%$ & $0.0330\%$ & $13.04\%$ & $24.59\%$ & $24/24$ \\
\addlinespace
exact solver & all & -- & $<10^{-12}\%$ & $<10^{-12}\%$ & -- & -- & -- \\
\bottomrule
\end{tabular}
\end{table}

With the same four-link blocks, \tool has a smaller error than \textit{nested\_4} on all $12$ inputs with $n=8$ at every horizon; with the two-link blocks of \textit{nested\_2}, it has a smaller error than \textit{nested\_2} on all $24$ inputs.
Wherever they are not exact, the hand-derived approximations overestimate the mean.
The largest errors at $n=8$ are $11.92\%$ for \textit{nested\_4}, $6.33\%$ for \tool with four-link blocks at $H=4$, and $0.284\%$ at $H=32$.
The tolerance matters only for the $8$ inputs with $n=4$ and four-link blocks, where the whole tree is one exact block and $\delta=0$; there $|K-\widehat K|\le4.6\times10^{-13}$ (relative $2.2\times10^{-15}$), the floating-point difference between the abstraction at the root and the independent reference solution.
All inputs with $n=8$ are covered without tolerance.
The estimates reproduce those of the original runs to four decimals.

Chains beyond the reach of the exact solver are compared in \Cref{app:evaluation:subsec:exp-chains}.

\subsection{Accuracy and Cost from Small to Large Chains}
\label{app:evaluation:subsec:exp-scaling}
\label{app:evaluation:subsec:exp-chains}

This subsection gives the protocol and all cells behind \Cref*{main:fig:chain-scaling} and RQ1 (\Cref*{main:sec:exp-chains}).

\paragraph{Inputs.}
The trees split each contiguous interval of links into two halves; for each leaf, $p_i=pU_i$ with independent $U_i\sim\mathrm{Uniform}(0.6,1.4)$, drawn with seeds $0$, $1$, and $2$ by the generator of \path{exp_pldi}.
The sizes are $n\in\{4,8,12,16,32,64,\ldots,4096\}$; $n=12$ is included only to show the growth of the exact solver and uses the same balanced split.
The four settings combine $a\in\{0.8,0.5\}$ with $p\in\{0.37,0.1\}$; \Cref*{main:fig:chain-scaling} shows $a=0.8$, $p=0.37$, and \Cref*{main:sec:exp-chains} summarizes the other three settings.
\tool solves no subtree exactly; \textit{nested\_}$b$ solves $b$-link blocks exactly by definition.

\paragraph{Methods and limits.}
Every (tree, method) pair runs in its own process, pinned to logical CPU~5, with a $4$\,GiB address-space limit and a $600$\,s limit per call; for $n\le64$, an untimed identical call precedes the timed one, otherwise the first call is timed.
A pair whose first run takes at most $10$\,s is measured in three processes, and a table cell is the median over the three trees of the per-tree medians; a pair whose first run takes longer is run once or twice.
A method that fails on a tree (memory, timeout, or error) is not run on larger trees, which the tables mark with a dash; \texttt{MEM} marks the size at which it failed.
The exact solver is the dense Markov-chain solver of the original implementation (\path{exact_solve}), whose state count is $2^n$; at $n=16$ the allocation of its transition matrix fails.
\textit{nested\_2} and \textit{nested\_4} are the original implementation's \path{analyze_rate_leaves}.
Truncated propagation and \tool are the vectorized kernel of \path{exp_pldi}, called through its worker; truncated propagation doubles $T$ from $512$ to at most $2^{22}$ and stops when the retained mass is within $10^{-10}$ of one at the last two truncation points and the mean changed by at most a relative $10^{-7}$.
In the adaptive mode, \tool tries $H=32,64,\ldots,16384$ with $J=4H$ and stops at the first $H$ with $\widehat K>\delta$ and $\delta/(\widehat K-\delta)\le1\%$; the tables give the time of the analysis at the final $H$ (``final'', plotted in \Cref*{main:fig:chain-scaling}) and the total over all attempts of the doubling schedule (``total''), and a run that fails the test at the cap keeps its estimate and bound but is marked as missing the target ($^\ast$ in the tables).
The fixed variant uses $H=128$.
The reference mean of a tree is its converged truncated propagation; errors are omitted where it did not complete.
Tables~\ref{app:evaluation:tab:chain-a08-p037}--\ref{app:evaluation:tab:chain-a05-p01} list every cell.

\paragraph{Fixed horizons.}
The fixed-$H$ curves of \tool in \Cref*{main:fig:chain-scaling} ($H=32$, $128$, $512$) come from a separate run (\path{grid_experiment.py}) on the trees with $a=0.8$ and $p=0.37$ that covers every $n\in\{4,8,16,\ldots,4096\}$ and every $H\in\{32,64,\ldots,4096\}$, with $J=4H$, one process per tree and horizon, pinned to the same CPU, the same warm-up rule, $4$\,GiB, and $600$\,s per call; a cell is the median over the three trees.
All $264$ runs complete, with a peak resident set of at most $0.5$\,GiB, and each of the $240$ intervals with a reference contains it.
On the homogeneous chains, \path{critique_followups/x1_homogeneous/homog_grid.py} runs the same protocol for $H\in\{32,128,512\}$ on one tree per size; all $33$ runs complete, and each interval contains its reference.
The runs execute one at a time, and at the final horizon of the adaptive mode their times agree with the column ``final'' of \Cref{app:evaluation:tab:chain-a08-p037} to within $1\%$ (for example, $8.63$\,s and $8.62$\,s at $4096$ links); the methods without a horizon in the figure use the times of the tables.

\paragraph{Floating-point rounding}
\tool computes in double precision, and the intervals do not account for rounding (\Cref*{main:sec:exp-setup}).
To measure its effect, \path{critique_followups/x3_floating} and \path{x3_floating_4096} reanalyze the $66$ heterogeneous trees with $a=0.8$, $p\in\{0.37,0.1\}$, and $4$ to $4096$ links at the final horizon of the adaptive mode, once as in the chain experiment and once with the convolutions of the retries computed by direct summation in extended precision.
Every double-precision run reproduces the estimate of the chain experiment to within a relative $10^{-12}$.
Extended precision changes $\widehat K$ by at most $1.6\times10^{-11}\,\widehat K$ and $\delta$ by at most $3.5\times10^{-10}\,\widehat K$, both below $10^{-9}\,\widehat K$ and far below the reported $\delta$.

\paragraph{Results.}
At $a=0.8$, \tool{} in the adaptive mode meets the $1\%$ target on every tree of both settings: with $p=0.37$ its final horizon grows from $32$ to $1024$ and the analysis at that horizon takes $8.6$\,s at $4096$ links ($27$\,s over all attempts); with $p=0.1$, where completion times are about four times longer, the horizon grows to $4096$ and the analysis to $23$\,s ($63$\,s over all attempts), with a peak resident set of $0.5$\,GiB.
Truncated propagation exceeds $4$\,GiB at $4096$ links ($p=0.37$) and at $2048$ links ($p=0.1$); the exact solver fails at $16$ links in every setting.
The relative error of \textit{nested\_4} grows from $7.8$--$9.0\%$ at $8$ links to $132$--$139\%$ at $1024$ links, and \textit{nested\_2} errs more.
At $a=0.5$, each retry amplifies the local losses by $2$ instead of $1.25$, and the mean completion time grows by about $3$ instead of $1.9$ per level (at $p=0.37$: $1533$ at $64$ links, $36{,}529$ at $512$, $102{,}811$ at $1024$, $879{,}985$ at $4096$).
\tool meets the target on every tree up to $512$ links for $p=0.37$ (at $H=16384$, $11$\,s per analysis) and up to $128$ links for $p=0.1$ (at $H=8192$, $1.3$\,s).
Truncated propagation stays within $4$\,GiB only up to $256$ and $64$ links.
On larger trees, the horizon cap $H=16384$ covers a shrinking share of the distribution ($H/\widehat K$ falls from $0.45$ at $512$ links to $0.02$ at $4096$), and the bound grows with the size, from $36\%$ at $1024$ links to $712\%$ at $4096$ links for $p=0.37$.
\textit{nested\_4} errs by $4\%$ at $8$ links, $20$--$21\%$ at $64$, and $33\%$ at $256$ links ($p=0.37$), the largest size with a reference.
Every interval of \tool contains its reference.
At $2048$ links ($p=0.37$), truncated propagation takes $47$\,s with a peak resident set of $2.0$\,GiB, and \tool $4.3$\,s at its final horizon with $0.2$\,GiB.

\paragraph{Homogeneous chains}
\Cref*{main:fig:chain-scaling}(a, b) uses the homogeneous chains $p_i=p$ (\path{critique_followups/x1_homogeneous}), with the same balanced trees, methods, limits, and warm-up rule, for $a=0.8$ and $p\in\{0.37,0.1\}$.
A homogeneous tree does not depend on a seed, so each size has one tree, and each method runs in one timed process.
The reference is the Markov chain up to $12$ links and truncated propagation beyond.
Table~\ref{app:evaluation:tab:chain-homogeneous} lists every cell.
The adaptive mode meets the $1\%$ target on every tree, and all $24$ intervals contain the reference.
Identical subtrees are analyzed once, so both \tool and truncated propagation complete at every size: at $4096$ links ($p=0.37$), \tool takes $0.039$\,s at its final horizon $H=1024$ ($0.18$\,s over all attempts), and truncated propagation $0.63$\,s.
The doubling formula errs by $18\%$ at $4$ links and by $271\%$ at $4096$; \textit{nested\_4} errs by $8.6\%$ at $8$ links, $50\%$ at $64$, and more than $100\%$ from $512$ links on, close to its errors on heterogeneous chains.

\begin{table}[tb]
\centering
\scriptsize
\setlength{\tabcolsep}{3pt}
\caption{Homogeneous chains, $a=0.8$ (\Cref*{main:fig:chain-scaling}(a, b) shows $p=0.37$); one tree per size and one timed process per method.
Errors are relative to the Markov chain up to $12$ links and to truncated propagation beyond; ``doubling'' is the formula $(3/(2a))^d/p$, defined for $n=2^d$; other columns as in Table~\ref{app:evaluation:tab:chain-a08-p037}.}
\label{app:evaluation:tab:chain-homogeneous}
\begin{adjustbox}{max width=\linewidth}
\begin{tabular}{@{}r r r rr r r rrrrr@{}}
\toprule
& & Exact & \multicolumn{2}{c}{Rel.\ error (\%)} & Trunc.\ prop. & & \multicolumn{5}{c@{}}{\tool, adaptive} \\
\cmidrule(lr){4-5}\cmidrule(l){8-12}
$p$ & $n$ & (s) & \texttt{nested\_4} & doubling & (s) & & $H$ & final (s) & total (s) & error (\%) & bound (\%) \\
\midrule
0.37 & 4 & 0.000246 & 0 & 18.2 & 0.00245 & & 32 & 0.000665 & 0.000665 & $5.08\!\times\!10^{-13}$ & $5.80\!\times\!10^{-12}$ \\
 & 8 & 0.00121 & 8.63 & 31.3 & 0.00347 & & 32 & 0.000975 & 0.000975 & $6.94\!\times\!10^{-13}$ & $1.66\!\times\!10^{-7}$ \\
 & 12$^{\dagger}$ & 1.03 & 17 & -- & 0.00464 & & 32 & 0.0013 & 0.0013 & $2.64\!\times\!10^{-9}$ & $6.42\!\times\!10^{-5}$ \\
 & 16 & MEM & 20.1 & 46.9 & 0.00464 & & 32 & 0.0013 & 0.0013 & $1.28\!\times\!10^{-8}$ & $4.82\!\times\!10^{-4}$ \\
 & 32 & -- & 33.9 & 64.8 & 0.018 & & 32 & 0.00166 & 0.00166 & $4.96\!\times\!10^{-5}$ & 0.0954 \\
 & 64 & -- & 49.8 & 85 & 0.0213 & & 64 & 0.00243 & 0.00448 & $3.77\!\times\!10^{-6}$ & 0.02 \\
 & 128 & -- & 67.9 & 108 & 0.0673 & & 128 & 0.00409 & 0.0102 & $1.71\!\times\!10^{-7}$ & 0.00311 \\
 & 256 & -- & 88.4 & 133 & 0.069 & & 128 & 0.00501 & 0.0127 & $3.27\!\times\!10^{-4}$ & 0.32 \\
 & 512 & -- & 111 & 162 & 0.125 & & 256 & 0.0099 & 0.0267 & $3.59\!\times\!10^{-5}$ & 0.0803 \\
 & 1024 & -- & 137 & 194 & 0.145 & & 512 & 0.0264 & 0.0639 & $2.52\!\times\!10^{-6}$ & 0.0157 \\
 & 2048 & -- & 167 & 230 & 0.294 & & 512 & 0.0323 & 0.0847 & 0.00166 & 0.888 \\
 & 4096 & -- & 199 & 271 & 0.625 & & 1024 & 0.0389 & 0.177 & $2.43\!\times\!10^{-4}$ & 0.264 \\
\midrule
0.1 & 4 & 0.000184 & 0 & 8.96 & 0.00602 & & 32 & 0.000666 & 0.000666 & $2.26\!\times\!10^{-4}$ & 0.0545 \\
 & 8 & 0.00108 & 9.65 & 20.1 & 0.00945 & & 64 & 0.00118 & 0.00213 & $3.45\!\times\!10^{-6}$ & 0.00703 \\
 & 12$^{\dagger}$ & 1.01 & 19.1 & -- & 0.0189 & & 64 & 0.00151 & 0.00274 & $1.72\!\times\!10^{-4}$ & 0.144 \\
 & 16 & MEM & 22 & 34.1 & 0.0193 & & 64 & 0.00158 & 0.00292 & $7.88\!\times\!10^{-4}$ & 0.448 \\
 & 32 & -- & 36.6 & 50.3 & 0.0374 & & 128 & 0.00279 & 0.00653 & $7.68\!\times\!10^{-5}$ & 0.116 \\
 & 64 & -- & 53.2 & 68.6 & 0.0454 & & 256 & 0.00572 & 0.014 & $5.42\!\times\!10^{-6}$ & 0.024 \\
 & 128 & -- & 71.9 & 89.4 & 0.0937 & & 512 & 0.0152 & 0.0322 & $2.43\!\times\!10^{-7}$ & 0.00378 \\
 & 256 & -- & 93 & 113 & 0.18 & & 512 & 0.0181 & 0.0394 & $4.00\!\times\!10^{-4}$ & 0.36 \\
 & 512 & -- & 117 & 139 & 0.391 & & 1024 & 0.0211 & 0.0689 & $4.47\!\times\!10^{-5}$ & 0.0916 \\
 & 1024 & -- & 143 & 168 & 0.442 & & 2048 & 0.0366 & 0.122 & $3.23\!\times\!10^{-6}$ & 0.0183 \\
 & 2048 & -- & 173 & 201 & 0.957 & & 2048 & 0.0437 & 0.164 & 0.00193 & 0.976 \\
 & 4096 & -- & 207 & 238 & 2.23 & & 4096 & 0.0824 & 0.321 & $2.90\!\times\!10^{-4}$ & 0.294 \\
\bottomrule
\end{tabular}

\end{adjustbox}
\end{table}

\begin{table}[tb]
\centering
\scriptsize
\setlength{\tabcolsep}{3pt}
\caption{Chain experiment, $a=0.8$, $p=0.37$ (\Cref*{main:fig:chain-scaling}(c, d)). Times are medians in seconds; for the adaptive analysis, ``final'' is the analysis at the final horizon $H$ and ``total'' the sum over all attempts of the doubling schedule; errors are relative to the reference of truncated propagation, and bounds are $100\,\delta/\widehat K$; RSS is the peak resident set; $^\dagger$: balanced trees with $12$ links; $^\ast$: the $1\%$ target was not met at the horizon cap.}
\label{app:evaluation:tab:chain-a08-p037}
\begin{adjustbox}{max width=\linewidth}
\begin{tabular}{@{}r r rr r r rrrrrr rrr@{}}
\toprule
& Exact & \multicolumn{2}{c}{Rel.\ error (\%)} & Trunc.\ prop. & & \multicolumn{6}{c}{\tool, adaptive} & \multicolumn{3}{c@{}}{\tool, $H{=}128$} \\
\cmidrule(lr){3-4}\cmidrule(lr){7-12}\cmidrule(l){13-15}
$n$ & (s) & \texttt{nested\_2} & \texttt{nested\_4} & (s) & & $H$ & final (s) & total (s) & RSS (MiB) & error (\%) & bound (\%) & time (s) & error (\%) & bound (\%) \\
\midrule
4 & 0.000291 & 4.41 & $1.36\!\times\!10^{-12}$ & 0.00357 & & 32 & 0.00103 & 0.00103 & 100 & $4.04\!\times\!10^{-11}$ & $3.40\!\times\!10^{-7}$ & 0.00166 & $5.20\!\times\!10^{-13}$ & $2.22\!\times\!10^{-12}$ \\
8 & 0.00144 & 12.4 & 7.83 & 0.00824 & & 32 & 0.00227 & 0.00228 & 100 & $1.12\!\times\!10^{-7}$ & $2.51\!\times\!10^{-4}$ & 0.00377 & $1.96\!\times\!10^{-12}$ & $8.72\!\times\!10^{-12}$ \\
12$^{\dagger}$ & 0.998 & 17.5 & 15.5 & 0.0128 & & 32 & 0.00346 & 0.00347 & 100 & $9.34\!\times\!10^{-6}$ & 0.00317 & 0.0059 & $6.60\!\times\!10^{-12}$ & $1.49\!\times\!10^{-10}$ \\
16 & MEM & 23.9 & 18.7 & 0.0176 & & 32 & 0.0048 & 0.0048 & 100 & $2.01\!\times\!10^{-5}$ & 0.00601 & 0.00811 & $5.88\!\times\!10^{-12}$ & $3.06\!\times\!10^{-10}$ \\
32 & -- & 37.9 & 31.9 & 0.129 & & 32 & 0.00991 & 0.00991 & 100 & $6.91\!\times\!10^{-4}$ & 0.24 & 0.0165 & $1.16\!\times\!10^{-10}$ & $6.61\!\times\!10^{-7}$ \\
64 & -- & 53.1 & 46.6 & 0.376 & & 64 & 0.025 & 0.0449 & 100 & $2.10\!\times\!10^{-4}$ & 0.0509 & 0.0342 & $3.80\!\times\!10^{-7}$ & $2.19\!\times\!10^{-4}$ \\
128 & -- & 71.6 & 64.2 & 0.761 & & 128 & 0.0681 & 0.158 & 101 & $5.76\!\times\!10^{-6}$ & 0.0104 & 0.0695 & $5.76\!\times\!10^{-6}$ & 0.0104 \\
256 & -- & 92.7 & 84.4 & 2.24 & & 128 & 0.138 & 0.321 & 102 & $8.81\!\times\!10^{-4}$ & 0.702 & 0.139 & $8.81\!\times\!10^{-4}$ & 0.702 \\
512 & -- & 116 & 107 & 6.99 & & 256 & 0.465 & 1.1 & 110 & $1.64\!\times\!10^{-4}$ & 0.227 & 0.275 & 0.0828 & 11.2 \\
1024 & -- & 143 & 132 & 23.4 & & 512 & 2.4 & 4.62 & 125 & $1.32\!\times\!10^{-5}$ & 0.0473 & 0.556 & 1.07 & 67.7 \\
2048 & -- & 173 & 161 & 46.8 & & 1024 & 4.27 & 13.5 & 175 & $1.77\!\times\!10^{-7}$ & 0.00899 & 1.1 & 5.19 & 258 \\
4096 & -- & -- & -- & MEM & & 1024 & 8.62 & 27.1 & 203 & -- & 0.63 & 2.23 & -- & 644 \\
\bottomrule
\end{tabular}

\end{adjustbox}
\end{table}

\begin{table}[tb]
\centering
\scriptsize
\setlength{\tabcolsep}{3pt}
\caption{Chain experiment, $a=0.8$, $p=0.1$; columns as in Table~\ref{app:evaluation:tab:chain-a08-p037}.}
\label{app:evaluation:tab:chain-a08-p01}
\begin{adjustbox}{max width=\linewidth}
\begin{tabular}{@{}r r rr r r rrrrrr rrr@{}}
\toprule
& Exact & \multicolumn{2}{c}{Rel.\ error (\%)} & Trunc.\ prop. & & \multicolumn{6}{c}{\tool, adaptive} & \multicolumn{3}{c@{}}{\tool, $H{=}128$} \\
\cmidrule(lr){3-4}\cmidrule(lr){7-12}\cmidrule(l){13-15}
$n$ & (s) & \texttt{nested\_2} & \texttt{nested\_4} & (s) & & $H$ & final (s) & total (s) & RSS (MiB) & error (\%) & bound (\%) & time (s) & error (\%) & bound (\%) \\
\midrule
4 & 0.000294 & 5.41 & $6.18\!\times\!10^{-12}$ & 0.00928 & & 32 & 0.00102 & 0.00102 & 100 & 0.00167 & 0.336 & 0.00169 & $2.00\!\times\!10^{-10}$ & $1.25\!\times\!10^{-6}$ \\
8 & 0.00147 & 14.8 & 8.95 & 0.0339 & & 64 & 0.00273 & 0.00497 & 100 & $2.64\!\times\!10^{-4}$ & 0.0755 & 0.00384 & $1.98\!\times\!10^{-7}$ & $3.31\!\times\!10^{-4}$ \\
12$^{\dagger}$ & 1 & 20.5 & 17.7 & 0.0527 & & 64 & 0.00435 & 0.00787 & 100 & $3.41\!\times\!10^{-4}$ & 0.457 & 0.00594 & $9.46\!\times\!10^{-6}$ & 0.00292 \\
16 & MEM & 27.3 & 20.7 & 0.073 & & 64--128 & 0.00801 & 0.0185 & 100 & $1.07\!\times\!10^{-4}$ & 0.0135 & 0.00822 & $1.90\!\times\!10^{-5}$ & 0.00603 \\
32 & -- & 42.1 & 34.6 & 0.236 & & 128 & 0.017 & 0.0393 & 100 & $5.34\!\times\!10^{-4}$ & 0.256 & 0.0169 & $5.34\!\times\!10^{-4}$ & 0.256 \\
64 & -- & 58.5 & 50.2 & 0.77 & & 256 & 0.0571 & 0.137 & 101 & $1.79\!\times\!10^{-4}$ & 0.0525 & 0.0341 & 0.0241 & 4.94 \\
128 & -- & 77.6 & 68.3 & 1.58 & & 512 & 0.271 & 0.559 & 102 & $4.60\!\times\!10^{-6}$ & 0.0106 & 0.0698 & 0.539 & 40.6 \\
256 & -- & 99.5 & 89.1 & 5.56 & & 512 & 0.538 & 1.1 & 107 & $9.46\!\times\!10^{-4}$ & 0.706 & 0.14 & 3.46 & 179 \\
512 & -- & 124 & 112 & 21.1 & & 1024 & 1.07 & 3.27 & 129 & $1.67\!\times\!10^{-4}$ & 0.224 & 0.28 & 10.3 & 509 \\
1024 & -- & 152 & 139 & 83.1 & & 2048 & 3.28 & 9.87 & 189 & $1.33\!\times\!10^{-5}$ & 0.047 & 0.561 & 20 & 1.04e+03 \\
2048 & -- & -- & -- & MEM & & 4096 & 11.4 & 31.3 & 389 & -- & 0.00884 & 1.12 & -- & 1.86e+03 \\
4096 & -- & -- & -- & -- & & 4096 & 22.8 & 62.6 & 490 & -- & 0.622 & 2.24 & -- & 2.97e+03 \\
\bottomrule
\end{tabular}

\end{adjustbox}
\end{table}

\begin{table}[tb]
\centering
\scriptsize
\setlength{\tabcolsep}{3pt}
\caption{Chain experiment, $a=0.5$, $p=0.37$; columns as in Table~\ref{app:evaluation:tab:chain-a08-p037}.}
\label{app:evaluation:tab:chain-a05-p037}
\begin{adjustbox}{max width=\linewidth}
\begin{tabular}{@{}r r rr r r rrrrrr rrr@{}}
\toprule
& Exact & \multicolumn{2}{c}{Rel.\ error (\%)} & Trunc.\ prop. & & \multicolumn{6}{c}{\tool, adaptive} & \multicolumn{3}{c@{}}{\tool, $H{=}128$} \\
\cmidrule(lr){3-4}\cmidrule(lr){7-12}\cmidrule(l){13-15}
$n$ & (s) & \texttt{nested\_2} & \texttt{nested\_4} & (s) & & $H$ & final (s) & total (s) & RSS (MiB) & error (\%) & bound (\%) & time (s) & error (\%) & bound (\%) \\
\midrule
4 & 0.000291 & 2.59 & $2.00\!\times\!10^{-12}$ & 0.00915 & & 32 & 0.00102 & 0.00103 & 100 & $1.01\!\times\!10^{-9}$ & $2.22\!\times\!10^{-5}$ & 0.0017 & $2.28\!\times\!10^{-12}$ & $4.49\!\times\!10^{-12}$ \\
8 & 0.00142 & 6.62 & 4.09 & 0.035 & & 32--64 & 0.00225 & 0.00226 & 100 & $1.39\!\times\!10^{-4}$ & 0.281 & 0.00386 & $1.77\!\times\!10^{-12}$ & $3.32\!\times\!10^{-8}$ \\
12$^{\dagger}$ & 0.999 & 7.69 & 6.74 & 0.0842 & & 64--128 & 0.00583 & 0.0135 & 100 & $3.67\!\times\!10^{-6}$ & 0.00833 & 0.00596 & $1.48\!\times\!10^{-6}$ & 0.00425 \\
16 & MEM & 11.8 & 9.08 & 0.115 & & 128 & 0.00809 & 0.019 & 100 & $4.23\!\times\!10^{-6}$ & 0.0565 & 0.00808 & $4.23\!\times\!10^{-6}$ & 0.0565 \\
32 & -- & 17.4 & 14.4 & 0.655 & & 256 & 0.0283 & 0.0674 & 100 & $1.03\!\times\!10^{-4}$ & 0.535 & 0.0165 & 0.00136 & 12.4 \\
64 & -- & 23.2 & 20.1 & 2.55 & & 1024 & 0.131 & 0.396 & 105 & $6.02\!\times\!10^{-6}$ & 0.0354 & 0.0342 & 0.418 & 131 \\
128 & -- & 29.5 & 26.2 & 21.7 & & 2048 & 0.407 & 1.21 & 113 & $6.96\!\times\!10^{-6}$ & 0.527 & 0.0692 & 2.88 & 500 \\
256 & -- & 36.3 & 32.9 & 93.1 & & 8192 & 2.6 & 6.43 & 188 & $1.74\!\times\!10^{-6}$ & 0.0442 & 0.14 & 7.2 & 1.13e+03 \\
512 & -- & -- & -- & MEM & & 16384 & 10.9 & 23.8 & 446 & -- & 0.66 & 0.276 & -- & 2.03e+03 \\
1024 & -- & -- & -- & -- & & 16384 & 21.9 & 47.8 & 758 & -- & 36.2$^{\ast}$ & 0.553 & -- & 3.21e+03 \\
2048 & -- & -- & -- & -- & & 16384 & 43.8 & 95.2 & 1206 & -- & 243$^{\ast}$ & 1.11 & -- & 4.83e+03 \\
4096 & -- & -- & -- & -- & & 16384 & 85.5 & 189 & 1578 & -- & 712$^{\ast}$ & 2.23 & -- & 6.97e+03 \\
\bottomrule
\end{tabular}

\end{adjustbox}
\end{table}

\begin{table}[tb]
\centering
\scriptsize
\setlength{\tabcolsep}{3pt}
\caption{Chain experiment, $a=0.5$, $p=0.1$; columns as in Table~\ref{app:evaluation:tab:chain-a08-p037}.}
\label{app:evaluation:tab:chain-a05-p01}
\begin{adjustbox}{max width=\linewidth}
\begin{tabular}{@{}r r rr r r rrrrrr rrr@{}}
\toprule
& Exact & \multicolumn{2}{c}{Rel.\ error (\%)} & Trunc.\ prop. & & \multicolumn{6}{c}{\tool, adaptive} & \multicolumn{3}{c@{}}{\tool, $H{=}128$} \\
\cmidrule(lr){3-4}\cmidrule(lr){7-12}\cmidrule(l){13-15}
$n$ & (s) & \texttt{nested\_2} & \texttt{nested\_4} & (s) & & $H$ & final (s) & total (s) & RSS (MiB) & error (\%) & bound (\%) & time (s) & error (\%) & bound (\%) \\
\midrule
4 & 0.00029 & 3.26 & $4.55\!\times\!10^{-12}$ & 0.0151 & & 64 & 0.00116 & 0.00218 & 100 & $2.48\!\times\!10^{-5}$ & 0.0764 & 0.00171 & $5.26\!\times\!10^{-9}$ & $2.19\!\times\!10^{-5}$ \\
8 & 0.00142 & 7.87 & 4.47 & 0.0878 & & 128--256 & 0.00381 & 0.00878 & 100 & $1.31\!\times\!10^{-4}$ & 0.278 & 0.00383 & $3.55\!\times\!10^{-4}$ & 0.827 \\
12$^{\dagger}$ & 1 & 8.9 & 7.52 & 0.24 & & 256--512 & 0.0101 & 0.0239 & 100 & $4.63\!\times\!10^{-6}$ & 0.561 & 0.00593 & 0.00764 & 16.5 \\
16 & MEM & 13.3 & 9.73 & 0.323 & & 512 & 0.0311 & 0.0633 & 100 & $3.34\!\times\!10^{-6}$ & 0.0487 & 0.00809 & 0.0238 & 34.5 \\
32 & -- & 19 & 15.2 & 2.53 & & 1024 & 0.0653 & 0.196 & 103 & $8.35\!\times\!10^{-5}$ & 0.483 & 0.0167 & 0.909 & 221 \\
64 & -- & 25 & 20.9 & 10.2 & & 4096 & 0.354 & 0.956 & 117 & $4.67\!\times\!10^{-6}$ & 0.0314 & 0.0341 & 3.98 & 665 \\
128 & -- & -- & -- & MEM & & 8192 & 1.31 & 3.25 & 140 & -- & 0.483 & 0.0698 & -- & 1.37e+03 \\
256 & -- & -- & -- & -- & & 16384 & 5.52 & 12.1 & 265 & -- & 3.38$^{\ast}$ & 0.14 & -- & 2.35e+03 \\
512 & -- & -- & -- & -- & & 16384 & 11.1 & 24.3 & 446 & -- & 74.8$^{\ast}$ & 0.278 & -- & 3.68e+03 \\
1024 & -- & -- & -- & -- & & 16384 & 22.3 & 48.8 & 759 & -- & 357$^{\ast}$ & 0.56 & -- & 5.42e+03 \\
2048 & -- & -- & -- & -- & & 16384 & 44.6 & 96.9 & 1213 & -- & 908$^{\ast}$ & 1.13 & -- & 7.77e+03 \\
4096 & -- & -- & -- & -- & & 16384 & 87.6 & 193 & 1577 & -- & 1.71e+03$^{\ast}$ & 2.25 & -- & $1.09\!\times\!10^{4}$ \\
\bottomrule
\end{tabular}

\end{adjustbox}
\end{table}

\subsection{Refinement and Analysis-Guided Beam Search}
\label{app:evaluation:subsec:exp-trajectories}
\label{app:evaluation:subsec:exp-tree-design}

\paragraph{Refinement setup.}
This experiment runs the greedy refinement of \Cref*{main:sec:greedy-refinement} (\Cref{app:refinement:sec:greedy-refinement}) on two fixed trees with $n=10$, $p=0.3$, and $a=0.5$: the pair-and-carry tree and an AVL tree; Table~\ref{app:evaluation:tab:tree-design-search-running} lists both bracketings.
The initial profile has $H=32$ ($J=128$) and solves the two-link blocks exactly (state limit $4$).
The candidate actions solve one more subtree with at most $1024$ states exactly, or raise the horizon to the next value in $\{64,128\}$ within the subtree of any analyzed node, exactly solved subtrees included, whose exact distributions are then abstracted at the larger horizon.
The root is not eligible for exact solution, and the objective is the root bound $\delta$.
The cost model of \Cref*{main:sec:greedy-refinement} charges $J^2$ for each node analyzed abstractly and $S^2+J^2$ for an exactly solved subtree with $S$ states, its solution and the abstraction of its result.
The budget is $B=\lceil rC_0\rceil$ for the ratios $r\in\{1.00,1.05,1.10,1.25,1.5,2,4\}$, where $C_0$ is the cost of the initial profile.
The budget bounds the cost of the final profile, so refinement need not use all of it; solving a subtree exactly can lower the cost when it replaces abstractly analyzed nodes, and such improving steps have score $+\infty$.

\begin{table}[tb]
\centering
\scriptsize
\caption{Greedy refinement on two ten-link repeater trees ($p=0.3$, $a=0.5$; original implementation).
Budgets are in cost units; $C_0$ is the cost of the initial profile.}
\label{app:evaluation:tab:trajectory-summary}
\begin{tabular}{lrrrrrr}
\toprule
Tree & $C_0$ & Budget $B$ & $B/C_0$ & Steps & Final $\delta$ & Final / initial \\
\midrule
Pair-and-carry & $147536$ & $147536$ & $1.00$ & $0$ & $47.38$ & $1.0000$ \\
Pair-and-carry & $147536$ & $154913$ & $1.05$ & $1$ & $42.76$ & $0.9025$ \\
Pair-and-carry & $147536$ & $162290$ & $1.10$ & $1$ & $42.76$ & $0.9025$ \\
Pair-and-carry & $147536$ & $184420$ & $1.25$ & $2$ & $5.845$ & $0.1234$ \\
Pair-and-carry & $147536$ & $221304$ & $1.50$ & $3$ & $5.845$ & $0.1234$ \\
Pair-and-carry & $147536$ & $295072$ & $2.00$ & $3$ & $5.845$ & $0.1234$ \\
Pair-and-carry & $147536$ & $590144$ & $4.00$ & $4$ & $0.1103$ & $0.0023$ \\
\addlinespace
AVL & $147520$ & $147520$ & $1.00$ & $0$ & $18.53$ & $1.0000$ \\
AVL & $147520$ & $154896$ & $1.05$ & $4$ & $1.27$ & $0.0685$ \\
AVL & $147520$ & $162272$ & $1.10$ & $4$ & $1.27$ & $0.0685$ \\
AVL & $147520$ & $184400$ & $1.25$ & $4$ & $1.27$ & $0.0685$ \\
AVL & $147520$ & $221280$ & $1.50$ & $4$ & $1.27$ & $0.0685$ \\
AVL & $147520$ & $295040$ & $2.00$ & $4$ & $1.27$ & $0.0685$ \\
AVL & $147520$ & $590080$ & $4.00$ & $6$ & $0.006956$ & $0.00038$ \\
\bottomrule
\end{tabular}
\end{table}

\paragraph{Refinement trajectories.}
Table~\ref{app:evaluation:tab:trajectory-summary} reports every budget.
Write $v_L$ and $v_R$ for the children of the root.
For the pair-and-carry tree, the first step solves $v_L$, the eight-link left subtree ($256$ states), exactly; this lowers the cost to $114704$ and the bound to $42.76$.
With $B/C_0\ge1.25$, the second step raises the horizon within $v_L$ from $32$ to $64$ (cost $163856$, $1.11\,C_0$), and the bound falls to $5.845$.
The loss of abstracting the exact distribution of $v_L$, amplified by the root retry, dominates the bound at $H=32$, and the larger horizon removes most of it.
With $B/C_0\in\{1.5,2\}$, a third step raises the horizon within $v_R$, links $9$ and $10$, which lowers the bound by $3\times10^{-6}$.
With $B/C_0=4$, the third step raises the horizon within $v_L$ to $128$, and the bound falls to $0.110$.
For the AVL tree, the first two steps solve $v_R$ (six links, $64$ states) and $v_L$ (four links) exactly, and the next two raise the horizon within each to $64$; with $B/C_0=1.05$ (final cost $1.03\,C_0$), the bound falls from $18.53$ to $1.270$.
With $B/C_0=4$, two more steps raise both horizons to $128$, and the bound falls to $0.0070$.
Once both children of the root are solved exactly, $\delta_{\mathrm{root}}=(\epsilon_L+\epsilon_R)/a$ depends only on the losses of abstracting their exact distributions, which the horizon controls.

\paragraph{Beam search on the ten-link chain.}
The search of \Cref*{main:alg:search} runs on the same ten-link chain ($p=0.3$, $a=0.5$, homogeneous), starting from the pair-and-carry tree.
It analyzes candidates at $H_s=32$ with two-link blocks solved exactly (state limit $4$), with beam width $w=4$ and batch size $b=32$, and analyzes $T=100$ candidates besides the starting tree.
In each round, it rewrites every internal node of the trees in the beam, from the root down; the moves are rotations, changes of the split point, and rebuilds of intervals with at most $6$ links (\path{search_all_nodes.py}).
The final trees are reanalyzed at $H_f=128$ ($J=512$), and their exact means come from the Markov chain ($1024$ states).

\begin{table}[tb]
\centering
\scriptsize
\caption{Beam search on a ten-link repeater ($p=0.3$, $a=0.5$; original implementation).
Intervals are $\widehat K\pm\delta$ at $H=128$; the last column lists the trees whose intervals lie entirely above the interval of the tree in that row.
The median-split balanced tree is not used by the search and is shown for comparison.}
\label{app:evaluation:tab:tree-design-search-running}
\begin{adjustbox}{max width=\linewidth}
\begin{tabular}{llrrrl}
\toprule
Tree & Bracketing & Root split & Exact mean & Interval & Separated from \\
\midrule
Pair-and-carry & $((((1,2),(3,4)),((5,6),(7,8))),(9,10))$ & $8+2$ & $153.948$ & $153.948\pm0.111$ & -- \\
AVL & $(((1,2),(3,4)),((5,(6,7)),(8,(9,10))))$ & $4+6$ & $127.627$ & $127.627\pm0.007$ & pair-and-carry, balanced \\
Median-split balanced & $(((1,2),(3,(4,5))),((6,7),(8,(9,10))))$ & $5+5$ & $128.421$ & $128.421\pm0.001$ & pair-and-carry \\
Found & $(((1,2),(3,4)),(((5,6),(7,8)),(9,10)))$ & $4+6$ & $126.663$ & $126.663\pm0.010$ & all three \\
\bottomrule
\end{tabular}
\end{adjustbox}
\end{table}

The search reduces the mean from $153.948$ to $126.663$, by $17.72\%$, and the interval of the found tree is separated from those of the other trees: $126.663+0.010<127.627-0.007$.
To check optimality, \path{verify_tree_optimum.py} enumerates all $4862$ ordered bracketings of the ten links and computes their exact means by distribution propagation ($2^{14}$ grid points, FFT-based retry).
Four bracketings attain the smallest mean, $126.663$: the found tree and the three trees obtained from it by exchanging the children of nodes, which leaves the mean unchanged for homogeneous parameters, $((((1,2),(3,4)),(5,6)),((7,8),(9,10)))$, $(((1,2),((3,4),(5,6))),((7,8),(9,10)))$, and $(((1,2),(3,4)),((5,6),((7,8),(9,10))))$.
The next smallest mean, $127.627$, belongs to the AVL tree and its seven mirror images.

\paragraph{Large-chain search and summary reuse.}
The search suite uses $n\in\{32,64,128,256\}$, nominal $p=0.3$, $a=0.8$, and seeds $0,1,2$.
Here $p_i$ is uniform on $[0.3p,1.7p]$, widening heterogeneity relative to the scaling suite.
A separate seeded generator orders local moves.
For these sizes, the initial pair-and-carry tree is the balanced tree, the doubling scheme of Shchukin et al.~\cite{repeater-paper}.

The beam width is $w=4$, so the beam contains the four best estimated trees seen so far.
At each round, candidates are generated root-first by depth and path; local ties are shuffled using the fixed seed.
Moves are order-preserving rotations, balanced rebuilds, and split-point changes with rebuilt children.
The batch size is $b=64$, and the search analyzes up to $T=1024$ candidates besides the starting tree.
The run stops earlier if no unseen candidate remains in the current beam's neighborhood, and Table~\ref{app:evaluation:tab:search} reports the actual candidate counts.
Ranking uses the estimated mean, with stable ties, at $H_s=32$.

The fresh and reuse variants differ only in memo lifetime.
Both memoize within a candidate; only reuse retains entries across candidates.
Keys include the structural input and analysis parameters.
Both variants start with empty caches and generate candidates in the same deterministic order.
The raw records store every candidate shape, its mean and bound, a digest of the ordered trace, and summary-computation counts.
Validation compares the traces and outputs: all $12$ pairs of runs have identical traces, candidates, estimates, and bounds.
In both variants, the driver shares tree objects between candidates and computes the structural key of a tree object once; a rerun without these two caches gives identical traces and results on all $12$ chains.
The fresh runs have a limit of $300$\,s, and the longest takes $85$\,s.
At $256$ links, the fresh variant computes $523{,}775$ node summaries ($1025$ analyses of $511$ nodes) and the reuse variant $3{,}659$ to $3{,}923$; both variants have a peak resident set of $149$\,MiB.

\begin{table}[tb]
\centering
\scriptsize
\setlength{\tabcolsep}{3pt}
\caption{Beam search on large chains with fresh and reused summaries (\texttt{x8\_search\_pow2}; $p=0.3$, $a=0.8$, three chains per size).
Times are in seconds, medians with the range over the chains; speedup and saved time are medians of per-chain ratios and differences; the gain is the median reduction of the mean relative to the starting (balanced) tree, against truncated-propagation references.}
\label{app:evaluation:tab:search}
\begin{adjustbox}{max width=\linewidth}
\begin{tabular}{rrrrrrr}
\toprule
$n$ & Candidates & Fresh (s) & Reuse (s) & Speedup & Saved (s) & Gain vs. start (\%) \\
\midrule
32 & 348--758 & 6.93 [3.54,7.58] & 0.475 [0.327,0.526] & 14.4 & 6.46 & 8.25 \\
64 & 1024 & 20.8 [20.4,21.3] & 0.951 [0.928,0.968] & 22 & 19.9 & 8.52 \\
128 & 1024 & 41.3 [41.1,42.6] & 1.36 [1.33,1.36] & 30.4 & 40 & 5.36 \\
256 & 1024 & 82.5 [82,83.7] & 2.06 [1.99,2.07] & 40.6 & 80.4 & 1.95 \\
\bottomrule
\end{tabular}

\end{adjustbox}
\end{table}

After the search, the selected tree and the initial tree are separately reanalyzed at $H=256$ with fresh caches.
Their numerical mean reduction and interval separation are reported separately.
The condition
$\widehat K_{\mathrm{selected}}+\delta_{\mathrm{selected}}
 <\widehat K_{\mathrm{start}}-\delta_{\mathrm{start}}$
establishes the improvement over the starting tree.
Separate truncated-propagation jobs provide numerical references for the two final trees.
The baseline tree $e_b$ of \Cref*{main:alg:search} is the balanced tree, so the starting tree is the only baseline.
Table~\ref{app:evaluation:tab:search-baselines} reports the selected tree's numerical improvement and interval separation relative to it.
All $12$ selected trees have intervals separated from that of the balanced tree.
Numerical gains over the balanced tree range from $1.78\%$ to $10.26\%$; the median at $256$ links is $1.95\%$.

\paragraph{Search under a time budget.}
Since the fresh and the reuse variant analyze the same candidate sequence, the trace of each run gives the best estimated mean found as a function of the elapsed search time (Figure~\ref{app:evaluation:fig:search-budget}).
The reuse variant finishes its search loop after $0.41$, $0.82$, $1.10$, and $1.53$\,s at $32$, $64$, $128$, and $256$ links (medians, excluding the final checks at $H_f=256$), and its final trees lie $8.25\%$, $8.52\%$, $5.36\%$, and $1.95\%$ below the start.
By those times, the fresh variant has analyzed $41$, $39$, $26$, and $18$ candidates, and its best trees lie $1.7\%$, $1.5\%$, $0.5\%$, and $0\%$ below the start.
It first reaches the final tree of the reuse variant after $5.3$, $18.6$, $39.3$, and $77.4$\,s.
The $x$ axis of the figure is the wall-clock time of the search loop, including candidate generation and cache lookups; the analyzer accounts for $76\%$ to $96\%$ of it with reuse and for over $99\%$ without.

\begin{figure}[tb]
\centering
\includegraphics[width=\linewidth]{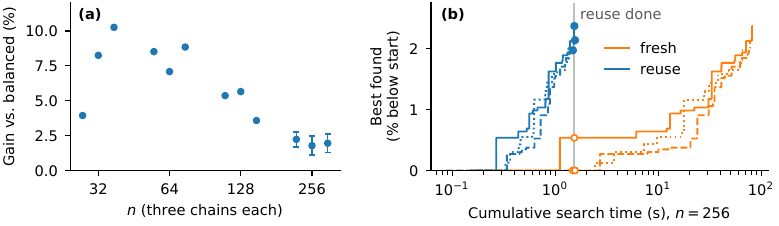}
\caption{Beam search on heterogeneous chains. (a)~Reduction of the mean completion time of the selected tree relative to the balanced tree, from the truncated-propagation references, for each of the three chains per size; the error bars are the intervals of \tool at $H_f=256$ and are narrower than the markers for $n\le128$. (b)~Best estimated mean found, as a percentage below the start tree, against the wall-clock search time for the three $256$-link chains (line styles), with and without reuse of the summaries of unchanged subtrees; the band marks the times at which the reuse runs finished, filled markers their final trees, and hollow markers the best tree of the fresh run at that time.}
\label{app:evaluation:fig:search-budget}
\end{figure}

\begin{table}[tb]
\centering
\scriptsize
\setlength{\tabcolsep}{3pt}
\caption{Search quality against the balanced tree, which is the starting tree (\texttt{x8\_search\_pow2}). Mean completion times use separate truncated-propagation references; ``Cert.\ $<$ bal.'' means that the intervals at $H=256$ are separated.}
\label{app:evaluation:tab:search-baselines}
\begin{adjustbox}{max width=\linewidth}
\begin{tabular}{rrrrrr}
\toprule
$n$ & Seed & Balanced $K$ & Selected $K$ & Gain vs. bal. (\%) & Cert. $<$ bal. \\
\midrule
32 & 0 & 51.8 & 49.8 & 3.94 & yes \\
32 & 1 & 87.6 & 80.4 & 8.25 & yes \\
32 & 2 & 80.2 & 72 & 10.3 & yes \\
64 & 0 & 107 & 98.1 & 8.52 & yes \\
64 & 1 & 133 & 124 & 7.09 & yes \\
64 & 2 & 113 & 103 & 8.84 & yes \\
128 & 0 & 190 & 180 & 5.36 & yes \\
128 & 1 & 219 & 206 & 5.65 & yes \\
128 & 2 & 210 & 202 & 3.58 & yes \\
256 & 0 & 343 & 335 & 2.23 & yes \\
256 & 1 & 360 & 354 & 1.78 & yes \\
256 & 2 & 358 & 352 & 1.95 & yes \\
\bottomrule
\end{tabular}

\end{adjustbox}
\end{table}

\subsection{Higher Moments}
\label{app:evaluation:subsec:exp-moment}

This subsection gives the configurations behind \Cref*{main:tab:design-search-and-moments}(b).
They compare the $M_1$-directed geometric tail, which preserves $M_1$ and records its second-moment bias, with the $M_2$-directed quadratic maximum-entropy tail, which matches both moments (\Cref*{main:sec:moment-tail}).
The variance interval is derived from the moment intervals as in \Cref*{main:sec:variance-query}.

\paragraph{Configurations behind \Cref*{main:tab:design-search-and-moments}(b)}
Table~\ref{app:evaluation:tab:moments-configs} lists the $12$ exactly solvable configurations behind \Cref*{main:tab:design-search-and-moments}(b), with four-link sub-chains solved exactly and exact references; the $H$-sweep rows of \Cref*{main:tab:design-search-and-moments}(b) for $H\in\{16,64\}$ coincide with the rows $n=8$, $p=0.1$ of this table.
In the configuration $n=4$, $p=0.37$, $H=64$, the quadratic fit returns the geometric boundary ($\theta_2=0$), so both abstractions use the same tail there.

\begin{table}[tb]
\centering
\scriptsize
\setlength{\tabcolsep}{4pt}
\caption{The $12$ exactly solvable configurations of the second-moment study (homogeneous balanced trees, $a=0.5$, four-link sub-chains solved exactly, exact references): errors and bounds in percent of the exact value under the $M_1$-directed and the $M_2$-directed abstraction, and the coverage of the $M_1$, $M_2$, and variance intervals, identical for both.}
\label{app:evaluation:tab:moments-configs}
\begin{tabular}{@{}rrr rr rr rr c@{}}
\toprule
& & & \multicolumn{2}{c}{$M_1$ err.\ (\%)} & \multicolumn{2}{c}{$M_2$ err.\ (\%)} & \multicolumn{2}{c}{$M_2$ bound (\%)} & Cov. \\
\cmidrule(lr){4-5}\cmidrule(lr){6-7}\cmidrule(lr){8-9}
$n$ & $p$ & $H$ & $M_1$-dir. & $M_2$-dir. & $M_1$-dir. & $M_2$-dir. & $M_1$-dir. & $M_2$-dir. & $M_1$/$M_2$/Var \\
\midrule
4 & 0.1 & 16 & $<10^{-6}$ & $<10^{-6}$ & 2.38 & $<10^{-6}$ & 2.38 & 0 & 1/1/1 \\
4 & 0.1 & 32 & $<10^{-6}$ & $<10^{-6}$ & 0.372 & $<10^{-6}$ & 0.372 & 0 & 1/1/1 \\
4 & 0.1 & 64 & $<10^{-6}$ & $<10^{-6}$ & 0.0101 & $<10^{-6}$ & 0.0101 & 0 & 1/1/1 \\
4 & 0.37 & 16 & $<10^{-6}$ & $<10^{-6}$ & 0.00685 & $<10^{-6}$ & 0.00685 & 0 & 1/1/1 \\
4 & 0.37 & 32 & $<10^{-6}$ & $<10^{-6}$ & $3.84\!\times\!10^{-6}$ & $<10^{-6}$ & $3.84\!\times\!10^{-6}$ & 0 & 1/1/1 \\
4 & 0.37 & 64 & 0 & 0 & $<10^{-6}$ & $<10^{-6}$ & $<10^{-6}$ & $<10^{-6}$ & 1/1/1 \\
8 & 0.1 & 16 & 0.702 & 0.293 & 10.5 & 0.321 & 305 & 292 & 1/1/1 \\
8 & 0.1 & 32 & 0.0903 & 0.0392 & 4.65 & 0.0463 & 120 & 114 & 1/1/1 \\
8 & 0.1 & 64 & 0.00157 & $6.89\!\times\!10^{-4}$ & 1.14 & $9.23\!\times\!10^{-4}$ & 24.3 & 22.2 & 1/1/1 \\
8 & 0.37 & 16 & 0.00106 & $4.66\!\times\!10^{-4}$ & 1.07 & $6.27\!\times\!10^{-4}$ & 23.5 & 21.4 & 1/1/1 \\
8 & 0.37 & 32 & $<10^{-6}$ & $<10^{-6}$ & 0.077 & $<10^{-6}$ & 1.23 & 0.859 & 1/1/1 \\
8 & 0.37 & 64 & $<10^{-6}$ & $2.76\!\times\!10^{-5}$ & $4.19\!\times\!10^{-4}$ & $1.05\!\times\!10^{-4}$ & 0.0122 & 0.0011 & 1/1/1 \\
\bottomrule
\end{tabular}

\end{table}

\subsection{RFID Collision Resolution and Fork--Join Tasks}
\label{app:evaluation:subsec:exp-case-studies}
\label{app:evaluation:subsec:exp-transfer}

\paragraph{RFID accuracy.}
The model is the concurrent binary query tree of \Cref{app:method:sec:instantiations}, with uniform identifiers and decoding probability $a\in\{1,0.8\}$, for $n\in\{8,16,32,64\}$ tags.
The classical baseline is $\E[H_n]+1/a$, the mean trie height plus the expected number of attempts of one reply, which is exact for $a=1$; the scalar baseline replaces $\E[\max(X,Y)]$ by $\max(\E[X],\E[Y])$ at each productive split.
\tool runs at $H=16$ with $J=2000$ and solves populations of at most $4$ tags exactly.
The exact means come from a dynamic program over the cost distributions on the grid $\{0,\ldots,2000\}$; the mass beyond the grid is below $3\times10^{-15}$ for these populations.
For $a=1$, its distributions agree with the closed form $\Pr(T(n)\le k+1)=\prod_{i=1}^{n-1}(1-i/2^k)$ to within $4\times10^{-15}$, and so does the prefix of the root summary of \tool, to within $2\times10^{-15}$ for $H\in\{8,16,32\}$; for both values of $a$, a simulation of the query tree with $200{,}000$ runs per population agrees with its means to within $2.01$ standard errors.
Table~\ref{app:evaluation:tab:rfid-results} repeats \Cref*{main:tab:cases}(b) and adds the rows for $a=1$; every interval contains the exact mean.

\begin{table}[tb]
\centering
\scriptsize
\setlength{\tabcolsep}{3pt}
\caption{RFID collision resolution over a channel with decoding probability $a$ (original implementation).
Errors and bounds $100\,\delta/K$ are relative to the exact mean $K$; the classical baseline is $\E[H_n]+1/a$.}
\label{app:evaluation:tab:rfid-results}
\begin{adjustbox}{max width=\linewidth}
\begin{tabular}{rr r rr r r}
\toprule
& & & \multicolumn{2}{c}{Baseline error} & \multicolumn{2}{c}{\tool, $H=16$} \\
\cmidrule(lr){4-5}\cmidrule(l){6-7}
$a$ & $n$ & Exact $K$ & Classical & Scalar & Rel.\ error & Rel.\ bound \\
\midrule
1 & 8 & 7.2694 & exact & 10.4 & $1.9\!\times\!10^{-12}$ & $1.4\!\times\!10^{-7}$ \\
1 & 16 & 9.3018 & exact & 15.2 & $2.3\!\times\!10^{-10}$ & $2.1\!\times\!10^{-6}$ \\
1 & 32 & 11.3175 & exact & 19.4 & $1.7\!\times\!10^{-8}$ & $2.8\!\times\!10^{-5}$ \\
1 & 64 & 13.3252 & exact & 22.9 & $1.1\!\times\!10^{-6}$ & 0.000388 \\
\midrule
0.8 & 8 & 7.8555 & 4.28 & 13.9 & $1.3\!\times\!10^{-11}$ & $5.6\!\times\!10^{-7}$ \\
0.8 & 16 & 9.9325 & 3.83 & 18.1 & $1.2\!\times\!10^{-9}$ & $6.8\!\times\!10^{-6}$ \\
0.8 & 32 & 11.9765 & 3.42 & 21.8 & $8.0\!\times\!10^{-8}$ & $8.2\!\times\!10^{-5}$ \\
0.8 & 64 & 14.0014 & 3.04 & 24.8 & $4.7\!\times\!10^{-6}$ & 0.00105 \\
\bottomrule
\end{tabular}
\end{adjustbox}
\end{table}

\paragraph{Fork--join deadline queries.}
The suite has $6$ random series--parallel tasks (seeds $0$ to $5$) with $16$ leaves each.
An internal node splits its leaves at a uniformly random position and chooses $+$, $\max$, or $\Mix$ with probabilities $0.20$, $0.75$, and $0.05$, with mixture weights uniform on $[0.3,0.7]$; the tasks contain no $\min$.
Every leaf has the time $V$ of \Cref*{main:sec:exp-setup}: $1$ with probability $0.99$, and otherwise $31+\Geom_{\ge1}(0.055)$.
Exact references propagate the full distribution, doubling the truncation point from $512$ until the omitted mass is below $10^{-12}$.
For each task, the analysis chooses $H$ without the reference, doubling it from $32$ until the width of its own mean interval is below $1\%$ of the estimate; every task stops at $H=64$ ($J=256$), where the prefix holds $97.3\%$ to $98.5\%$ of the mass.
The query is the miss probability $\Pr(X>D)$ for deadlines $D$ at the $0.90$, $0.99$, $0.999$, and $0.9999$ quantiles of the exact response-time distribution.
For $D\le H$, the analysis reads $\Pr(X>D)$ from the exact prefix (\Cref*{main:lem:exact-prefix}); for $D>H$, it returns the interval $S_s(D)\pm\epsilon$ of \Cref*{main:sec:error-bounded-semantics}, with $\epsilon$ the shape bound of the root, clipped to $[0,1]$.
Every interval contains the exact miss probability.

\begin{table}[tb]
\centering
\scriptsize
\setlength{\tabcolsep}{3pt}
\caption{Deadline-miss query $\Pr(X>D)$ on six fork--join tasks with $16$ leaves (original implementation).
``Miss prob.''\ and ``Width'' (the full width of the interval of \tool) are medians over the tasks; ``Upper / exact'' gives the range over the tasks of the upper end of the interval divided by the exact miss probability. The interval uses the shape bound $\epsilon$ (\Cref*{main:sec:error-bounded-semantics}) or the Kolmogorov bound $\epsilon_K$ (\Cref{app:extension:sec:threshold-extension}).}
\label{app:evaluation:tab:forkjoin-deadline}
\begin{adjustbox}{max width=\linewidth}
\begin{tabular}{@{}lrlr rr rr@{}}
\toprule
& & & & \multicolumn{2}{c}{$S_s(D)\pm\epsilon$} & \multicolumn{2}{c@{}}{$S_s(D)\pm\epsilon_K$} \\
\cmidrule(lr){5-6}\cmidrule(l){7-8}
Quantile & $D$ & Location & Miss prob. & Width & Upper / exact & Width & Upper / exact \\
\midrule
$0.90$   & $2$--$40$    & prefix & $9.70\times10^{-2}$ & $0$ & $1$ & $0$ & $1$ \\
$0.99$   & $72$--$83$   & tail & $9.74\times10^{-3}$ & $2.85\times10^{-3}$ & $1.05$--$1.27$ & $6.84\times10^{-5}$ & $1.000$--$1.002$ \\
$0.999$  & $114$--$126$ & tail & $9.80\times10^{-4}$ & $2.40\times10^{-3}$ & $1.48$--$3.75$ & $6.84\times10^{-5}$ & $1.02$--$1.08$ \\
$0.9999$ & $156$--$169$ & tail & $9.77\times10^{-5}$ & $1.53\times10^{-3}$ & $5.80$--$28.8$ & $6.84\times10^{-5}$ & $1.13$--$1.72$ \\
\bottomrule
\end{tabular}
\end{adjustbox}
\end{table}

Table~\ref{app:evaluation:tab:forkjoin-deadline} reports the results.
At the $0.90$ quantile, every deadline lies inside the prefix, and the analysis returns the exact miss probability (absolute error below $10^{-12}$).
At the tail deadlines, every interval of \tool contains the exact miss probability.
For comparison, the Markov inequality applied to the upper end of the mean interval of the interval analysis gives median upper bounds between $0.13$ and $0.26$ at the tail deadlines and $0.58$ at the $0.90$ quantile (\Cref*{main:tab:cases}(a)).
The last two columns replace $\epsilon$ by the Kolmogorov bound $\epsilon_K$ of \Cref{app:extension:sec:threshold-extension}, computed from the same summaries (\path{x6_forkjoin_kolmogorov}).
On every task, $\epsilon_K$ is $37$ to $46$ times smaller than $\epsilon$, and every interval $S_s(D)\pm\epsilon_K$ contains the exact miss probability.
At the $0.9999$ quantile, its upper end is at most $1.72$ times the exact value, and the median Markov bound is $921$ times larger.

\paragraph{Mean and variance of the response time.}
Tasks from the same generator, with $8$ to $64$ leaves and eight seeds per size, are also analyzed for the mean and the variance of the response time, with $H$ chosen by the doubling rule above (Table~\ref{app:evaluation:tab:forkjoin-moments}); it stops at $H=32$ for some tasks with $8$ leaves and at $H=64$ otherwise.
This rule asks for $2\delta<0.01\,\widehat K$, which is stricter than the $1\%$ target of the adaptive mode of \Cref*{main:sec:exp-chains}.
The comparison is an interval analysis that propagates intervals of the first two moments through the same operators (for $\max$, the interval between the larger child mean and the sum of the child means); its intervals are sound and $10^3$ to $10^4$ times wider than those of \tool.
Every interval of \tool contains the exact mean and variance.

\begin{table}[tb]
\centering
\scriptsize
\caption{
Mean and variance of the fork--join response time.
The columns give the horizon $H$;
the widths of the intervals of \tool{} and of the moment-interval analysis, as percentages of the exact value;
and the coverage of the intervals of \tool{}, i.e., the number of tasks whose \tool{} interval contains the exact mean or variance, out of the number of tasks.
$H$ and the widths are medians over the tasks.
}
\label{app:evaluation:tab:forkjoin-moments}
\begin{tabular}{@{}rrr rr rr rr@{}}
\toprule
& & & \multicolumn{2}{c}{Mean width (\%)} & \multicolumn{2}{c}{Variance width (\%)} & \multicolumn{2}{c@{}}{Coverage} \\
\cmidrule(lr){4-5}\cmidrule(lr){6-7}\cmidrule(l){8-9}
Leaves & Tasks & $H$ & \tool & interval & \tool & interval & mean & var.\ \\
\midrule
8 & 8 & 64 & 0.0134 & 139 & 0.134 & 541 & 8/8 & 8/8 \\
16 & 8 & 64 & 0.0153 & 161 & 0.299 & 969 & 8/8 & 8/8 \\
32 & 8 & 64 & 0.0455 & 226 & 0.656 & 2811 & 8/8 & 8/8 \\
64 & 8 & 64 & 0.114 & 287 & 1.38 & 6456 & 8/8 & 8/8 \\
\bottomrule
\end{tabular}

\end{table}

\section{Extensions}
\label{app:extension:sec:extensions}

The soundness results of the paper cover finite cost expressions with independent operands.
This appendix develops the extensions of \Cref*{main:sec:discussion} (unbounded generators, operators without prefix completeness, and dependence) and a sharper bound for the deadline queries of \Cref*{main:sec:error-bounded-semantics}.

\subsection{Unbounded Generators}
\label{app:extension:sec:unbounded-extension}

The generators of the paper produce one finite tree, and the analysis proceeds bottom-up over it.
Some protocol specifications recurse without a decreasing parameter, so their unfolding need not terminate and there is no finite bottom-up order.
We call such a program an \emph{unbounded generator}.
Like the generators of the paper, it is deterministic and expresses random choices by $\Mix$ nodes; the object of interest is the cost distribution of its unfolding, defined as a fixed point.

\paragraph{Contexts.}
An unbounded generator has a finite set $C$ of contexts.
A context $c\in C$ is a named recursion point that stands for the tree generated from $c$.
Like a nonterminal of a probabilistic grammar, it may be called in the bodies of contexts, including its own.
Each context is a probabilistic choice $\Mix_{w_c}$ among finitely many bodies $e_{c,j}$, where $w_{c,j}\ge0$ is the probability of body $e_{c,j}$ and $\sum_jw_{c,j}=1$.
A body is a cost expression built from the operators and atoms of the paper and from recursive calls to contexts; a call to $c'$ stands for an independent instance of the tree generated from $c'$.

\paragraph{Fixed-point semantics.}
Let $X_c$ be the cost distribution of the tree generated from context $c$, and $\vec X=(X_c)_{c\in C}$.
Each body induces a transformer $\sem{e_{c,j}}(\vec X)$, evaluated with the operators' concrete semantics and with each call to $c'$ replaced by $X_{c'}$.
Combining the bodies of $c$ by their probabilities gives $F_c(\vec X)=\sum_jw_{c,j}\sem{e_{c,j}}(\vec X)$, and the semantics of the generator is the solution of the system $\vec X=\vec F(\vec X)$, one equation per context, which is unique among vectors of distributions with finite means under the condition below.
It is in general not the cost of any single finite tree.

\paragraph{Abstract iteration and error bound.}
The analysis would iterate the abstract transformer on a vector of summaries, $\vec s^{(n+1)}=\beta_H(\vec F(\vec s^{(n)}))$, applied componentwise and starting from the summaries of the zero cost.
A limit $\vec s^{\ast}$ is a vector of summaries reproduced by one abstract step.
It does not by itself carry an error bound; that requires controlling how error introduced at one context propagates around the recursion.
Each abstract step introduces a loss $\eta_c$ at each context; let $\vec\eta=(\eta_c)_c$.
The operator contracts bound the propagation through one unfolding by the error-flow matrix of \Cref*{main:sec:extensions}, a nonnegative matrix $\mathbf L$ whose entry $\mathbf L_{c,c'}$ sums, over the bodies of $c$ weighted by $w_{c,j}$, the path products of the coefficients $L_i$ from the root of each body to its calls to $c'$.
If the spectral radius of $\mathbf L$ is below $1$, the vector $\vec\epsilon=(\epsilon_c)_c$ of shape errors $\epsilon_c=\dW(X_c,s^\ast_c)$ satisfies
\[
  \vec\epsilon\le(I-\mathbf L)^{-1}\vec\eta,
\]
where $(I-\mathbf L)^{-1}=\sum_{n\ge0}\mathbf L^n$ is finite and nonnegative because the spectral radius is below $1$.
Indeed, the contracts and the triangle inequality give $\vec\epsilon\le\mathbf L\vec\epsilon+\vec\eta$ componentwise, so $\vec\epsilon\le\sum_{n<N}\mathbf L^n\vec\eta+\mathbf L^N\vec\epsilon$ for every $N$, and $\mathbf L^N\vec\epsilon\to0$ because $\vec\epsilon$ is finite.
The same condition keeps the means finite: (DistErr) with the zero cost as second argument gives $\E[F_c(\vec X)]\le\E[F_c(\vec0)]+\sum_{c'}\mathbf L_{c,c'}\E[X_{c'}]$, so the means of the iterates from the zero cost stay below $(I-\mathbf L)^{-1}\E[\vec F(\vec0)]$, and $\vec F$ is a contraction on vectors with finite means.
Higher moments need a vector-valued system and stronger moment conditions.

\paragraph{A branching example.}
Consider one context that stops with probability $1-w$ and otherwise spawns two independent subproblems:
\[
  G=\Mix_{(1-w,\,w)}\bigl(\Atom(0),\ \Atom(1)+G_1+G_2\bigr),
\]
where $G_1$ and $G_2$ are independent instances of $G$ and $\Atom(1)$ is the local cost of a branching step.
The matrix $\mathbf L$ is the $1\times1$ matrix $2w$: each of the two recursive calls contributes the constant $1$ of the sum, weighted by $w$.
The condition $2w<1$ is exactly the subcriticality condition of the branching process, and under it the expected cost is finite.
The matrix condition generalizes this to mutually recursive generators and arbitrary operator libraries.

\subsection{Deadline and Tail-Probability Queries}
\label{app:extension:sec:threshold-extension}

\Cref*{main:sec:error-bounded-semantics} bounds a deadline query by the interval $S_s(D)\pm\epsilon$, using the shape bound $\epsilon$.
This subsection shows that the Kolmogorov distance gives a bound $\epsilon_K$ that is never larger; on the fork--join tasks of \Cref*{main:sec:exp-other}, it is about $40$ times smaller.

\paragraph{Query and distance.}
For a deadline $D\in\N$, let $Q_D(X)=\Pr(X>D)=S_X(D)$.
The matching distance is the Kolmogorov (uniform survival) distance $d_K(X,Y)=\sup_{t\ge0}|S_X(t)-S_Y(t)|$, which controls every deadline query, $|Q_D(X)-Q_D(Y)|\le d_K(X,Y)$, with equality for the worst deadline.
Since $d_K\le\dW$, a Wasserstein bound also bounds deadline queries, but $d_K$ is sharper; conversely, $d_K$ does not control the mean, because a small uniform error spread over a long tail can have a large survival sum.

\paragraph{Answers from summaries.}
For a geometric summary $s=(f_0,\ldots,f_H,\lambda)$, the estimate is $S_s(D)=1-\sum_{u\le D}f_u$ if $D\le H$, and $\rho\,\lambda^{D-H}$ if $D>H$; for a general tail family it is $\rho\Pr_{\tau_\theta}(R\ge D-H)$ when $D>H$.
This is a point query on the survival function, cheaper than a mean.
When all operators satisfy (PComplete), \Cref*{main:lem:exact-prefix} gives $S_s(D)=Q_D(\sem{e})$ exactly for every $D\le H$, without constructing the full distribution; for $D>H$, the error comes only from the shape of the residual tail.

\paragraph{Contracts under $d_K$.}
The built-in operators satisfy contracts under $d_K$ with the same constants as for $\dW$:
$d_K(\Mix_w(\vec X),\Mix_w(\vec Y))\le\sum_iw_i\,d_K(X_i,Y_i)$;
$d_K(\max(X,Z),\max(Y,Z))\le d_K(X,Y)$ and $d_K(\min(X,Z),\min(Y,Z))\le d_K(X,Y)$;
$d_K(X+Z,Y+Z)\le d_K(X,Y)$ for $Z$ independent of the operands; and
$d_K(\RepeatOp_N(X),\RepeatOp_N(Y))\le\E[N]\,d_K(X,Y)$.
The proofs are pointwise calculations on survival functions; for instance, $S_{X+Z}(t)=\sum_z\Pr(Z=z)S_X(t-z)$ averages shifted survival functions and cannot increase their supremum distance, and $\RepeatOp_N$ follows by iterating the sum bound for a fixed count and averaging over $N$.

\paragraph{Abstraction loss under $d_K$.}
Because abstraction keeps the prefix and the tail mass, the survival gap $|S_X(t)-S_s(t)|$ with $s=\beta_H(X)$ is zero at all $t\le H$, and for the geometric tail the loss $d_K(X,s)$ equals $\rho\sup_{x\ge0}|S_{R_X}(x)-\lambda^{x+1}|$, where $R_X$ is the residual tail.
For any $J\ge H$, the bound
\[
  \max\Bigl(\max_{H<t\le J}|S_X(t)-S_s(t)|,\ S_X(J+1),\ S_s(J+1)\Bigr)
\]
is sound because survival functions are nonincreasing: beyond $J$, the gap is at most the larger of the two remaining survival values.
The judgment $e\Down_K(s,\epsilon_K)$ asserts $d_K(\sem{e},s)\le\epsilon_K$, and its soundness proof is the induction of \Cref*{main:thm:generic-soundness} with the $d_K$ contracts; at the root, $|Q_D(\sem{e})-S_s(D)|\le\epsilon_K$ for every $D$, and the error is $0$ for $D\le H$.
Each local $d_K$ bound is at most the corresponding $\dW$ bound (\ref*{main:eq:loss-bound}), and the contracts have the same constants, so $\epsilon_K\le\epsilon$: the interval $S_s(D)\pm\epsilon_K$ lies inside $S_s(D)\pm\epsilon$.

\paragraph{Which operators are hard depends on the query.}
For the mean, sums and repetitions compose, and $\max$ and $\min$ need shape information.
For a fixed deadline, $\max$, $\min$, and mixtures compose pointwise: $1-S_{\max(X,Y)}(D)=(1-S_X(D))(1-S_Y(D))$, $S_{\min(X,Y)}(D)=S_X(D)S_Y(D)$, and $S_{\Mix_w(\vec X)}(D)=\sum_iw_iS_{X_i}(D)$; but $S_{X+Y}(D)$ needs the whole prefix up to $D$, and so does random repetition.
A prefix--tail summary supports all three queries (\Cref{app:extension:tab:queries}).

\begin{table}[tb]
\centering
\small
\caption{The distance that matches each query, and the operators that need shape information for it.}
\label{app:extension:tab:queries}
\begin{tabular}{lll}
\toprule
Query & Matching distance & Operators that need shape information \\
\midrule
Mean $M_1$ & $\dW=d_1$ & $\max$, $\min$ \\
Raw moment $M_j$ & $d_j$ & $\max$, $\min$; cross-order terms for $+$ and $\RepeatOp$ \\
Deadline $Q_D$ & $d_K$ & $+$, $\RepeatOp$ \\
\bottomrule
\end{tabular}
\end{table}

Table~\ref{app:evaluation:tab:forkjoin-deadline} compares the intervals $S_s(D)\pm\epsilon$ and $S_s(D)\pm\epsilon_K$ on the fork--join tasks of \Cref*{main:sec:exp-other}.

\subsection{Operators Without Prefix Completeness, and Dependence}
\label{app:extension:sec:kernel-limitations}

\paragraph{Operators without prefix completeness.}
\Cref*{main:sec:extensions} treats operators with an action on distributions and error contracts but no finite-prefix transformer, with the idle time $|X-Y|$ as its example: the bounds remain sound, the concretization loses its prefix condition, and the stored prefix is the surrogate's.
Another example is $(X-c)^+$ for a constant $c$, which moves tail mass back to small costs, so its prefix depends on the tail of its operand.

For $|X-Y|$, let the child summaries be $s_X=(f_0,\ldots,f_H,\lambda_X)$ and $s_Y=(g_0,\ldots,g_H,\lambda_Y)$, with tail masses $\rho_X$ and $\rho_Y$.
The masses of $\widetilde Z=|s_X-s_Y|$ are $\Pr(\widetilde Z=k)=\sum_{t\ge0}\bigl(\Pr(s_X=\nobreak t)\Pr(s_Y=\nobreak{t+k})+\Pr(s_Y=\nobreak t)\Pr(s_X=\nobreak{t+k})\bigr)$ for $k\ge1$, and the first summand alone for $k=0$.
The terms with $t\le H$ form finite sums.
For $t>H$, both factors lie in geometric tails, and $\sum_{t>H}\Pr(s_X=\nobreak t)\Pr(s_Y=\nobreak{t+k})=\rho_X\rho_Y(1-\lambda_X)(1-\lambda_Y)\,\lambda_Y^k/(1-\lambda_X\lambda_Y)$, and symmetrically with $\lambda_X^k$.
By the tail-sum formula, $\E[\min(s_X,s_Y)]=\sum_{t<H}S_{s_X}(t)\,S_{s_Y}(t)+\rho_X\rho_Y/(1-\lambda_X\lambda_Y)$, which gives the mean $\E[\widetilde Z]=\E[s_X]+\E[s_Y]-2\,\E[\min(s_X,s_Y)]$ used in \Cref*{main:sec:extensions}.

\paragraph{Dependence and shared randomness.}
The summaries assume that siblings and repeated attempts are independent (\Cref*{main:sec:independence}).
If two child costs are correlated through shared randomness, memory constraints, or a common environment, a pair of summaries does not determine $\max(X,Y)$ or $X+Y$, and joint summaries, coupling information, or dependence annotations are needed.
Several dependent costs of one system raise the same issue; the completion time and the fidelity of a repeater are an example~\cite{brand2020waitingtime}.

\bibliographystyleApp{ACM-Reference-Format}
\bibliographyApp{src/main}

\end{document}